\documentclass[11pt]{amsart}
\usepackage[margin=1in]{geometry}
\usepackage{amssymb,xcolor,colortbl,diagbox,comment,graphicx,mathpazo,cite}
\usepackage{MnSymbol}
\usepackage{color,float,fullpage,mathrsfs,mathtools}
\usepackage{enumerate,enumitem}
\usepackage[colorlinks=true, pdfstartview=FitV, linkcolor=blue, citecolor=red, urlcolor=blue]{hyperref}
\usepackage{multirow,array}
\usepackage{makecell}
\usepackage{epstopdf}
\usepackage{underscore}

\newtheorem{proposition}{Proposition}
\newtheorem{theorem}{Theorem}
\newtheorem{corollary}{Corollary}
\newtheorem{lemma}{Lemma}
\theoremstyle{definition}

\newtheorem{remark}{Remark}


\title[Defocusing-defocusing coupled Hirota equations]{Inverse scattering transform for the defocusing-defocusing coupled Hirota equations with non-parallel boundary conditions at infinity}

\author{Peng-Fei Han$^{1}$}
\author{Wen-Xiu Ma$^{1,2,3,4}$}
\author{Yi Zhang$^{1}$}

\thanks{Corresponding authors. E-mail address: mawx@cas.usf.edu (Wen-Xiu Ma), hanpf1995@163.com (Peng-Fei Han), zy2836@163.com (Yi Zhang)}

\dedicatory{$^{1}$Department of Mathematics, Zhejiang Normal University, Jinhua 321004, People's Republic of China. \\
$^{2}$Department of Mathematics, King Abdulaziz University, Jeddah 21589, Saudi Arabia. \\
$^{3}$Department of Mathematics and Statistics, University of South Florida, Tampa, FL 33620-5700, USA. \\
$^{4}$Material Science Innovation and Modelling, Department of Mathematical Sciences, North-West University, Mafikeng Campus, Mmabatho 2735, South Africa.}


\date{\today}

\begin{document}

\begin{abstract}
The inverse scattering transform for the defocusing-defocusing coupled Hirota equations is strictly discussed with non-zero boundary conditions at infinity including non-parallel boundary conditions, specifically referring to the asymptotic polarization vectors. To address the non-analyticity encountered in some of the Jost eigenfunctions, the "adjoint" Lax pair is employed. The inverse problem is formulated as an appropriate matrix Riemann-Hilbert problem. A key difference between non-parallel and parallel boundary conditions lies in the asymptotic behavior of the scattering coefficients, which significantly impacts the normalization of the eigenfunctions and the properties of sectionally meromorphic matrices within the Riemann-Hilbert problem framework. When the asymptotic polarization vectors are non-orthogonal, two distinct methodologies are introduced to convert the Riemann-Hilbert problem into a series of linear algebraic-integral equations. In contrast, when the asymptotic polarization vectors are orthogonal, only one method is feasible. Ultimately, it is demonstrated that pure soliton solutions do not exist in both orthogonal and non orthogonal polarization vector cases. This study provides a comprehensive framework for analyzing the defocusing-defocusing coupled Hirota equations using the inverse scattering transform, offering new insights into the characteristics and solutions of the equations.
\end{abstract}
\maketitle

\tableofcontents

\section{Introduction}
\label{s:intro}

In the field of nonlinear physics, complex interactions between multiple fields often lead to complex dynamic behavior~\cite{1,2,3,4}. The dynamics of multi-component nonlinear physical systems are governed by coupled nonlinear Schr\"odinger (NLS) type evolution equations, a topic that has been discussed in~\cite{7}. These equations are pivotal in understanding the complex interactions within such systems. The Manakov system~\cite{8} is completely integrable and has been extensively researched within the scientific community~\cite{9,10,11,12,13}. The general coupled Hirota equations~\cite{14} are used as extensions of the Manakov system, introducing additional terms to consider higher-order nonlinearity and dispersion effects, making them particularly important for analyzing systems with rich nonlinear characteristics~\cite{15,16}. Being able to capture the nonlinear interactions between different components of the system makes it a valuable tool for studying phenomena such as optical solitons~\cite{17}, plasma waves~\cite{18}, and fluid dynamics~\cite{19}. For example, the general coupled Hirota equations can generally describe the propagation of optical pulses with different degrees of nonlinearity and dispersion, which is crucial for understanding the formation and stability of soliton solutions~\cite{20,21}. Consequently, the initial value problem (IVP) can theoretically be addressed using the inverse scattering transform (IST), a method that has garnered significant interest from the scientific community for its application to the general coupled Hirota equations~\cite{22}.

Studying the general coupled Hirota equations with non-zero boundary conditions (NZBCs) presents significant challenges as it requires a deeper understanding of the mathematical structure behind them and the development of advanced analytical and numerical techniques~\cite{23}. The emergence of NZBCs in the application of the IST introduces additional complexity and requires more refined methods~\cite{24,A1,A2}. The recent analysis of the general coupled Hirota equations with NZBCs has been driven by the need to understand the behavior of the system in more realistic settings~\cite{23}. For example, in the study of Bose-Einstein condensates~\cite{25,26}, the presence of finite background density may lead to the emergence of novel excitations and instabilities that have not been captured in traditional Manakov system. The general coupled Hirota equations provide a more comprehensive framework for simulating such scenarios due to their ability to expand~\cite{27}. The exploration of IVPs for the general coupled Hirota equations, particularly under NZBCs, has sparked the development of innovative mathematical approaches and enhanced our comprehension of the underlying physics~\cite{28}. The exploration of such systems continues to reveal the rich and diverse behaviors of multi-component nonlinear systems, further consolidating their importance in the field of nonlinear physics.

By utilizing the zero-curvature conditions, the construction of multi-component Hirota equations~\cite{32,33} becomes crucial, enabling researchers to delve into the complex interactions between different wave components. From the perspective of matrix eigenvalue problems~\cite{34,35}, the IST has become a powerful tool for solving IVPs related to these equations. This method not only helps to derive soliton solutions, but also provides an effective way to reveal many Lie symmetries and their corresponding conservation laws of the equations. In addition, applying matrix constraints to eigenvalue problems can lead to a large number of simplified but still integrable equations. These simplifications provide a new perspective for studying the local and nonlocal aspects of equations while maintaining the invariance of the zero-curvature conditions. Especially for nonlocal integrable equations, their correlation in non-Hermitian quantum mechanics and their ability to simulate nonlinear mathematical and physical phenomena have aroused great interest~\cite{38}. By exploring the multi-component Hirota equations through the Ablowitz-Kaup-Newell-Segur hierarchical structure, its powerful ability in elucidating the behavior of nonlinear systems has been demonstrated. Whether through local or nonlocal group reduction, these equations reveal rich information about the structure and properties of integrable systems, making them important tools for studying nonlinear phenomena~\cite{40}.

The general coupled Hirota equations~\cite{14,15,16,20,21,22,27,28,41,43} offer a comprehensive model for studying the wave propagation of two ultrashort optical fields in optical fibers, accounting for the intricate interplay of nonlinear and dispersive effects that govern the evolution of the pulses as they travel through the fiber.
\begin{subequations}\label{1.1}
\begin{align}
\mathrm{i}\widetilde{q}_{1,t}+& \widetilde{q}_{1,xx}+2(\sigma_{1}\left|\widetilde{q}_{1}\right|^{2}+\sigma_{2}\left|\widetilde{q}_{2}\right|^{2})
\widetilde{q}_{1}+\mathrm{i}\sigma[\widetilde{q}_{1,xxx}+(6\sigma_{1}\left|\widetilde{q}_{1}\right|^{2}
+3\sigma_{2}\left|\widetilde{q}_{2}\right|^{2})\widetilde{q}_{1,x} +3\sigma_{2}\widetilde{q}_{1}\widetilde{q}_{2}^{*}\widetilde{q}_{2,x}]=0, \\
\mathrm{i}\widetilde{q}_{2,t}+& \widetilde{q}_{2,xx}+2(\sigma_{1}\left|\widetilde{q}_{1}\right|^{2}+\sigma_{2}\left|\widetilde{q}_{2}\right|^{2})
\widetilde{q}_{2}+\mathrm{i}\sigma[\widetilde{q}_{2,xxx}+(6\sigma_{2}\left|\widetilde{q}_{2}\right|^{2}
+3\sigma_{1}\left|\widetilde{q}_{1}\right|^{2})\widetilde{q}_{2,x}
+3\sigma_{1}\widetilde{q}_{2}\widetilde{q}_{1}^{*}\widetilde{q}_{1,x}]=0,
\end{align}
\end{subequations}
where $\widetilde{q}_{1}=\widetilde{q}_{1}(x,t)$ and $\widetilde{q}_{2}=\widetilde{q}_{2}(x,t)$ are the two-component electric fields, the parameters $\sigma_{1}$, $\sigma_{2}$ and $\sigma$ are real constants. The above systems~\eqref{1.1} provide a linear combination of special reductions of the NLS and mKdV equations, which are introduced and discussed in (5.13) and (5.14) in~\cite{34}. A matrix form of these mKdV equations can be found in~\cite{44}. In present work, we consider a type of $\sigma_{1}=\sigma_{2}=-1$ in the general coupled Hirota equations~\cite{41}, namely the defocusing-defocusing coupled Hirota equations~\cite{23}
\begin{equation}\label{1.2}
\begin{split}
\mathrm{i}\widetilde{\mathbf{q}}_{t}+\widetilde{\mathbf{q}}_{xx}-2\Vert \widetilde{\mathbf{q}} \Vert^{2}\widetilde{\mathbf{q}}+\mathrm{i}\sigma[\widetilde{\mathbf{q}}_{xxx}-3\Vert \widetilde{\mathbf{q}} \Vert^{2}\widetilde{\mathbf{q}}_{x}
-3(\widetilde{\mathbf{q}}^{\dagger}\widetilde{\mathbf{q}}_{x})\widetilde{\mathbf{q}}]=\mathbf{0},
\end{split}
\end{equation}
where $\widetilde{\mathbf{q}}=\widetilde{\mathbf{q}}(x,t)=(\widetilde{q}_{1}(x,t),
\widetilde{q}_{2}(x,t))^{T}$ and $\dagger$ denotes conjugate transpose. Applying the simple space-independent phase rotation $\widetilde{\mathbf{q}}(x,t)=\mathbf{q}(x,t)\mathrm{e}^{-2\mathrm{i}q_{0}^{2}t}$ transforms Eq.~\eqref{1.2} into the defocusing-defocusing coupled Hirota equations with NZBCs, namely
\begin{equation}\label{1.3}
\begin{split}
\mathrm{i}\mathbf{q}_{t}+\mathbf{q}_{xx}-2(\Vert \mathbf{q} \Vert^{2}-q_{0}^{2})\mathbf{q}+\mathrm{i}\sigma[\mathbf{q}_{xxx}-3\Vert \mathbf{q} \Vert^{2}\mathbf{q}_{x}-3(\mathbf{q}^{\dagger}\mathbf{q}_{x})\mathbf{q}]=\mathbf{0},
\end{split}
\end{equation}
and the corresponding NZBCs at infinity are
\begin{equation}\label{1.4}
\begin{split}
\lim_{x\rightarrow\pm\infty}{\mathbf{q}}(x,t)=\mathbf{q}_{\pm}=\mathbf{q}_{0}
\mathrm{e}^{\mathrm{i}\delta_{\pm}},
\end{split}
\end{equation}
where $\mathbf{q}=\mathbf{q}(x,t)=(q_{1}(x,t),q_{2}(x,t))^{T}$ and $\mathbf{q}_{0}$ are two-component vectors, $T$ represents transposition, with $\delta_{\pm}$ and $q_{0}=\Vert \mathbf{q}_{0} \Vert$ are real numbers.

The utilization of sophisticated techniques such as the IST and Riemann-Hilbert (RH) method has significantly advanced our understanding of the scalar Hirota equation under various boundary conditions~\cite{46,47,48,49,50}. Zhang, Chen and Yan~\cite{46} investigated the IST and soliton solutions for both focusing and defocusing Hirota equations under NZBCs. Chen and Yan~\cite{47} tackled the matrix RH problem associated with the Hirota equation under NZBCs. Zhang and Ling~\cite{48} provided a thorough investigation into the RH representations for two distinct Darboux matrices related to the Hirota equation, thereby clarifying the relationship between the IST and the Darboux transformation. Moreover, by integrating iterative Darboux matrices with results from a single spectral parameter, a long-term asymptotic solution characterized by spectral parameters is presented. The focused Hirota equation's Cauchy problem, which involves initial conditions that exhibit rapid decay, has been investigated employing the RH approach as referenced in~\cite{49}. However, Ye, Han, and Zhang~\cite{50} have explored the formal solutions of the defocusing Hirota equation under the influence of fully asymmetric NZBCs through the IST and the RH problem. Collectively, these scholarly endeavors have not only enriched our comprehension of the Hirota equation within the context of diverse boundary conditions but have also laid a solid groundwork for future investigations into additional boundary conditions and their potential applications.

After conducting in-depth research on the defocusing-defocusing coupled Hirota equations with parallel~\cite{23} boundary conditions $\mathbf{q}_{+} \parallel \mathbf{q}_{-}$, this paper further explores the situation under non-parallel boundary conditions $\mathbf{q}_{+} \nparallel \mathbf{q}_{-}$. Although previous research~\cite{23} has mainly focused on $\mathbf{q}_{+} \parallel \mathbf{q}_{-}$, this study extends this field by analyzing the dynamic behavior of $\mathbf{q}_{+} \nparallel \mathbf{q}_{-}$ for the first time. Our goal is to fill the gap in existing literature and provide new insights into the behavior of the defocusing-defocusing coupled Hirota equations under non-parallel boundary conditions. This research significantly contributes to our comprehension of the defocusing-defocusing coupled Hirota equations and paves the way for further studies. It particularly highlights the potential for investigating how diverse boundary conditions can influence the system's dynamics.

The special case of parallel NZBCs~\eqref{1.4} has been resolved~\cite{23} and we consider extending the IST to the non-parallel NZBCs where any $\mathbf{q}_{\pm}$ are only constrained by $\Vert \mathbf{q}_{\pm} \Vert=q_{0}$. However, the boundary conditions in the parallel case~\cite{23} have two constraints: $\Vert \mathbf{q}_{\pm} \Vert=q_{0}$ and $|\mathbf{q}_{+}^{\dagger}\mathbf{q}_{-}|=q_{0}^{2}$. It's important to mention that the findings from~\cite{23} on the direct problem are not limited to parallel situations, they are also applicable to non-parallel cases. Due to the difference in eigenvector matrices between parallel and non-parallel cases, it is necessary to reconstruct the expression for the inverse problem in non-parallel cases. The above are the two main differences between non-parallel NZBCs and parallel NZBCs.

There are two important reasons for studying non-parallel NZBCs, which are challenging in theory and application. Firstly, the non-parallel boundary conditions break the traditional boundary constraints and make the construction of IST more complicated~\cite{12}. Especially in the orthogonal NZBCs, that is $\mathbf{q}_{+}^{\dagger}\mathbf{q}_{-}=0$. To differentiate the scenario with orthogonal NZBCs, the term "non-orthogonal NZBCs" is used for the more general situation involving non-parallel NZBCs. In fact, it can be found that the construction inverse problem of orthogonal and non-orthogonal cases is different from the traditional method~\cite{12,51}. We will observe that although there are two different construction methods for non-orthogonal cases, only one of them has the ability to extend to orthogonal cases. Secondly, the introduction of non-parallel boundary conditions makes it possible to explore new physical phenomena and mathematical structures. Therefore, although non-parallel boundary conditions bring additional mathematical problems, they also provide new perspectives and tools for understanding complex systems.

Furthermore, the above issues genuinely highlight the distinct solution behaviors observed in parallel versus non-parallel scenarios. Specifically, if $\mathbf{q}_{-}$ and $\mathbf{q}_{+}$ are not parallel, then there cannot be the pure reflectionless solutions. This parallels the scenario within the defocusing Hirota equation~\cite{50} when subjected to fully asymmetric NZBCs. Another reason for considering the defocusing-defocusing coupled Hirota equations~\eqref{1.3} with non-parallel NZBCs is also motivated by the desire to explore physically significant situations. Similarly, we can use the relevant real three-component Stokes vectors to equivalently describe the polarization state~\cite{52}. Moreover, the non-parallel solutions of $\mathbf{q}_{-}$ and $\mathbf{q}_{+}$ describe a system configuration that includes transitions between non-parallel asymptotic polarization states. In the case of considering a Bose-Einstein condensate, the solutions represent a system state that are spatially divided into interactions with different properties. Consequently, we anticipate that the findings from this research will furnish a crucial analytical framework, empowering scientists to precisely measure the observed physical phenomena~\cite{A9}. Through this method, researchers will be able to gain a deeper understanding of the impact of solutions under non-parallel conditions on system dynamics, providing new perspectives for analytical and empirical research in related fields~\cite{A10}.

The organization of the subsequent sections in this paper is as follows. In Section~\ref{s:Direct scattering problem} we delineate the Jost eigenfunctions, the scattering matrix and the modified eigenfunctions and rigorously establish their analytical properties. In Section~\ref{s:Discrete spectrum and its related properties}, we delve into the characteristics of the discrete spectrum and analyze the asymptotic properties of the modified Jost eigenfunctions and the scattering matrix elements as $z\rightarrow\infty$ and $z\rightarrow 0$. We consider the asymptotic behavior of reflection coefficients in both orthogonal and non-orthogonal cases. It proves that there are non-existence of reflectionless potentials solutions under non-parallel conditions. Section~\ref{s:Inverse problem} presents the inverse problem. We examine the RH problem for both non-orthogonal and orthogonal scenarios. By using the meromorphic matrices, the corresponding residue conditions and norming constants are obtained. We construct the formal solution of the RH problem and reconstruction formula with the help of the Plemelj's formula. In Section~\ref{s:Discussion and final remarks}, we present our discussions and concluding observations.

\section{Direct scattering problem}
\label{s:Direct scattering problem}

Now, let's delve deeper into the construction of the IST method for the defocusing-defocusing coupled Hirota equations~\eqref{1.3} in specific scenarios with non-parallel NZBCs~\cite{12}, as follows:
\begin{equation}\label{2.1}
\begin{split}
0\leq \frac{|\mathbf{q}_{+}^{\dagger} \mathbf{q}_{-}|}{q_{0}^{2}}<1,
\end{split}
\end{equation}
unlike the parallel case discussed in reference~\cite{23}, in which $|\mathbf{q}_{+}^{\dagger}\mathbf{q}_{-}|=q_{0}^{2}$. In the orthogonal NZBCs, we found $\mathbf{q}_{+}^{\dagger}\mathbf{q}_{-}=0$ which is in stark contrast to the parallel NZBCs situation in previous studies~\cite{23}. Here, we concentrate our efforts on examining the direct problems. The construction of direct problems largely follows the framework in reference~\cite{23}. But there is a significant difference: the normalization of the Jost eigenfunctions.

\subsection{Lax pair and uniformization variable}
\label{s:Lax pair and uniformization variable}

The defocusing-defocusing coupled Hirota equations~\eqref{1.3} have a Lax pair formulation, which is given by:
\begin{equation}\label{2.2}
\begin{split}
\psi_{x}&=\mathbf{X}\psi, \quad \psi_{t}=\mathbf{T}\psi,
\end{split}
\end{equation}
where $\psi=\psi(x,t)$, the matrices represented by $\mathbf{X}$ and $\mathbf{T}$ are presented below:
\begin{equation}\label{2.3}
\begin{split}
\mathbf{X}=\mathbf{X}(k;x,t)&=\mathrm{i}k\mathbf{J}+\mathrm{i}\mathbf{Q}, \\
\mathbf{T}=\mathbf{T}(k;x,t)&=4\mathrm{i}\sigma k^{3}\mathbf{J}-\mathrm{i}q_{0}^{2}\mathbf{J}
+k^{2}(4\mathrm{i}\sigma\mathbf{Q}+2\mathrm{i}\mathbf{J})
+k(2\mathrm{i}\mathbf{Q}-2\sigma\mathbf{Q}_{x}\mathbf{J}-2\mathrm{i}\sigma\mathbf{J}\mathbf{Q}^{2})
-2\mathrm{i}\sigma\mathbf{Q}^{3} \\
&-\mathrm{i}\mathbf{J}\mathbf{Q}^{2}-\mathrm{i}\sigma\mathbf{Q}_{xx}
-\mathbf{Q}_{x}\mathbf{J}+\sigma[\mathbf{Q},\mathbf{Q}_{x}],
\end{split}
\end{equation}
where $\left[ Q_{1},Q_{2} \right]=Q_{1}Q_{2}-Q_{2}Q_{1}$, $k$ is the spectral parameter and the formulations of matrices $\mathbf{J}$ and $\mathbf{Q}=\mathbf{Q}(x,t)$ are delineated below:
\begin{equation}\label{2.4}
\begin{split}
\mathbf{J}=\begin{pmatrix}
    1 & \mathbf{0_{1\times2}}  \\
    \mathbf{0_{2\times1}} & -\mathbf{I_{2\times2}}
  \end{pmatrix},  \quad
\mathbf{Q}=\begin{pmatrix}
    0 & -\mathbf{q}^{\dagger}   \\
    \mathbf{q} & \mathbf{0_{2\times2}}
  \end{pmatrix}.
\end{split}
\end{equation}

The asymptotic scattering problems as $x\to\pm\infty$ are as follows:
\begin{equation}\label{2.5}
\begin{split}
\psi_{x}=\mathbf{X}_{\pm}\psi, \quad \psi_{t}=\mathbf{T}_{\pm}\psi,
\end{split}
\end{equation}
where
\begin{subequations}\label{2.6}
\begin{align}
\lim_{x\rightarrow\pm\infty}{\mathbf{X}}=\mathbf{X}_{\pm}&=\mathrm{i}k\mathbf{J}+\mathrm{i}\mathbf{Q}_{\pm}, \\
\lim_{x\rightarrow\pm\infty}{\mathbf{T}}=\mathbf{T}_{\pm}&=4\mathrm{i}\sigma k^{3}\mathbf{J}-\mathrm{i}q_{0}^{2}\mathbf{J}
+k^{2}(4\mathrm{i}\sigma\mathbf{Q}_{\pm}+2\mathrm{i}\mathbf{J})
+k(2\mathrm{i}\mathbf{Q}_{\pm}-2\mathrm{i}\sigma\mathbf{J}\mathbf{Q}_{\pm}^{2})
-2\mathrm{i}\sigma\mathbf{Q}_{\pm}^{3}-\mathrm{i}\mathbf{J}\mathbf{Q}_{\pm}^{2}.
\end{align}
\end{subequations}
Then, the eigenvalues of the matrices $\mathbf{X}_{\pm}$ and $\mathbf{T}_{\pm}$ are deduced as delineated below:
\begin{equation}\label{2.7}
\begin{split}
\mathbf{X}_{\pm,1}=-\mathrm{i}k, \quad
\mathbf{X}_{\pm,2,3}=\pm\mathrm{i}\lambda, \quad
\mathbf{T}_{\pm,1}=-\mathrm{i}(\lambda^{2}+k^{2}+4\sigma k^{3}), \quad
\mathbf{T}_{\pm,2,3}=\pm 2\mathrm{i}\lambda[k+\sigma(3k^{2}-\lambda^{2})],
\end{split}
\end{equation}
where $\lambda(k)=\sqrt{k^{2}-q_{0}^{2}}$. In the case that these eigenvalues have branching, they are processed according to the method proposed in~\cite{10}. We bring in a two-sheeted Riemann surface characterized by $\lambda(k)=\sqrt{k^{2}-q_{0}^{2}}$, the branch points are identified at the $k$ values that lead to $\lambda(k)=0$ being nullified, specifically $k=\pm q_{0}$. Subsequently, we introduce the transformation variable $z=k+\lambda$ and proceed with the analysis.
\begin{equation}\label{2.8}
\begin{split}
k(z)=\frac{1}{2}\left(z+\frac{q_{0}^{2}}{z}\right), \quad \lambda(z)=\frac{1}{2}\left(z-\frac{q_{0}^{2}}{z}\right).
\end{split}
\end{equation}
The branch cuts that separate the two sheets of the Riemann surface correspond to the real $z$-axis: the first sheet $\mathbb{D}^{+}=\left\{z \in \mathbb{C}: \operatorname{Im} z>0\right\}$ and the second sheet $\mathbb{D}^{-}=\left\{z \in \mathbb{C}: \operatorname{Im} z<0\right\}$.

\subsection{Jost solutions, scattering matrix and analyticity}
\label{s:Jost solutions, scattering matrix and analyticity}

Usually, the scattering problem's continuous spectrum includes all $k$ values from both sheets where $\lambda(k)\in \mathbb{R}$. Then, the continuous spectrum is the set $k\in \mathbb{R} \backslash (-q_{0},q_{0})$. The continuous spectrum is $z\in \mathbb{R}$ in the complex $z$-plane. We define a two-component vector $\mathbf{v}=(v_{1},v_{2})^{T}$ and its corresponding orthogonal vector as $\mathbf{v}^{\perp}=(v_{2},-v_{1})^{\dagger}$.

Consider new eigenvector matrices different from the reference~\cite{23}, which are called non-parallel eigenvector matrices $\widehat{\mathbf{Y}}_{\pm}(z)$
\begin{equation}\label{2.9}
\everymath{\displaystyle}
\begin{split}
\widehat{\mathbf{Y}}_{\pm}(z)=\widehat{\rho}(z)\begin{pmatrix}
    \mathrm{i} & 0 &  \frac{q_{0}}{z}  \\
    \frac{\mathrm{i}\mathbf{q}_{\pm}}{z} &
    \frac{\mathbf{q}_{\pm}^{\bot}}{\widehat{\rho}(z)q_{0}} &
     \frac{\mathbf{q}_{\pm}}{q_{0}}
  \end{pmatrix}, \quad
\mathbf{Y}_{\pm}^{[\mathrm{bk}]}(z)=\begin{pmatrix}
    \mathrm{i} & 0 &  \frac{q_{0}}{z}  \\
    \frac{\mathrm{i}\mathbf{q}_{\pm}}{z} &
    \frac{\mathbf{q}_{\pm}^{\bot}}{q_{0}} &
     \frac{\mathbf{q}_{\pm}}{q_{0}}
  \end{pmatrix}, \quad \widehat{\rho}(z)=\frac{z^{2}}{z^{2}-q_{0}^{2}},
\end{split}
\end{equation}
where $\mathbf{Y}_{\pm}^{[\mathrm{bk}]}(z)$ are the eigenvector matrices under parallel NZBCs. By using NZBCs~\eqref{1.4} and equation~\eqref{2.6}, it can be concluded that the condition $[\mathbf{X}_{\pm},\mathbf{T}_{\pm}]=\mathbf{0_{3\times3}}$ holds, $\mathbf{X}_{\pm}$ and $\mathbf{T}_{\pm}$ allow for the existence of common eigenvectors: $\mathbf{X}_{\pm}=\mathrm{i}\widehat{\mathbf{Y}}_{\pm}\mathbf{\Lambda}_{1}
\widehat{\mathbf{Y}}_{\pm}^{-1}$ and $\mathbf{T}_{\pm}=\mathrm{i}\widehat{\mathbf{Y}}_{\pm}\mathbf{\Lambda}_{2}
\widehat{\mathbf{Y}}_{\pm}^{-1}$, where $\widehat{\mathbf{Y}}_{\pm}=\widehat{\mathbf{Y}}_{\pm}(z)$ and
\begin{subequations}\label{2.10}
\begin{align}
\mathbf{\Lambda}_{1}&=\mathbf{\Lambda}_{1}(z)=\operatorname{diag}\left( \lambda,-k,-\lambda \right), \\
\mathbf{\Lambda}_{2}&=\mathbf{\Lambda}_{2}(z)=\operatorname{diag}\left(
2\lambda[k+\sigma(3k^{2}-\lambda^{2})],-(\lambda^{2}+k^{2}+4\sigma k^{3}), -2\lambda[k+\sigma(3k^{2}-\lambda^{2})] \right).
\end{align}
\end{subequations}
In the subsequent citations of the article, we particularly focused on the following aspects:
\begin{equation}\label{2.11}
\begin{split}
 \quad
\widehat{\mathbf{Y}}_{\pm}^{-1}(z)=\begin{pmatrix}
    -\mathrm{i} & \mathrm{i}\mathbf{q}_{\pm}^{\dagger}/z  \\
    0 & (\mathbf{q}_{\pm}^{\bot})^{\dagger}/q_{0}  \\
    -q_{0}/z &  \mathbf{q}_{\pm}^{\dagger}/q_{0}  \\
  \end{pmatrix}, \quad  \det\widehat{\mathbf{Y}}_{\pm}^{-1}(z)=\frac{1}{\mathrm{i}\widehat{\rho}(z)},
  \quad
  \det\mathbf{Y}_{\pm}^{[\mathrm{bk}]}(z)=\frac{1}{\widehat{\rho}(z)}.
\end{split}
\end{equation}

\begin{remark}\label{rem:1}
Since $\det\widehat{\mathbf{Y}}_{\pm}(z)=\mathrm{i}\widehat{\rho}(z)$, it can be inferred that the eigenfunctions exhibit singularities at $z=\pm q_{0}$. Be aware that this eigenvector matrix differs from the one employed in parallel scenario as described in~\cite{23}. Therefore, the matrixs $\widehat{\mathbf{Y}}_{\pm}(z)$ defined in~\eqref{2.9} have the relationship $\widehat{\mathbf{Y}}_{\pm}(z)=\mathbf{Y}_{\pm}^{[\mathrm{bk}]}(z)\widehat{\mathbf{J}}_{1}(z)$ with the eigenvector matrix $\mathbf{Y}_{\pm}^{[\mathrm{bk}]}(z)$ in~\cite{23}, where $\widehat{\mathbf{J}}_{1}(z)=\operatorname{diag}\left( z/2\lambda, 1, z/2\lambda \right)$.
\end{remark}

\begin{remark}\label{rem:2}
The decision to employ various eigenvector matrices is due to the fact that $\mathbf{Y}_{\pm}^{[\mathrm{bk}]}(z)$ have a singularity at $z=0$. This problem is particularly evident when $\mathbf{q}_{+}$ and $\mathbf{q}_{-}$ are not parallel, causing certain reflection coefficients to become diverge as $z$ approaches zero or infinity. Although the new eigenvector matrices exhibit singularities at $z=\pm q_{0}$, it provides us with a superior method for solving problems.
\end{remark}

\begin{remark}\label{rem:3}
Given that $\det\widehat{\mathbf{Y}}_{\pm}(z)$ have a double root at $z=0$, the related asymptotic eigenvectors are linearly dependent at that point. This is another important difference compared to $\mathbf{Y}_{\pm}^{[\mathrm{bk}]}(z)$, as $\det\mathbf{Y}_{\pm}^{[\mathrm{bk}]}(\pm q_{0})=0$ implies that the asymptotic eigenvectors at $z=\pm q_{0}$ are linearly correlated.
\end{remark}

We specify the Jost eigenfunctions $\psi(z;x,t)$ that meet the boundary conditions:
\begin{equation}\label{2.12}
\begin{split}
\psi_{\pm}(z)=\psi_{\pm}(z;x,t)=\widehat{\mathbf{Y}}_{\pm}(z) \mathrm{e}^{\mathrm{i}\mathbf{\Delta}(z)}+o(1), \quad x\rightarrow\pm\infty, \quad z\in \mathbb{R},
\end{split}
\end{equation}
with $\mathbf{\Delta}(z)=\mathbf{\Delta}(z;x,t)=\operatorname{diag} \left( \delta_{1},\delta_{2},-\delta_{1} \right)$ and
\begin{equation}\label{2.13}
\begin{split}
\delta_{1}=\delta_{1}(z;x,t)=\lambda x+2\lambda[k+\sigma(3k^{2}-\lambda^{2})]t, \quad
\delta_{2}=\delta_{2}(z;x,t)=-kx-(\lambda^{2}+k^{2}+4\sigma k^{3})t.
\end{split}
\end{equation}
We define the modified eigenfunctions $\nu_{\pm}(z)=\nu_{\pm}(z;x,t)=\psi_{\pm}(z)\mathrm{e}^{-\mathrm{i}\mathbf{\Delta}(z)}$ such that $\lim_{x\rightarrow\pm\infty}{\nu_{\pm}}(z)=\widehat{\mathbf{Y}}_{\pm}(z)$. Due to the difference between $\widehat{\mathbf{Y}}_{\pm}(z)$ and $\mathbf{Y}_{\pm}^{[\mathrm{bk}]}(z)$ in~\cite{23}, the matrixs $\psi_{\pm}(z)$ defined in~\eqref{2.12} have the clear relationship $\psi_{\pm}(z)=\psi_{\pm}^{[\mathrm{bk}]}(z)\widehat{\mathbf{J}}_{1}(z)$ with $\psi_{\pm}^{[\mathrm{bk}]}(z)$ are the matrix Jost solutions under parallel NZBCs.

We perform factor decomposition on the asymptotic behavior of the potential and reframe the Lax pair~\eqref{2.5} as
\begin{equation}\label{2.14}
\begin{split}
(\psi_{\pm})_{x}=\mathbf{X}_{\pm}\psi_{\pm}+(\mathbf{X}-\mathbf{X}_{\pm})\psi_{\pm}, \quad
(\psi_{\pm})_{t}=\mathbf{T}_{\pm}\psi_{\pm}+(\mathbf{T}-\mathbf{T}_{\pm})\psi_{\pm},
\end{split}
\end{equation}
with $\psi_{\pm}=\psi_{\pm}(z)$ and it can be equivalent to the following equations:
\begin{equation}\label{2.15}
\begin{split}
(\widehat{\mathbf{Y}}_{\pm}^{-1}\nu_{\pm})_{x}=[\mathrm{i}\mathbf{\Lambda}_{1},
\widehat{\mathbf{Y}}_{\pm}^{-1}\nu_{\pm}]
+\widehat{\mathbf{Y}}_{\pm}^{-1}(\mathbf{X}-\mathbf{X}_{\pm})\nu_{\pm},  \quad
(\widehat{\mathbf{Y}}_{\pm}^{-1}\nu_{\pm})_{t}=[\mathrm{i}\mathbf{\Lambda}_{2},
\widehat{\mathbf{Y}}_{\pm}^{-1}\nu_{\pm}]
+\widehat{\mathbf{Y}}_{\pm}^{-1}(\mathbf{T}-\mathbf{T}_{\pm})\nu_{\pm},
\end{split}
\end{equation}
where $\nu_{\pm}=\nu_{\pm}(z)$. It has been confirmed that the spectral issue related to $\nu_{\pm}(z)$ is equivalent to the Volterra integral equations
\begin{equation}\label{2.16}
\begin{split}
\nu_{\pm}(z)=\widehat{\mathbf{Y}}_{\pm}(z)+\int_{\pm\infty}^{x}\widehat{\mathbf{Y}}_{\pm}(z)
\mathrm{e}^{\mathrm{i}(x-y)\mathbf{\Lambda}_{1}}\widehat{\mathbf{Y}}_{\pm}^{-1}(z)
\left[\mathbf{X}(z;y,t)-\mathbf{X}_{\pm}(z)\right]\nu_{\pm}(z;y,t)
\mathrm{e}^{-\mathrm{i}(x-y)\mathbf{\Lambda}_{1}}\,\mathrm{d}y.
\end{split}
\end{equation}
It can be seen from the expression in~\eqref{2.11} that $\widehat{\mathbf{Y}}_{\pm}^{-1}(z)$ have a simple pole at $z=0$. However, the expression $\widehat{\mathbf{Y}}_{\pm}(z)\mathrm{e}^{\pm\mathrm{i}(x-y)\mathbf{\Lambda}_{1}(z)}
\widehat{\mathbf{Y}}_{\pm}^{-1}(z)$ remains bounded for any $x,y\in \mathbb{R}$ as $z\rightarrow0$. We establish the Jost eigenfunctions by solving the integral equations given in~\eqref{2.16}. By applying the Neumann iteration method, we successfully verify the following theorems~\cite{23}.

\begin{theorem}\label{thm:1}
Suppose that $\mathbf{q}(x,t)-\mathbf{q}_{-} \in L^{1}(\mathbb{R}_{x}^{-})$ or $\mathbf{q}(x,t)-\mathbf{q}_{+} \in L^{1}(\mathbb{R}_{x}^{+})$, the modified eigenfunctions $\nu_{\pm}(z)$ are continuous functions of $z$ for $z\in\mathbb{R} \backslash \{\pm q_{0}\}$. Then, the columns of $\nu_{\pm}(z)$ fulfill the requisite properties: $\nu_{-,1}(z)$ and $\nu_{+,3}(z)$ for $z\in \mathbb{D}^{-}$, $\nu_{-,3}(z)$ and $\nu_{+,1}(z)$ for $z\in \mathbb{D}^{+}$.
\end{theorem}

Moreover, it can be inferred from $\nu_{\pm}(z)=\psi_{\pm}(z)\mathrm{e}^{-\mathrm{i}\mathbf{\Delta}(z)}$ that the columns of $\nu_{\pm}(z)$ and $\psi_{\pm}(z)$ have the same analytical property. Next, we will assume that both conditions in Theorem~\ref{thm:1} hold, namely $\mathbf{q}(x,t)-\mathbf{q}_{\pm} \in L^{1}(\mathbb{R}_{x}^{\pm})$. Then $\det\psi_{\pm}(z)=\mathrm{i}\widehat{\rho}(z)\mathrm{e}^{\mathrm{i}\delta_{2}(z)}$ for $(x,t)\in\mathbb{R}^{2}$ and $z\in\mathbb{R} \backslash \{\pm q_{0}\}$. Next, construct the corresponding $3\times3$ scattering matrix defined by the following expression
\begin{equation}\label{2.17}
\begin{split}
\psi_{+}(z)=\psi_{-}(z)\mathbf{H}(z), \quad z\in\mathbb{R} \backslash \{0,\pm q_{0}\},
\end{split}
\end{equation}
where $\mathbf{H}(z)=(h_{ij}(z))$. By using $\det\psi_{\pm}(z)=\mathrm{i}\widehat{\rho}(z)\mathrm{e}^{\mathrm{i}\delta_{2}(z)}$ and~\eqref{2.17}, we obtain $\det\mathbf{H}(z)=1$ for $z\in\mathbb{R} \backslash \{0,\pm q_{0}\}$. Similarly, the inverse matrix $\mathbf{S}(z)=\mathbf{H}^{-1}(z)=(s_{ij}(z))$. Therefore, the scattering matrix defined under non-parallel conditions is different from the definition in~\cite{23}. More precisely, the scattering matrix defined in the parallel case~\cite{23} is $\psi_{+}^{[\mathrm{bk}]}(z)=\psi_{-}^{[\mathrm{bk}]}(z)\mathbf{H}^{[\mathrm{bk}]}(z)$. So it can be clearly stated that $\mathbf{H}^{[\mathrm{bk}]}(z)$ in the parallel case~\cite{23} and $\mathbf{H}(z)$ in the non-parallel case have the following relationship for $z\in\mathbb{R} \backslash \{0,\pm q_{0}\}$
\begin{equation}\label{2.18}
\begin{split}
\mathbf{H}(z)=\widehat{\mathbf{J}}_{1}^{-1}(z)\mathbf{H}^{[\mathrm{bk}]}(z)\widehat{\mathbf{J}}_{1}(z), \quad
\mathbf{S}(z)=\widehat{\mathbf{J}}_{1}^{-1}(z)\mathbf{S}^{[\mathrm{bk}]}(z)\widehat{\mathbf{J}}_{1}(z),
\end{split}
\end{equation}
where the inverse matrix $\mathbf{S}^{[\mathrm{bk}]}(z)=[\mathbf{H}^{[\mathrm{bk}]}(z)]^{-1}$.

\begin{theorem}\label{thm:2}
According to the same assumption in Theorem~\ref{thm:1}, the scattering coefficients have the properties: $s_{33}(z)$ and $h_{11}(z)$ for $z\in \mathbb{D}^{+}$, $s_{11}(z)$ and $h_{33}(z)$ for $z\in \mathbb{D}^{-}$.
\end{theorem}

\begin{remark}\label{rem:4}
In addition to the difference of the eigenvector matrix, the above results are the same as those in the parallel case~\cite{23}. Obviously, there is a significant difference between non-parallel and parallel cases. In the framework of the parallel case~\cite{23}, the matrix Jost solutions $\psi_{\pm}^{[\mathrm{bk}]}(z)$ can be analytically continued to $z=\pm q_{0}$ provided that the scattering potential meets an extra decay criterion as $x\rightarrow \pm\infty$. However, in the non-parallel case, this kind of continuity extension is impossible. The reason is that the eigenvector matrix produces poles at $z=\pm q_{0}$, which hinders its continuity.
\end{remark}

\subsection{Adjoint problem}
\label{s:Adjoint problem}

To address the issue that $\nu_{\pm,2}(z)$ is not the analytically modified eigenfunctions, consider its corresponding "adjoint" Lax pair in~\cite{23} (following the terminology of~\cite{10}):
\begin{equation}\label{2.19}
\begin{split}
\widetilde{\psi}_{x}&=\widetilde{\mathbf{X}}\widetilde{\psi}, \quad
\widetilde{\psi}_{t}=\widetilde{\mathbf{T}}\widetilde{\psi},
\end{split}
\end{equation}
where
\begin{equation}\label{2.20}
\begin{split}
\widetilde{\mathbf{X}}=\widetilde{\mathbf{X}}(k;x,t)&=-\mathrm{i}k\mathbf{J}-\mathrm{i}\mathbf{Q}^{*}, \\
\widetilde{\mathbf{T}}=\widetilde{\mathbf{T}}(k;x,t)&=-4\mathrm{i}\sigma k^{3}\mathbf{J}+\mathrm{i}q_{0}^{2}\mathbf{J}
-k^{2}(4\mathrm{i}\sigma\mathbf{Q}^{*}+2\mathrm{i}\mathbf{J})
+k(2\mathrm{i}\sigma\mathbf{J}(\mathbf{Q}^{*})^{2}-2\mathrm{i}\mathbf{Q}^{*}
-2\sigma\mathbf{Q}_{x}^{*}\mathbf{J}) \\
&+2\mathrm{i}\sigma(\mathbf{Q}^{*})^{3}+\mathrm{i}\mathbf{J}(\mathbf{Q}^{*})^{2}
+\mathrm{i}\sigma\mathbf{Q}_{xx}^{*}-\mathbf{Q}_{x}^{*}\mathbf{J}
+\sigma[\mathbf{Q}^{*},\mathbf{Q}_{x}^{*}],
\end{split}
\end{equation}
where $\widetilde{\psi}=\widetilde{\psi}(z)=\widetilde{\psi}(z;x,t)$, $\widetilde{\mathbf{X}}=\mathbf{X}^{*}$ and $\widetilde{\mathbf{T}}=\mathbf{T}^{*}$ for all $z\in \mathbb{R}$, while $*$ denotes complex conjugation and $\mathbf{J}\mathbf{Q}=-\mathbf{Q}\mathbf{J}$, $\mathbf{J}\mathbf{Q}^{*}=-\mathbf{Q}^{*}\mathbf{J}$, $\mathbf{J}\mathbf{Q}^{*}=\mathbf{Q}^{T}\mathbf{J}$. It can be proven that the defocusing-defocusing coupled Hirota equations~\eqref{1.3} are also equivalent to the equation $\widetilde{\mathbf{X}}_{t}-\widetilde{\mathbf{T}}_{x}+[\widetilde{\mathbf{X}}, \widetilde{\mathbf{T}}]=\mathbf{0}$. Then one has:

\begin{proposition}\label{pro:1}
If $\widetilde{\mathbf{v}}_{2}(z)$ and $\widetilde{\mathbf{v}}_{3}(z)$ are two arbitrary solutions of the "adjoint" Lax pair~\eqref{2.19}, while "$\times$" denotes the usual cross product, then $\mathbf{v}_{1}(z)=\mathrm{e}^{\mathrm{i}\delta_{2}(z)}\mathbf{J}
[\widetilde{\mathbf{v}}_{2}(z) \times \widetilde{\mathbf{v}}_{3}(z)]$ is a solution of the Lax pair~\eqref{2.2}.
\end{proposition}

As $x\to\pm\infty$, the behavior of the solutions derived from the "adjoint" Lax pair will approach an asymptotic state in terms of both spatial $\widetilde{\psi}_{x}=\widetilde{\mathbf{X}}_{\pm}\widetilde{\psi}$ and temporal $\widetilde{\psi}_{t}=\widetilde{\mathbf{T}}_{\pm}\widetilde{\psi}$. The eigenvalues of $\widetilde{\mathbf{X}}_{\pm}$ are $\mathrm{i}k$ and $\pm\mathrm{i}\lambda$,
the eigenvalues of $\widetilde{\mathbf{T}}_{\pm}$ are $\mathrm{i}(\lambda^{2}+k^{2}+4\sigma k^{3})$ and $\pm 2\mathrm{i}\lambda[k+\sigma(3k^{2}-\lambda^{2})]$. Then, the Jost solutions $\widetilde{\psi}_{\pm}(z)$ of the "adjoint" Lax pair~\eqref{2.19}
\begin{equation}\label{2.21}
\begin{split}
\widetilde{\psi}_{\pm}(z)=\widetilde{\mathbf{Y}}_{\pm}(z) \mathrm{e}^{-\mathrm{i}\mathbf{\Delta}(z)}+o(1), \quad x\rightarrow\pm\infty, \quad z\in\mathbb{R} \backslash \{\pm q_{0}\},
\end{split}
\end{equation}
where $\widetilde{\mathbf{Y}}_{\pm}(z)=\widehat{\mathbf{Y}}_{\pm}^{*}(z)$ for $z\in\mathbb{R} \backslash \{\pm q_{0}\}$ and $\det \widetilde{\mathbf{Y}}_{\pm}(z)=-\mathrm{i}\widehat{\rho}(z)$. Introducing the modified Jost solutions $\widetilde{\nu}_{\pm}(z)=\widetilde{\psi}_{\pm}(z)\mathrm{e}^{\mathrm{i}\mathbf{\Delta}(z)}$. Through analytical analysis~\cite{23}, it can be concluded that $\widetilde{\nu}_{-,3}(z)$ and $\widetilde{\nu}_{+,1}(z)$ for $z\in \mathbb{D}^{-}$, $\widetilde{\nu}_{-,1}(z)$ and $\widetilde{\nu}_{+,3}(z)$ for $z\in \mathbb{D}^{+}$. Construct the "adjoint" scattering matrix $\widetilde{\mathbf{H}}(z)=(\widetilde{h}_{ij}(z))$ for $z\in\mathbb{R} \backslash \{0,\pm q_{0}\}$ corresponding to the relationship $\widetilde{\psi}_{+}(z)=\widetilde{\psi}_{-}(z)\widetilde{\mathbf{H}}(z)$. Similarly, define $\widetilde{\mathbf{S}}(z)=\widetilde{\mathbf{H}}^{-1}(z)=(\widetilde{s}_{ij}(z))$. The scattering coefficient satisfies analytical properties: $\widetilde{s}_{33}(z)$ and $\widetilde{h}_{11}(z)$ for $z\in \mathbb{D}^{-}$, $\widetilde{s}_{11}(z)$ and $\widetilde{h}_{33}(z)$ for $z\in \mathbb{D}^{+}$. Next, we can establish two new auxiliary eigenfunctions for the original Lax pair \eqref{2.2}:
\begin{equation}\label{2.22}
\begin{split}
\gamma(z)=-\frac{\mathrm{i}\mathrm{e}^{\mathrm{i}\delta_{2}}\mathbf{J}[\widetilde{\psi}_{-,3}(z)
\times \widetilde{\psi}_{+,1}(z)]}{\widehat{\rho}(z)}, \quad z\in \mathbb{D}^{-},  \quad
\widetilde{\gamma}(z)=-\frac{\mathrm{i}\mathrm{e}^{\mathrm{i}\delta_{2}} \mathbf{J}[\widetilde{\psi}_{-,1}(z)
\times \widetilde{\psi}_{+,3}(z)]}{\widehat{\rho}(z)}, \quad z\in \mathbb{D}^{+},
\end{split}
\end{equation}
where $\gamma(z)=\gamma(z;x,t)$ and $\widetilde{\gamma}(z)=\widetilde{\gamma}(z;x,t)$. The relationship between $\mathbf{S}(z)$ and $\widetilde{\mathbf{S}}(z)$ for $z\in\mathbb{R} \backslash \{0,\pm q_{0}\}$ is $\widetilde{\mathbf{S}}^{-1}(z)=\widehat{\mathbf{J}}_{2}(z)\mathbf{S}^{T}(z)\widehat{\mathbf{J}}_{2}^{-1}(z)$, where $\widehat{\mathbf{J}}_{2}(z)=\operatorname{diag}\,(-1,\widehat{\rho}(z),1)$. Subsequently, we can establish the following relationship~\cite{23}:

\begin{corollary}\label{cor:1}
For all cyclic indices $j$, $l$ and $m$ with $\widehat{\rho}_{1}(z)=-1$, $\widehat{\rho}_{2}(z)=\widehat{\rho}(z)$ and $\widehat{\rho}_{3}(z)=1$,
\begin{equation}\label{2.23}
\begin{split}
\psi_{\pm,j}(z)=-\frac{\mathrm{i}\mathrm{e}^{\mathrm{i}\delta_{2}}\mathbf{J}
[\widetilde{\psi}_{\pm,l}(z) \times \widetilde{\psi}_{\pm,m}(z)]}{\widehat{\rho}_{j}(z)}, \quad
\widetilde{\psi}_{\pm,j}(z)=\frac{\mathrm{i}\mathrm{e}^{-\mathrm{i}\delta_{2}}\mathbf{J}
[\psi_{\pm,l}(z) \times \psi_{\pm,m}(z)]}{\widehat{\rho}_{j}(z)}, \quad
z\in\mathbb{R} \backslash \{0,\pm q_{0}\}.
\end{split}
\end{equation}
\end{corollary}

\begin{corollary}\label{cor:2}
The Jost eigenfunctions exhibit the following decompositions for $z\in\mathbb{R} \backslash \{0,\pm q_{0}\}$,
\begin{subequations}\label{2.24}
\begin{align}
\psi_{-,2}(z)&=\frac{s_{32}(z)\psi_{-,3}(z)-\widetilde{\gamma}(z)}{s_{33}(z)}
=\frac{s_{12}(z)\psi_{-,1}(z)+\gamma(z)}{s_{11}(z)}, \\
\psi_{+,2}(z)&=\frac{h_{12}(z)\psi_{+,1}(z)-\widetilde{\gamma}(z)}{h_{11}(z)}
=\frac{h_{32}(z)\psi_{+,3}(z)+\gamma(z)}{h_{33}(z)}.
\end{align}
\end{subequations}
\end{corollary}

\subsection{Symmetries}
\label{s:Symmetries}

We consider two corresponding symmetries~\cite{23} $z\mapsto z^{*}$ and $z\mapsto -q_{0}^{2}/z$.
\begin{proposition}\label{pro:2}
If $\psi(z)$ is a non-singular solution of the scattering problem, then $\mathbf{J}[\psi^{\dagger}(z^{*})]^{-1}$ is a solution of the Lax pair~\eqref{2.2} and the symmetries $\psi_{\pm}(z)=\mathbf{J}[\psi_{\pm}^{\dagger}(z^{*})]^{-1}\widehat{\mathbf{J}}_{3}(z)$ for $z\in\mathbb{R} \backslash \{0,\pm q_{0}\}$, where $\widehat{\mathbf{J}}_{3}(z)=\operatorname{diag}\,(\widehat{\rho}(z),-1,-\widehat{\rho}(z))
=-\widehat{\rho}(z)\widehat{\mathbf{J}}_{2}^{-1}(z)$.
\end{proposition}

The analytic Jost eigenfunctions satisfy the symmetric relationships:
\begin{subequations}\label{2.25}
\begin{align}
\psi_{-,1}^{*}(z^{*})&=\frac{\mathrm{i}\mathbf{J}\mathrm{e}^{-\mathrm{i}\delta_{2}(z)}
[\widetilde{\gamma}(z) \times \psi_{-,3}(z)]}{s_{33}(z)}, \quad \psi_{+,3}^{*}(z^{*})=\frac{\mathrm{i}\mathbf{J}\mathrm{e}^{-\mathrm{i}\delta_{2}(z)}
[\widetilde{\gamma}(z) \times \psi_{+,1}(z)]}{h_{11}(z)},
\quad \operatorname{Im} z\geq0,  \\
\psi_{+,1}^{*}(z^{*})&=\frac{\mathrm{i}\mathbf{J}\mathrm{e}^{-\mathrm{i}\delta_{2}(z)}
[\gamma(z) \times \psi_{+,3}(z)]}{-h_{33}(z)}, \quad
\psi_{-,3}^{*}(z^{*})=\frac{\mathrm{i}\mathbf{J}\mathrm{e}^{-\mathrm{i}\delta_{2}(z)}
[\gamma(z) \times \psi_{-,1}(z)]}{-s_{11}(z)}, \quad \operatorname{Im} z\leq0.
\end{align}
\end{subequations}
Using the property $\widehat{\mathbf{J}}_{3}(z)=-\widehat{\rho}(z)\widehat{\mathbf{J}}_{2}^{-1}(z)$, we have the relationship $\mathbf{S}^{\dagger}(z)=\widehat{\mathbf{J}}_{2}^{-1}(z)\mathbf{H}(z)\widehat{\mathbf{J}}_{2}(z)$ for $z\in\mathbb{R} \backslash \{0,\pm q_{0}\}$. Therefore, we conclude that
\begin{subequations}\label{2.26}
\begin{align}
h_{11}(z)&=s_{11}^{*}(z), \quad \quad \quad \,\,\,\, h_{12}(z)=-\frac{s_{21}^{*}(z)}{\widehat{\rho}(z)}, \quad  \quad  h_{13}(z)=-s_{31}^{*}(z),  \\
h_{21}(z)&=-\widehat{\rho}(z)s_{12}^{*}(z), \quad \, h_{22}(z)=s_{22}^{*}(z), \quad  \quad \,\,\,\,\,
h_{23}(z)=\widehat{\rho}(z)s_{32}^{*}(z),  \\
h_{31}(z)&=-s_{13}^{*}(z), \quad \quad \quad  h_{32}(z)=\frac{s_{23}^{*}(z)}{\widehat{\rho}(z)}, \quad  \quad \,\,\,\, h_{33}(z)=s_{33}^{*}(z).
\end{align}
\end{subequations}
Then, the conclusions $h_{11}(z)=s_{11}^{*}(z^{*})$ for $\operatorname{Im} z\geq0$ and
$h_{33}(z)=s_{33}^{*}(z^{*})$ for $\operatorname{Im} z\leq0$ are obtained based on the Schwarz reflection principle.

\begin{corollary}\label{cor:3}
The new auxiliary eigenfunctions adhere to the symmetry relations:
\begin{equation}\label{2.27}
\begin{split}
\gamma(z)=-\frac{\mathrm{i}\mathrm{e}^{\mathrm{i}\delta_{2}(z)}\mathbf{J}
[\psi_{-,3}^{*}(z^{*}) \times \psi_{+,1}^{*}(z^{*})]}{\widehat{\rho}(z)}, \quad
\widetilde{\gamma}(z)=-\frac{\mathrm{i}\mathrm{e}^{\mathrm{i}\delta_{2}(z)}\mathbf{J}
[\psi_{-,1}^{*}(z^{*}) \times \psi_{+,3}^{*}(z^{*})]}{\widehat{\rho}(z)},
\end{split}
\end{equation}
where $\psi_{\pm,j}^{*}(z)=\mathrm{i}\mathrm{e}^{-\mathrm{i}\delta_{2}(z)}\mathbf{J}
[\psi_{\pm,l}(z) \times \psi_{\pm,m}(z)]/\widehat{\rho}_{j}(z)$ and $j$, $l$ and $m$ are cyclic indices for $z\in\mathbb{R} \backslash \{0,\pm q_{0}\}$.
\end{corollary}

\begin{proposition}\label{pro:3}
If $\psi(z)$ is a non-singular solution of the scattering problem, then $\psi(q_{0}^{2}/z)$ is
a solution of the Lax pair~\eqref{2.2}. The subsequent properties can be derived from the principle of progressiveness.
\begin{equation}\label{2.28}
\begin{split}
\psi_{\pm}(z)=\psi_{\pm}(\frac{q_{0}^{2}}{z})\widehat{\mathbf{J}}_{4}(z), \quad \widehat{\mathbf{J}}_{4}(z)=\begin{pmatrix}
    0 & 0 &  \mathrm{i}z/q_{0}  \\
    0 & 1 &  0  \\
    -\mathrm{i}z/q_{0} & 0 & 0
  \end{pmatrix}.
\end{split}
\end{equation}
\end{proposition}

Consistent with the previous discussion, the eigenfunctions exhibit $\psi_{\pm,2}(z)=\psi_{\pm,2}(q_{0}^{2}/z)$ for $z\in\mathbb{R} \backslash \{0,\pm q_{0}\}$ and
\begin{equation}\label{2.29}
\begin{split}
\psi_{\pm,1}(z)&=-\frac{\mathrm{i}z}{q_{0}}\psi_{\pm,3}(\frac{q_{0}^{2}}{z}), \quad
\operatorname{Im} z\gtrless0, \quad \psi_{\pm,3}(z)=\frac{\mathrm{i}z}{q_{0}}\psi_{\pm,1}(\frac{q_{0}^{2}}{z}), \quad
\operatorname{Im} z\lessgtr0.
\end{split}
\end{equation}
Moreover, the scattering matrices adhere to the following relationship for $z\in\mathbb{R} \backslash \{0,\pm q_{0}\}$,
\begin{equation}\label{2.30}
\begin{split}
\mathbf{S}(z)&=\widehat{\mathbf{J}}_{4}^{-1}(z)\mathbf{S}(\frac{q_{0}^{2}}{z})\widehat{\mathbf{J}}_{4}(z), \quad \,
\mathbf{H}(z)=\widehat{\mathbf{J}}_{4}^{-1}(z)\mathbf{H}(\frac{q_{0}^{2}}{z})\widehat{\mathbf{J}}_{4}(z),
\end{split}
\end{equation}
then we obtain
\begin{subequations}\label{2.31}
\begin{align}
s_{11}(z)&=s_{33}(\frac{q_{0}^{2}}{z}), \quad \quad  \,\,\,\, s_{12}(z)=\frac{\mathrm{i}q_{0}}{z}s_{32}(\frac{q_{0}^{2}}{z}), \quad \,\,\,
s_{13}(z)=-s_{31}(\frac{q_{0}^{2}}{z}),  \\
s_{21}(z)&=-\frac{\mathrm{i}z}{q_{0}}s_{23}(\frac{q_{0}^{2}}{z}), \quad  s_{22}(z)=s_{22}(\frac{q_{0}^{2}}{z}), \qquad  \quad
s_{23}(z)=\frac{\mathrm{i}z}{q_{0}}s_{21}(\frac{q_{0}^{2}}{z}),  \\
s_{31}(z)&=-s_{13}(\frac{q_{0}^{2}}{z}), \quad \quad
s_{32}(z)=-\frac{\mathrm{i}q_{0}}{z}s_{12}(\frac{q_{0}^{2}}{z}), \quad
s_{33}(z)=s_{11}(\frac{q_{0}^{2}}{z}).
\end{align}
\end{subequations}
Therefore, the analytical properties of the corresponding scattering coefficients include $s_{11}(z)=s_{33}(q_{0}^{2}/z)$ and $h_{33}(z)=h_{11}(q_{0}^{2}/z)$ for $z\in\mathbb{D}^{-}$, $s_{33}(z)=s_{11}(q_{0}^{2}/z)$ and $h_{11}(z)=h_{33}(q_{0}^{2}/z)$ for $z\in\mathbb{D}^{+}$. The auxiliary eigenfunctions exhibit $\gamma(z)=-\widetilde{\gamma}(q_{0}^{2}/z)$ for $\operatorname{Im} z\leq0$. We introduce the reflections as follows for $z\in\mathbb{R} \backslash \{0,\pm q_{0}\}$,
\begin{subequations}\label{2.32}
\begin{align}
\beta_{1}(z)&=\frac{s_{13}(z)}{s_{11}(z)}=-\frac{h_{31}^{*}(z)}{h_{11}^{*}(z)}, \qquad \,\,\,\,\,\,
\beta_{1}(\frac{q_{0}^{2}}{z})=-\frac{s_{31}(z)}{s_{33}(z)}=\frac{h_{13}^{*}(z)}{h_{33}^{*}(z)}, \\
\beta_{2}(z)&=\frac{h_{21}(z)}{h_{11}(z)}=-\widehat{\rho}(z)\frac{s_{12}^{*}(z)}{s_{11}^{*}(z)}, \quad \beta_{2}(\frac{q_{0}^{2}}{z})=\frac{q_{0}}{\mathrm{i}z}\frac{h_{23}(z)}{h_{33}(z)}
=\frac{q_{0}\widehat{\rho}(z)}{\mathrm{i}z}\frac{s_{32}^{*}(z)}{s_{33}^{*}(z)}.
\end{align}
\end{subequations}

\section{Discrete spectrum and its related properties}
\label{s:Discrete spectrum and its related properties}

\subsection{Discrete spectrum}
\label{s:Discrete spectrum}

The discrete spectrum and distribution of the defocusing-defocusing coupled Hirota equations with the NZBCs~\eqref{1.3} were studied below, and two $3\times3$ matrices $\mathbf{\Psi}^{+}(z)=(\psi_{+,1}(z),-\widetilde{\gamma}(z),\psi_{-,3}(z))$ for $z\in \mathbb{D}^{+}$ and $\mathbf{\Psi}^{-}(z)=(\psi_{-,1}(z),\gamma(z),\psi_{+,3}(z))$ for $z\in \mathbb{D}^{-}$ were given based on reference~\cite{23}. By calculation, there are $\det \mathbf{\Psi}^{+}(z)=\mathrm{i}\mathrm{e}^{\mathrm{i}\delta_{2}(z)}h_{11}(z)s_{33}(z)
\widehat{\rho}(z)$ for $\operatorname{Im} z\geq0$ and $\det \mathbf{\Psi}^{-}(z)=\mathrm{i}\mathrm{e}^{\mathrm{i}\delta_{2}(z)}s_{11}(z)h_{33}(z)
\widehat{\rho}(z)$ for $\operatorname{Im} z\leq0$.

\begin{proposition}\label{pro:4}
Let $\mathbf{v}(z)$ denote a nontrivial solution to the scattering problem in~\eqref{2.2}. If $\mathbf{v}(z)\in L^{2}(\mathbb{R})$, then $z\in C_{0}$. Here $C_{0}$ is a circle with a radius of $q_{0}$ centered on the origin of the complex $z$-plane.
\end{proposition}

The discrete eigenvalues $z_{m}$ on the circle $C_{0}$ is $\{ z_{m},z_{m}^{*} \}$ and the discrete eigenvalues $\theta_{m}$ off circle $C_{0}$ is $\{\theta_{m},\theta_{m}^{*},q_{0}^{2}/\theta_{m},q_{0}^{2}/\theta_{m}^{*} \}$. Following up on the work cited in~\cite{23}, we proceed to present the subsequent propositions.

\begin{proposition}\label{pro:5} Suppose that $h_{11}(z)$ possesses a zero $\theta_{m}$ within the upper half plane of $z$. Then
$h_{11}(\theta_{m})=0 \Leftrightarrow s_{11}(\theta_{m}^{*})=0 \Leftrightarrow s_{33}(q_{0}^{2}/\theta_{m}^{*})=0 \Leftrightarrow h_{33}(q_{0}^{2}/\theta_{m})=0$.
\end{proposition}

\begin{proposition}\label{pro:6}
If $\operatorname{Im} \theta_{m}>0$ and $|\theta_{m}|\neq q_{0}$. Then $\widetilde{\gamma}(\theta_{m})\neq\mathbf{0}$ and the following conclusions are equivalent:
\begin{itemize} \itemsep0.75pt
\everymath{\displaystyle}
  \item[(1)] $\widetilde{\gamma}(\theta_{m})=\mathbf{0} \Leftrightarrow \gamma(\theta_{m}^{*})=\mathbf{0} \Leftrightarrow \gamma(\frac{q_{0}^{2}}{\theta_{m}})=\mathbf{0} \Leftrightarrow
      \widetilde{\gamma}(\frac{q_{0}^{2}}{\theta_{m}^{*}})=\mathbf{0}$.
  \item[(2)] $\psi_{-,3}(\theta_{m})$ and $\psi_{+,1}(\theta_{m})$ are linearly correlated. $\psi_{-,1}(\theta_{m}^{*})$ and $\psi_{+,3}(\theta_{m}^{*})$ are linearly correlated.
  \item[(3)] $\psi_{-,3}(\frac{q_{0}^{2}}{\theta_{m}^{*}})$ and $\psi_{+,1}(\frac{q_{0}^{2}}{\theta_{m}^{*}})$ are linearly correlated. $\psi_{-,1}(\frac{q_{0}^{2}}{\theta_{m}})$ and $\psi_{+,3}(\frac{q_{0}^{2}}{\theta_{m}})$ are linearly correlated.
\end{itemize}
\end{proposition}

\begin{proposition}\label{pro:7}
Let $z_{m}$ be a zero of $h_{11}(z)$ in the upper half plane with $|z_{m}|=q_{0}$. Then $\widetilde{\gamma}(z_{m})=\gamma(z_{m}^{*})=\mathbf{0}$, there exist constants $c_{m}$ and $\bar{c}_{m}$ such that $\psi_{+,1}(z_{m})=c_{m}\psi_{-,3}(z_{m})$ and $\psi_{+,3}(z_{m}^{*})=\bar{c}_{m}\psi_{-,1}(z_{m}^{*})$.
\end{proposition}

\begin{proposition}\label{pro:8}
Let $\theta_{m}$ be a zero of $h_{11}(z)$ in the upper half plane with $|\theta_{m}|\neq q_{0}$. Then $|\theta_{m}|<q_{0}$ and $s_{33}(\theta_{m})\neq 0$, there exist constants $f_{m}$, $\hat{f}_{m}$, $\check{f}_{m}$ and $\bar{f}_{m}$ such that
\begin{subequations}\label{3.1}
\begin{align}
\psi_{+,1}(\theta_{m})&=\frac{f_{m}}{s_{33}(\theta_{m})}\widetilde{\gamma}(\theta_{m}), \quad
\gamma(\theta_{m}^{*})=\bar{f}_{m}\psi_{-,1}(\theta_{m}^{*}), \\
\psi_{+,3}(\frac{q_{0}^{2}}{\theta_{m}})&=\check{f}_{m}\gamma(\frac{q_{0}^{2}}{\theta_{m}}), \qquad \,\,\,\,\,  \widetilde{\gamma}(\frac{q_{0}^{2}}{\theta_{m}^{*}})
=\hat{f}_{m}\psi_{-,3}(\frac{q_{0}^{2}}{\theta_{m}^{*}}).
\end{align}
\end{subequations}
\end{proposition}

\begin{proposition}\label{pro:9}
Suppose that $h_{11}(z)$ has simple zeros $\{z_{m}\}_{m=1}^{M_{1}}$ on $C_{0}$, it can be inferred that the norming constants adhere to the symmetry relationship:
\begin{equation}\label{3.2}
\begin{split}
\bar{c}_{m}=-c_{m}, \quad c_{m}^{*}=\frac{s_{11}^{\prime}(z_{m}^{*})}{h_{33}^{\prime}(z_{m}^{*})}\bar{c}_{m},
\quad m=1,2, \ldots, M_{1}.
\end{split}
\end{equation}
\end{proposition}

\begin{proposition}\label{pro:10}
Suppose that $h_{11}(z)$ has zeros $\{\theta_{m}\}_{m=1}^{M_{2}}$ off $C_{0}$, it is known that the norming constants adhere to the symmetry relationship:
\begin{equation}\label{3.3}
\begin{split}
\bar{f}_{m}=-\frac{f_{m}^{*}}{\widehat{\rho}(\theta_{m}^{*})}, \quad \hat{f}_{m}=-\frac{\mathrm{i}\theta_{m}^{*}f_{m}^{*}}{q_{0}\widehat{\rho}(\theta_{m}^{*})}, \quad
\check{f}_{m}=\frac{q_{0}f_{m}}{\mathrm{i}\theta_{m}s_{33}(\theta_{m})},
\quad m=1,2, \ldots, M_{2}.
\end{split}
\end{equation}
\end{proposition}

\subsection{Asymptotic behavior}
\label{s:Asymptotic behavior}

The asymptotic properties of the modified Jost eigenfunctions as $z\rightarrow\infty$ and $z\rightarrow 0$ can be analyzed using the following expansions for $\nu_{\pm}(z)$:
\begin{equation}\label{3.4}
\begin{split}
\nu_{\pm}(z)=\nu_{0}^{\pm}+\frac{\nu_{1}^{\pm}}{z}+\frac{\nu_{2}^{\pm}}{z^{2}}+\cdots, \quad z\rightarrow \infty,  \quad \nu_{\pm}(z)=\nu_{0}^{\pm}+\nu_{1}^{\pm}z+\nu_{2}^{\pm}z^{2}+\cdots, \quad \, z\rightarrow 0,
\end{split}
\end{equation}
where $\nu_{\hbar}^{\pm}=\nu_{\hbar}^{\pm}(x,t)$. When substituting the Eq.~\eqref{3.4} into the differential Eqs.~\eqref{2.15} and combining them, the resulting expressions exhibit.
\begin{proposition}\label{pro:11}
The asymptotic expansion as $z\rightarrow\infty$ is delineated as follows:
\begin{equation}\label{3.5}
\everymath{\displaystyle}
\begin{split}
\nu_{\pm,1}(z)=\begin{pmatrix}
    \mathrm{i}\mp\frac{g_{1\pm}}{z}  \\
    \frac{\mathrm{i}\mathbf{q}}{z}
  \end{pmatrix}+O\left(\frac{1}{z^{2}}\right), \quad
\nu_{\pm,3}(z)=\begin{pmatrix}
    \frac{\mathbf{q}^{\dagger}\mathbf{q}_{\pm}}{q_{0}z}   \\
    \frac{\mathbf{q}_{\pm}}{q_{0}}\mp \frac{\mathrm{i}\mathbf{q}_{\pm}^{\perp}g_{2\pm}}{q_{0}z}
    \mp \frac{\mathrm{i}\mathbf{q}_{\pm}g_{3\pm}}{q_{0}z}
  \end{pmatrix}+O\left(\frac{1}{z^{2}}\right),
\end{split}
\end{equation}
where
\begin{subequations}\label{3.6}
\begin{align}
g_{1+}&=\int_{x}^{\infty} \frac{\left|\mathbf{q}^{\dagger}\mathbf{q}_{+}^{\perp}\right|^{2}
+\left|\mathbf{q}^{\dagger}\mathbf{q}_{+}\right|^{2}-q_{0}^{4}}{q_{0}^{2}}\,\mathrm{d}y, \quad
g_{1-}=\int_{-\infty}^{x} \frac{\left|\mathbf{q}^{\dagger}\mathbf{q}_{-}^{\perp}\right|^{2}
+\left|\mathbf{q}^{\dagger}\mathbf{q}_{-}\right|^{2}-q_{0}^{4}}{q_{0}^{2}}\,\mathrm{d}y, \\
g_{2+}&=\int_{x}^{\infty} \frac{\left[(\mathbf{q}_{+}^{\perp})^{\dagger}\mathbf{q}\right]
\left[\mathbf{q}^{\dagger}\mathbf{q}_{+}\right]}{q_{0}^{2}}\,\mathrm{d}y, \quad
g_{3+}=\int_{x}^{\infty}\frac{\left|\mathbf{q}^{\dagger}\mathbf{q}_{+}\right|^{2}
-q_{0}^{4}}{q_{0}^{2}}\,\mathrm{d}y, \\
g_{2-}&=\int_{-\infty}^{x}  \frac{\left[(\mathbf{q}_{-}^{\perp})^{\dagger}\mathbf{q}\right]
\left[\mathbf{q}^{\dagger}\mathbf{q}_{-}\right]}{q_{0}^{2}}\,\mathrm{d}y, \quad
g_{3-}=\int_{-\infty}^{x}\frac{\left|\mathbf{q}^{\dagger}\mathbf{q}_{-}\right|^{2}
-q_{0}^{4}}{q_{0}^{2}}\,\mathrm{d}y.
\end{align}
\end{subequations}
Similarly,
\begin{equation}\label{3.7}
\everymath{\displaystyle}
\begin{split}
\nu_{\pm,1}(z)=\begin{pmatrix}
    0   \\  -\frac{\mathrm{i}z}{q_{0}^{2}}\mathbf{q}_{\pm}
  \end{pmatrix}+O(z^{2}), \quad
\nu_{\pm,3}(z)=\begin{pmatrix}
    -\frac{z}{q_{0}}   \\  \mathbf{0}
  \end{pmatrix}+O(z^{2}), \quad z\rightarrow 0.
\end{split}
\end{equation}
\end{proposition}

In addition, the modified auxiliary eigenfunctions are defined as $d(z)=\gamma(z)\mathrm{e}^{-\mathrm{i}\delta_{2}(z)}$ for $z\in \mathbb{D}^{-}$ and $\widetilde{d}(z)=\widetilde{\gamma}\mathrm{e}^{-\mathrm{i}\delta_{2}}$ for $z\in \mathbb{D}^{+}$. The asymptotic properties of $d(z)$ and $\widetilde{d}(z)$ for $z\rightarrow\infty$ and $z\rightarrow 0$ are as follows:
\begin{proposition}\label{pro:12}
As $z\rightarrow\infty$ in the $z$-plane,
\begin{subequations}\label{3.8}
\everymath{\displaystyle}
\begin{align}
d(z)&=\begin{pmatrix}
    \frac{\mathbf{q}^{\dagger}\mathbf{q}_{-}^{\perp}}{q_{0}z}  \\
    \frac{\mathbf{q}_{-}^{\perp}}{q_{0}}+\frac{\mathrm{i}}{q_{0}z} \left[ \mathbf{q}_{-}g_{2-}^{*}-\mathbf{q}_{-}^{\perp}g_{1+}^{*}
    -\mathbf{q}_{-}^{\perp}g_{3-}^{*} \right]
  \end{pmatrix}+O\left(\frac{1}{z^{2}}\right), \\
\widetilde{d}(z)&=\begin{pmatrix}
    -\frac{\mathbf{q}^{\dagger}\mathbf{q}_{+}^{\perp}}{q_{0}z}  \\
    -\frac{\mathbf{q}_{+}^{\perp}}{q_{0}}+\frac{\mathrm{i}}{q_{0}z} \left[ \mathbf{q}_{+}g_{2+}^{*}-\mathbf{q}_{+}^{\perp}g_{1-}^{*}
    -\mathbf{q}_{+}^{\perp}g_{3+}^{*} \right]
  \end{pmatrix}+O\left(\frac{1}{z^{2}}\right),
\end{align}
\end{subequations}
and
\begin{equation}\label{3.9}
\everymath{\displaystyle}
\begin{split}
d(z)&=\begin{pmatrix}
    0  \\  \frac{\mathbf{q}_{+}^{\perp}}{q_{0}}
  \end{pmatrix}+O(z), \quad
\widetilde{d}(z)=\begin{pmatrix}
    0  \\  -\frac{\mathbf{q}_{-}^{\perp}}{q_{0}}
  \end{pmatrix}+O(z), \quad z\rightarrow 0.
\end{split}
\end{equation}
\end{proposition}

\begin{proposition}\label{pro:13}
The asymptotic behavior of scattering matrix entries as $z\rightarrow\infty$ is
\begin{subequations}\label{3.10}
\begin{align}
s_{32}(z)&=\frac{\mathbf{q}_{+}^{\dagger}\mathbf{q}_{-}^{\perp}}{q_{0}^{2}}+O\left(\frac{1}{z}\right), \quad
s_{22}(z)=h_{33}(z)=\frac{\mathbf{q}_{-}^{\dagger}\mathbf{q}_{+}}{q_{0}^{2}}+O\left(\frac{1}{z}\right), \quad
s_{11}(z)=1+O\left(\frac{1}{z}\right), \\
h_{32}(z)&=\frac{\mathbf{q}_{-}^{\dagger}\mathbf{q}_{+}^{\perp}}{q_{0}^{2}}+O\left(\frac{1}{z}\right), \quad
h_{22}(z)=s_{33}(z)=\frac{\mathbf{q}_{+}^{\dagger}\mathbf{q}_{-}}{q_{0}^{2}}+O\left(\frac{1}{z}\right), \quad
h_{11}(z)=1+O\left(\frac{1}{z}\right), \\
s_{23}(z)&=\frac{(\mathbf{q}_{+}^{\perp})^{\dagger}\mathbf{q}_{-}}{q_{0}^{2}}+O\left(\frac{1}{z}\right), \quad s_{21}(z)=s_{31}(z)=O\left(\frac{1}{z^{2}}\right), \quad
s_{12}(z)=s_{13}(z)=O\left(\frac{1}{z}\right), \\
h_{23}(z)&=\frac{(\mathbf{q}_{-}^{\perp})^{\dagger}\mathbf{q}_{+}}{q_{0}^{2}}+O\left(\frac{1}{z}\right), \quad
h_{21}(z)=h_{31}(z)=O\left(\frac{1}{z^{2}}\right), \quad
h_{12}(z)=h_{13}(z)=O\left(\frac{1}{z}\right).
\end{align}
\end{subequations}
Similarly, one can show that as $z\rightarrow 0$
\begin{subequations}\label{3.11}
\begin{align}
s_{11}(z)&=h_{22}(z)=\frac{\mathbf{q}_{+}^{\dagger}\mathbf{q}_{-}}{q_{0}^{2}}+O(z), \quad
s_{12}(z)=\frac{\mathrm{i}\mathbf{q}_{+}^{\dagger}\mathbf{q}_{-}^{\perp}}{q_{0}z}+O(1), \quad
s_{32}(z)=O(1),  \\
s_{22}(z)&=h_{11}(z)=\frac{\mathbf{q}_{-}^{\dagger}\mathbf{q}_{+}}{q_{0}^{2}}+O(z), \quad
h_{12}(z)=\frac{\mathrm{i}\mathbf{q}_{-}^{\dagger}\mathbf{q}_{+}^{\perp}}{q_{0}z}+O(1), \quad
h_{32}(z)=O(1), \\
s_{33}(z)&=1+O(z), \quad s_{21}(z)=s_{31}(z)=s_{13}(z)=s_{23}(z)=O(z), \\
h_{33}(z)&=1+O(z), \quad h_{21}(z)=h_{31}(z)=h_{13}(z)=h_{23}(z)=O(z).
\end{align}
\end{subequations}
\end{proposition}

\begin{proposition}\label{pro:14}
Assuming $\mathbf{q}_{+}^{\dagger}\mathbf{q}_{-}=0$ and setting $\mathbf{q}_{+}=\mathbf{q}_{-}^{\perp}\mathrm{e}^{-\mathrm{i}\vartheta}$ for $\vartheta\in \mathbb{R}$. Then, as $z\rightarrow\infty$ in the complex plane (different from the asymptotic behavior in Proposition~\ref{pro:13}),
\begin{equation}\label{3.12}
\begin{split}
s_{33}(z)&=\frac{\mathrm{i}\vartheta_{1}}{q_{0}^{2}z}+O\left(\frac{1}{z^{2}}\right), \quad
h_{33}(z)=-\frac{\mathrm{i}\vartheta_{1}^{*}}{q_{0}^{2}z}+O\left(\frac{1}{z^{2}}\right), \quad
s_{22}(z)=h_{22}(z)=O\left(\frac{1}{z}\right),
\end{split}
\end{equation}
and as $z\rightarrow0$ in the complex plane (different from the asymptotic behavior in Proposition~\ref{pro:13}),
\begin{equation}\label{3.13}
\begin{split}
s_{11}(z)&=\frac{\mathrm{i}\vartheta_{1}}{q_{0}^{4}}z+O(z^{2}), \quad h_{11}(z)=-\frac{\mathrm{i}\vartheta_{1}^{*}}{q_{0}^{4}}z+O(z^{2}), \quad s_{22}(z)=h_{22}(z)=O(z),
\end{split}
\end{equation}
where
\begin{equation}\label{3.14}
\begin{split}
\vartheta_{1}&=\frac{\mathbf{q}_{+}^{\dagger}\mathbf{q}_{-}^{\perp}}{q_{0}^{2}} \int_{\mathbb{R}} \left[ (\mathbf{q}_{-}^{\perp})^{\dagger} \mathbf{q}(x,t) \right]
\left[ \mathbf{q}^{\dagger}(x,t) \mathbf{q}_{-} \right]\,\mathrm{d}x
=\int_{\mathbb{R}} \left[ \mathbf{q}_{+}^{\dagger} \mathbf{q}(x,t) \right]
\left[ \mathbf{q}^{\dagger}(x,t) \mathbf{q}_{-} \right]\,\mathrm{d}x.
\end{split}
\end{equation}
\end{proposition}

\begin{proposition}\label{pro:15}
In the non-orthogonal case,
\begin{equation}\label{3.15}
\begin{split}
\beta_{1}(z)=O\left(\frac{1}{z}\right), \quad
\beta_{2}(z)=\beta_{1}\left(\frac{q_{0}^{2}}{z}\right)=O\left(\frac{1}{z^{2}}\right), \quad \beta_{2}\left(\frac{q_{0}^{2}}{z}\right)=\frac{q_{0}}{\mathrm{i}z}
\left[ \frac{(\mathbf{q}_{-}^{\perp})^{\dagger}\mathbf{q}_{+}}{\mathbf{q}_{-}^{\dagger}\mathbf{q}_{+}} \right] +O\left(\frac{1}{z^{2}}\right),
\quad z\rightarrow\infty.
\end{split}
\end{equation}
Similarly,
\begin{equation}\label{3.16}
\begin{split}
\beta_{1}(z)=\beta_{1}\left(\frac{q_{0}^{2}}{z}\right)=\beta_{2}\left(\frac{q_{0}^{2}}{z}\right)=O(z), \quad
\beta_{2}(z)=\frac{z}{\mathrm{i}q_{0}}
\left[ \frac{(\mathbf{q}_{-}^{\perp})^{\dagger}\mathbf{q}_{+}}{\mathbf{q}_{-}^{\dagger}\mathbf{q}_{+}} \right] +O(z^{2}), \quad z\rightarrow0.
\end{split}
\end{equation}
In the orthogonal case,
\begin{equation}\label{3.17}
\begin{split}
\beta_{1}(z)&=\beta_{1}\left(\frac{q_{0}^{2}}{z}\right)=O\left(\frac{1}{z}\right), \quad \beta_{2}(z)=O\left(\frac{1}{z^{2}}\right), \quad \beta_{2}\left(\frac{q_{0}^{2}}{z}\right)=\frac{q_{0}}{\vartheta_{1}^{*}} \left[(\mathbf{q}_{-}^{\perp})^{\dagger}\mathbf{q}_{+} \right] +O\left(\frac{1}{z}\right), \quad z\rightarrow\infty.
\end{split}
\end{equation}
Similarly,
\begin{equation}\label{3.18}
\begin{split}
\beta_{1}(z)=O(1), \quad \beta_{1}\left(\frac{q_{0}^{2}}{z}\right)=\beta_{2}\left(\frac{q_{0}^{2}}{z}\right)=O(z), \quad
\beta_{2}(z)=\frac{q_{0}}{\vartheta_{1}^{*}} \left[(\mathbf{q}_{-}^{\perp})^{\dagger}\mathbf{q}_{+} \right]+O(z), \quad z\rightarrow0.
\end{split}
\end{equation}
\end{proposition}

\begin{remark}\label{rem:5}
We can observe that the asymptotic behavior of the reflection coefficients differs from that in~\cite{23}, due to the different eigenvector matrices. Under the parallel case $\mathbf{q}_{+} \| \mathbf{q}_{-}$, the occurrence of $\beta_{2}(q_{0}^{2}/z)^{[\mathrm{bk}]}(z)=O(z)$ as $z\rightarrow\infty$ and  $\beta_{2}^{[\mathrm{bk}]}(z)=O(z^{-1})$ as $z\rightarrow0$ is allowed to hold. But in the non-parallel case, the corresponding sectionally meromorphic matrices involved in RH problem tend to diverge as $z\rightarrow0$ and $z\rightarrow\infty$, so it is necessary to construct a new eigenvector matrix.
\end{remark}

\subsection{Behavior at the branch points}
\label{s:Behavior at the branch points}

We will analyze the characteristics of $\nu_{\pm}(z;x,t)$ and scattering matrix at $z=\pm q_{0}$. Since $\det\widehat{\mathbf{Y}}_{\pm}^{-1}(\pm q_{0})=0$, the matrices $\widehat{\mathbf{Y}}_{\pm}^{-1}(z)$ are degenerate at $z=\pm q_{0}$. Due to $\lambda(\pm q_{0})=0$, the matrices $\widehat{\mathbf{Y}}_{\pm}(z)$ have a pole at $z=\pm q_{0}$ and the corresponding two exponential terms $\mathrm{e}^{\pm\mathrm{i}\lambda(\pm q_{0})x}=1$.

Introducing the weighted Sobolev spaces $L^{1,\hbar}(\mathbb{R}_{x}^{\pm})=\left\{f: \mathbb{R} \rightarrow \mathbb{C} \mid (1+|x|)^{\hbar} f \in L^{1}(\mathbb{R}_{x}^{\pm}) \right\}$. Based on~\eqref{2.16} and introducing $\mu_{\pm}(z)=\nu_{\pm}(z)\widehat{\mathbf{Y}}_{\pm}^{-1}(z)$, we obtain
\begin{equation}\label{3.19}
\begin{split}
\mu_{\pm}(z)=\mathbf{I}+\int_{\pm\infty}^{x}\mathbf{E}_{\pm}(z)
\left[\mathbf{X}(z;y,t)-\mathbf{X}_{\pm}(z)\right]\mu_{\pm}(z;y,t)\mathbf{E}_{\pm}^{-1}(z)
\,\mathrm{d}y,
\end{split}
\end{equation}
where
$\mathbf{E}_{\pm}(z)=\widehat{\mathbf{Y}}_{\pm}(z)\mathrm{e}^{\mathrm{i}(x-y)\mathbf{\Lambda}_{1}}
\widehat{\mathbf{Y}}_{\pm}^{-1}(z)$ and $\mathbf{E}_{\pm}^{-1}(z)=\widehat{\mathbf{Y}}_{\pm}(z)\mathrm{e}^{-\mathrm{i}(x-y)\mathbf{\Lambda}_{1}}
\widehat{\mathbf{Y}}_{\pm}^{-1}(z)$. Notice that
\begin{subequations}\label{3.20}
\everymath{\displaystyle}
\begin{align}
\lim_{z\rightarrow\pm q_{0}}\mathbf{E}_{\pm}(z)&=\begin{pmatrix}
    1\pm\mathrm{i}q_{0}(x-y) & -\mathrm{i}(x-y)\mathbf{q}_{\pm}^{\dagger}   \\
    \mathrm{i}(x-y)\mathbf{q}_{\pm}  & \frac{1}{q_{0}^{2}} \left[ [1\mp\mathrm{i}q_{0}(x-y)] \mathbf{q}_{\pm}\mathbf{q}_{\pm}^{\dagger}+\mathrm{e}^{\mp\mathrm{i}q_{0}(x-y)}
    \mathbf{q}_{\pm}^{\perp}(\mathbf{q}_{\pm}^{\perp})^{\dagger} \right]  \\
  \end{pmatrix}, \\
\lim_{z\rightarrow\pm q_{0}}\mathbf{E}_{\pm}^{-1}(z)&=\begin{pmatrix}
    1\mp\mathrm{i}q_{0}(x-y) & \mathrm{i}(x-y)\mathbf{q}_{\pm}^{\dagger}   \\
    -\mathrm{i}(x-y)\mathbf{q}_{\pm}  & \frac{1}{q_{0}^{2}} \left[ [1\pm\mathrm{i}q_{0}(x-y)] \mathbf{q}_{\pm}\mathbf{q}_{\pm}^{\dagger}+\mathrm{e}^{\pm\mathrm{i}q_{0}(x-y)}
    \mathbf{q}_{\pm}^{\perp}(\mathbf{q}_{\pm}^{\perp})^{\dagger} \right]  \\
  \end{pmatrix}.
\end{align}
\end{subequations}
Thus, if $\mathbf{q}(x,t)-\mathbf{q}_{+} \in L^{1,1}(\mathbb{R}_{x}^{+})$, the Eqs.~\eqref{3.19} are convergent at $z=\pm q_{0}$. $\mu_{+}(z)$ is well-defined and continuous at $z=\pm q_{0}$, so that $\mu_{+}(z)=B_{\pm}(x,t)+o(1)$ for $z\rightarrow \pm q_{0}$, where $B_{\pm}(x,t)=\mu_{+}(\pm q_{0})$. Since $\mu_{+}(z)=\nu_{+}(z) \widehat{\mathbf{Y}}_{+}^{-1}(z)$ for $z\in \mathbb{R} \backslash \{0,\pm q_{0}\}$, we can obtain that
\begin{subequations}\label{3.21}
\begin{align}
\nu_{+,1}(z)&=\frac{\mathrm{i}}{2\lambda} \left[ z\mu_{+,1}(z)+q_{+,1}\mu_{+,2}(z)
+q_{+,2}\mu_{+,3}(z) \right], \quad  \nu_{+,2}(z)=\frac{q_{+,2}^{*}}{q_{0}}\mu_{+,2}(z)-\frac{q_{+,1}^{*}}{q_{0}}\mu_{+,3}(z),  \\
\nu_{+,3}(z)&=\widehat{\rho}(z) \left[ \frac{q_{0}}{z}\mu_{+,1}(z)
+\frac{q_{+,1}}{q_{0}}\mu_{+,2}(z)+\frac{q_{+,2}}{q_{0}}\mu_{+,3}(z) \right], \quad z\rightarrow\pm q_{0},
\end{align}
\end{subequations}
where $q_{\pm,\hbar}$ for $\hbar=1,2$ denotes the $\hbar$-th component of $\mathbf{q}_{\pm}$. Additionally, we can rewrite it as $\nu_{+,2}(z)=\check{B}_{\pm,2}(x,t)+o(1)$ for $z\rightarrow\pm q_{0}$, where $\check{B}_{\pm,\hbar}(x,t)$ for $\hbar=1,2,3$ are obtained in terms of $B_{\pm}(x,t)=\mu_{+}(\pm q_{0})$. Similarly,
\begin{equation}\label{3.22}
\begin{aligned}
\nu_{+,1}(z)=\frac{\check{B}_{\pm,1}(x,t)}{z\mp q_{0}}+o\left( \frac{1}{z\mp q_{0}} \right),  \quad  \nu_{+,3}(z)=\frac{\check{B}_{\pm,3}(x,t)}{z\mp q_{0}}
+o\left( \frac{1}{z\mp q_{0}} \right), \quad  z\rightarrow\pm q_{0}.
\end{aligned}
\end{equation}
Assuming $\mathbf{q}(x,t)-\mathbf{q}_{-}\in L^{1,1}(\mathbb{R}_{x}^{-})$, we can obtain $\nu_{-,2}(z)=\hat{B}_{\pm,2}(x,t)+o(1)$ for $z\rightarrow\pm q_{0}$ and
\begin{equation}\label{3.23}
\begin{aligned}
\nu_{-,1}(z)=\frac{\hat{B}_{\pm,1}(x,t)}{z\mp q_{0}}+o\left(\frac{1}{z\mp q_{0}}\right),  \quad
\nu_{-,3}(z)=\frac{\hat{B}_{\pm,3}(x,t)}{z\mp q_{0}}+o\left(\frac{1}{z\mp q_{0}}\right), \quad z\rightarrow\pm q_{0}.
\end{aligned}
\end{equation}

By assuming that $\mathbf{q}(x,t)-\mathbf{q}_{\pm}$ decays faster at $x\rightarrow\pm\infty$, the asymptotic behavior of the characteristic function at the branch point is studied. Next, we differentiate the integral Eqs.~\eqref{3.19} to obtain
\begin{equation}\label{3.24}
\begin{aligned}
\frac{\partial\mu_{\pm}(z)}{\partial z}&=\int_{\pm\infty}^{x} \frac{\partial\mathbf{E}_{\pm}(z)}{\partial z}
\left[\mathbf{X}(z;y,t)-\mathbf{X}_{\pm}(z)\right]\mu_{\pm}(z;y,t)\mathbf{E}_{\pm}^{-1}(z)
\,\mathrm{d}y \\
&+\int_{\pm\infty}^{x}\mathbf{E}_{\pm}(z) \left[\mathbf{X}(z;y,t)-\mathbf{X}_{\pm}(z)\right]
\frac{\partial \mu_{\pm}(z;y,t)}{\partial z} \mathbf{E}_{\pm}^{-1}(z)\,\mathrm{d}y \\
&+\int_{\pm\infty}^{x}\mathbf{E}_{\pm}(z) \left[\mathbf{X}(z;y,t)-\mathbf{X}_{\pm}(z)\right]
\mu_{\pm}(z;y,t)\frac{\partial \mathbf{E}_{\pm}^{-1}(z)}{\partial z}\,\mathrm{d}y,
\end{aligned}
\end{equation}
and $\lim_{z\rightarrow \pm q_{0}} \partial \mathbf{E}_{\pm}(z)/\partial z
=\lim_{z\rightarrow \pm q_{0}} \partial \mathbf{E}_{\pm}^{-1}(z)/\partial z=\mathbf{0}_{3\times3}$. If $\mathbf{q}(x,t)-\mathbf{q}_{+} \in L^{1,2}\left(\mathbb{R}_{x}^{+}\right)$, $\partial\mu_{+}(z)/\partial z$ is continuous as $z\rightarrow\pm q_{0}$ from the real $z$-axis. So the Taylor expansion of $\mu_{+}(z)$ at $z=\pm q_{0}$ yields:
\begin{equation}\label{3.25}
\begin{aligned}
\mu_{+}(z)=\mu_{+}(\pm q_{0})+\left.\frac{\partial \mu_{+}(z)}{\partial z}\right|_{z=\pm q_{0}}(z\mp q_{0})
+\left.\frac{\partial^{2} \mu_{+}(z)}{\partial z^{2}}\right|_{z=\pm q_{0}}(z\mp q_{0})^{2}+\cdots.
\end{aligned}
\end{equation}
Therefore, $\mu_{+}(z)$ at the branch points can be written as:
\begin{equation}\label{3.26}
\begin{aligned}
\mu_{+}(z)=B_{\pm}(x,t)+A_{\pm}(x,t)(z\mp q_{0})+o(z\mp q_{0}), \quad z \rightarrow \pm q_{0},
\end{aligned}
\end{equation}
where $B_{\pm}(x,t)=\mu_{+}(\pm q_{0})$ and $A_{\pm}(x,t)=\left. \partial \mu_{+}(z)/\partial z \right|_{z=\pm q_{0}}$. Since $\mu_{+}(z)=\nu_{+}(z)\widehat{\mathbf{Y}}_{+}^{-1}(z)$ for $z\in \mathbb{R} \backslash \{0,\pm q_{0}\}$, as $z\rightarrow\pm q_{0}$ we have
\begin{subequations}\label{3.27}
\begin{align}
\nu_{+,1}(z)&=\frac{\check{B}_{\pm,1}}{z\mp q_{0}}+\check{A}_{\pm,1}+o(1), \quad
\nu_{+,3}(z)=\frac{\check{B}_{\pm,3}}{z\mp q_{0}}+\check{A}_{\pm,3}+o(1),  \\
\nu_{+,2}(z)&=\check{B}_{\pm,2}+\check{A}_{\pm,2}(z\mp q_{0})+o(z\mp q_{0}),
\end{align}
\end{subequations}
where $\check{B}_{\pm,\hbar}=\check{B}_{\pm,\hbar}(x,t)$ and $\check{A}_{\pm,\hbar}=\check{A}_{\pm,\hbar}(x,t)$, $\check{A}_{\pm,\hbar}(x,t)$ are obtained in terms of $A_{\pm}(x,t)=\mu_{+}^{\prime}(\pm q_{0})$. If $\mathbf{q}(x,t)-\mathbf{q}_{-} \in L^{1,2}(\mathbb{R}_{x}^{-})$, then $\nu_{-}(z)$ at the branch points can be written as:
\begin{subequations}\label{3.28}
\begin{align}
\nu_{-,1}(z)&=\frac{\hat{B}_{\pm,1}}{z\mp q_{0}}+\hat{A}_{\pm,1}+o(1), \quad
\nu_{-,3}(z)=\frac{\hat{B}_{\pm,3}}{z\mp q_{0}} +\hat{A}_{\pm,3}+o(1), \quad &z\rightarrow\pm q_{0}, \\
\nu_{-,2}(z)&=\hat{B}_{\pm,2}+\hat{A}_{\pm,2}(z\mp q_{0})+o(z\mp q_{0}), \quad &z\rightarrow\pm q_{0},
\end{align}
\end{subequations}
where $\hat{B}_{\pm,\hbar}=\hat{B}_{\pm,\hbar}(x,t)$ and $\hat{A}_{\pm,\hbar}=\hat{A}_{\pm,\hbar}(x,t)$ for $\hbar=1,2,3$.

Next, we analyzed the asymptotic behavior of $d(z)$ and $\widetilde{d}(z)$ at $z=\pm q_{0}$ as follows:
\begin{equation}\label{3.29}
\begin{aligned}
d(z)=-\frac{\mathrm{i}\mathbf{J}\left[ \nu_{-,3}^{*}(z^{*}) \times \nu_{+,1}^{*}(z^{*}) \right]}{\widehat{\rho}(z)}, \quad
\widetilde{d}(z)=-\frac{\mathrm{i}\mathbf{J}\left[ \nu_{-,1}^{*}(z^{*}) \times \nu_{+,3}^{*}(z^{*}) \right]}{\widehat{\rho}(z)}.
\end{aligned}
\end{equation}
If $\mathbf{q}(x,t)-\mathbf{q}_{\pm} \in L^{1,1}\left(\mathbb{R}_{x}^{\pm}\right)$, as $z\rightarrow\pm q_{0}$ we have
\begin{equation}\label{3.30}
\begin{aligned}
\left[ \nu_{+,3}^{*}(z) \times \nu_{-,1}^{*}(z) \right]=\frac{ \left[ \check{B}_{\pm,3}^{*} \times \hat{B}_{\pm,1}^{*} \right]+o(1)}{(z\mp q_{0})^{2}}, \quad
\left[ \nu_{+,1}^{*}(z) \times \nu_{-,3}^{*}(z) \right]=\frac{ \left[ \check{B}_{\pm,1}^{*} \times \hat{B}_{\pm,3}^{*} \right]+o(1)}{(z\mp q_{0})^{2}}.
\end{aligned}
\end{equation}
Correspondingly, there are the following behaviors at the branch points:
\begin{equation}\label{3.31}
\begin{aligned}
d(z)=\pm \frac{2\mathrm{i}d_{\pm}^{(-1)}}{q_{0}(z\mp q_{0})}+o\left(\frac{1}{z\mp q_{0}}\right), \quad
\widetilde{d}(z)=\pm \frac{2\mathrm{i}\widetilde{d}_{\pm}^{(-1)}}{q_{0}(z\mp q_{0})}+o\left(\frac{1}{z\mp q_{0}}\right), \quad z\rightarrow\pm q_{0},
\end{aligned}
\end{equation}
where $d_{\pm}^{(-1)}=\mathbf{J}\left[ \check{B}_{\pm,1}^{*} \times \hat{B}_{\pm,3}^{*} \right]$ and
$\widetilde{d}_{\pm}^{(-1)}=\mathbf{J}\left[ \check{B}_{\pm,3}^{*} \times \hat{B}_{\pm,1}^{*} \right]$. If $\mathbf{q}(x,t)-\mathbf{q}_{\pm} \in L^{1,2}(\mathbb{R}_{x}^{\pm})$, then we have
\begin{subequations}\label{3.32}
\begin{align}
\mathbf{J}\left[ \nu_{+,3}^{*}(z) \times \nu_{-,1}^{*}(z) \right]
&=\frac{\widetilde{d}_{\pm}^{(-1)}+\widetilde{d}_{\pm}^{(0)}(z\mp q_{0})+o(z\mp q_{0})}{(z\mp q_{0})^{2}}, \quad &z\rightarrow\pm q_{0}, \\
\mathbf{J}\left[ \nu_{+,1}^{*}(z) \times \nu_{-,3}^{*}(z) \right]
&=\frac{d_{\pm}^{(-1)}+d_{\pm}^{(0)}(z\mp q_{0})+o(z\mp q_{0})}{(z\mp q_{0})^{2}}, \quad &z\rightarrow\pm q_{0},
\end{align}
\end{subequations}
where $d_{\pm}^{(0)}=\mathbf{J}\left[ \check{B}_{\pm,1}^{*} \times \hat{A}_{\pm,3}^{*} \right]+\mathbf{J}\left[ \check{A}_{\pm,1}^{*} \times \hat{B}_{\pm,3}^{*} \right]$ and $\widetilde{d}_{\pm}^{(0)}=\mathbf{J}\left[ \check{B}_{\pm,3}^{*} \times \hat{A}_{\pm,1}^{*} \right]+\mathbf{J}\left[ \check{A}_{\pm,3}^{*} \times \hat{B}_{\pm,1}^{*} \right]$. Correspondingly, there are the following behaviors at the branch points:
\begin{equation}\label{3.33}
\begin{aligned}
d(z)=\pm \frac{2\mathrm{i}d_{\pm}^{(-1)}}{q_{0}(z\mp q_{0})} \pm \frac{2\mathrm{i}}{q_{0}}d_{\pm}^{(0)}+o(1), \quad
\widetilde{d}(z)=\pm \frac{2\mathrm{i}\widetilde{d}_{\pm}^{(-1)}}{q_{0}(z\mp q_{0})}
\pm \frac{2\mathrm{i}}{q_{0}}\widetilde{d}_{\pm}^{(0)}+o(1), \quad z\rightarrow\pm q_{0}.
\end{aligned}
\end{equation}
The column vectors of the eigenvector matrices have the following symmetric relationship at branch points: $\lim_{z\rightarrow q_{0}}\widehat{\mathbf{Y}}_{\pm,1}(z)=\mathrm{i}\lim_{z\rightarrow q_{0}}\widehat{\mathbf{Y}}_{\pm,3}(z)$ and $\lim_{z\rightarrow -q_{0}}\widehat{\mathbf{Y}}_{\pm,1}(z)=-\mathrm{i}\lim_{z\rightarrow -q_{0}}\widehat{\mathbf{Y}}_{\pm,3}(z)$.

The asymptotic behavior of $\psi_{\pm}(z)$ at $z=\pm q_{0}$ can be observed as $\psi_{\pm,1}(q_{0})=-\mathrm{i}\psi_{\pm,3}(q_{0})$ and $\psi_{\pm,1}(-q_{0})=\mathrm{i}\psi_{\pm,3}(-q_{0})$. We examine the characteristics of $\mathbf{S}(z)$ in the vicinity of the branch points, which are expressible through Wronskian determinants:
\begin{equation}\label{3.34}
\begin{split}
s_{jl}(z)=\frac{W_{jl}(z)}{\mathrm{i}\widehat{\rho}(z)}\mathrm{e}^{-\mathrm{i}\delta_{2}(z)}
=\frac{z^{2}-q_{0}^{2}}{\mathrm{i}z^{2}}W_{jl}(z)\mathrm{e}^{-\mathrm{i}\delta_{2}(z)}
,
\end{split}
\end{equation}
where $W_{jl}(z)=\det\,(\psi_{-,l}(z),\psi_{+,j+1}(z),\psi_{+,j+2}(z))$. If $\mathbf{q}(x,t)-\mathbf{q}_{\pm} \in L^{1,1}(\mathbb{R}_{x}^{\pm})$, then we have
\begin{subequations}\label{3.35}
\begin{align}
W_{11}(z)&=\frac{\mathrm{e}^{\mathrm{i}\delta_{2}(z)}}{(z\mp q_{0})^{2}} \left[ \det \left( \hat{B}_{\pm,1}, \check{B}_{\pm,2}, \check{B}_{\pm,3} \right) +o(1) \right], \quad &z\rightarrow\pm q_{0}, \\
W_{12}(z)&=\frac{\mathrm{e}^{2\mathrm{i}\delta_{2}(z)-\mathrm{i}\delta_{1}(z)}}{z\mp q_{0}} \left[ \det \left( \hat{B}_{\pm,2}, \check{B}_{\pm,2}, \check{B}_{\pm,3} \right) +o(1) \right], \quad &z\rightarrow\pm q_{0},
\end{align}
\end{subequations}
from which it follows that
\begin{equation}\label{3.36}
\begin{split}
s_{11}(z)=\mp\frac{2\mathrm{i}W_{11,\pm}^{(-1)}}{q_{0}(z\mp q_{0})}
+o\left(\frac{1}{z\mp q_{0}}\right), \quad
s_{12}(z)=\mp\frac{2\mathrm{i}\mathrm{e}^{\mathrm{i}\delta_{2}(\pm q_{0})-\mathrm{i}\delta_{1}(\pm q_{0})}}{q_{0}}W_{12,\pm}^{(0)}+o(1), \quad  z\rightarrow\pm q_{0},
\end{split}
\end{equation}
where $W_{11,\pm}^{(-1)}=\det \left( \hat{B}_{\pm,1}, \check{B}_{\pm,2}, \check{B}_{\pm,3} \right)$ and
$W_{12,\pm}^{(0)}=\det \left( \hat{B}_{\pm,2}, \check{B}_{\pm,2}, \check{B}_{\pm,3} \right)$.

Similarly, assuming $\mathbf{q}(x,t)-\mathbf{q}_{\pm} \in L^{1,2}(\mathbb{R}_{x}^{\pm})$ then we have
\begin{subequations}\label{3.37}
\begin{align}
W_{11}(z)&=\frac{\mathrm{e}^{\mathrm{i}\delta_{2}(z)}}{(z\mp q_{0})^{2}} \left[ W_{11,\pm}^{(-1)}+W_{11,\pm}^{(0)}(z\mp q_{0})+o(z\mp q_{0}) \right], \quad &z\rightarrow\pm q_{0}, \\
W_{12}(z)&=\frac{\mathrm{e}^{2\mathrm{i}\delta_{2}(z)-\mathrm{i}\delta_{1}(z)}}{z\mp q_{0}} \left[ W_{12,\pm}^{(0)}+W_{12,\pm}^{(1)}(z\mp q_{0})+o(z\mp q_{0}) \right], \quad &z\rightarrow\pm q_{0},
\end{align}
\end{subequations}
where
\begin{subequations}\label{3.38}
\begin{align}
W_{11,\pm}^{(0)}&=\det \left( \hat{A}_{\pm,1}, \check{B}_{\pm,2},\check{B}_{\pm,3} \right)+ \det \left( \hat{B}_{\pm,1}, \check{A}_{\pm,2},\check{B}_{\pm,3} \right)
+\det \left( \hat{B}_{\pm,1}, \check{B}_{\pm,2},\check{A}_{\pm,3} \right), \\
W_{12,\pm}^{(1)}&=\det \left( \hat{A}_{\pm,2}, \check{B}_{\pm,2},\check{B}_{\pm,3} \right) + \det \left( \hat{B}_{\pm,2}, \check{A}_{\pm,2},\check{B}_{\pm,3} \right)
+\det \left( \hat{B}_{\pm,2}, \check{B}_{\pm,2},\check{A}_{\pm,3} \right),
\end{align}
\end{subequations}
from which it follows that
\begin{subequations}\label{3.39}
\begin{align}
s_{11}(z)&=\mp\frac{2\mathrm{i}}{q_{0}(z\mp q_{0})}W_{11,\pm}^{(-1)}\mp\frac{2\mathrm{i}}{q_{0}}W_{11,\pm}^{(0)}+o(1), \quad &z\rightarrow\pm q_{0},  \\
s_{12}(z)&=\mp\frac{2\mathrm{i}\mathrm{e}^{\mathrm{i}\delta_{2}(\pm q_{0})-\mathrm{i}\delta_{1}(\pm q_{0})}}{q_{0}} \left[ W_{12,\pm}^{(0)}+W_{12,\pm}^{(1)}(z\mp q_{0}) \right] +o(z\mp q_{0}), \quad &z\rightarrow\pm q_{0}.
\end{align}
\end{subequations}

Furthermore, the asymptotic behavior for $\mathbf{S}(z)$ in the vicinity of $z=\pm q_{0}$ can be articulated in the following manner:
\begin{equation}\label{3.40}
\begin{split}
\mathbf{S}(z)=\mp\frac{2\mathrm{i}}{q_{0}(z\mp q_{0})}\mathbf{S}_{\pm}^{(-1)}
\mp \frac{2\mathrm{i}\mathrm{e}^{\mathrm{i}\delta_{2}(\pm q_{0})-\mathrm{i}\delta_{1}(\pm q_{0})}}{q_{0}} \mathbf{S}_{\pm}^{(0)}+o(1),
\end{split}
\end{equation}
where
\begin{equation}\label{3.41}
\begin{split}
\mathbf{S}_{\pm}^{(-1)}=W_{11,\pm}^{(-1)}\begin{pmatrix}
    1 & 0 & \pm\mathrm{i}  \\
    0 & 0 & 0  \\
    \pm\mathrm{i} & 0 & -1
  \end{pmatrix}, \quad
\mathbf{S}_{\pm}^{(0)}=W_{12,\pm}^{(0)}\begin{pmatrix}
    0 & 1 & 0  \\
    0 & 0 & 0  \\
    0 & \pm\mathrm{i} & 0
  \end{pmatrix}.
\end{split}
\end{equation}
If $\mathbf{q}(x,t)-\mathbf{q}_{\pm} \in L^{1,1}(\mathbb{R}_{x}^{ \pm})$, we know that $\lim_{z \rightarrow \pm q_{0}}\beta_{1}(z)=\pm\mathrm{i}$ and $\lim_{z \rightarrow \pm q_{0}}\beta_{2}(z)=0$.

\subsection{Non-existence of reflectionless potentials}
\label{s:Non-existence of reflectionless potentials}

One of the important differences between the equations~\eqref{1.3} with non-parallel NZBCs and the equations~\eqref{1.3} with parallel NZBCs~\cite{23} is that in the non-parallel case, there is no reflectionless potentials when the vectors $\mathbf{q}_{\pm}$ are not parallel. Therefore, the following theorem is given.
\begin{theorem}\label{thm:3}
There are no pure soliton solutions $\mathbf{q}(x,t)$ of the defocusing-defocusing coupled Hirota equations~\eqref{1.3} satisfying~\eqref{1.4} with $\mathbf{q}_{+}\nparallel\mathbf{q}_{-}$ and $\mathbf{q}(x,t)-\mathbf{q}_{\pm} \in L^{1}(\mathbb{R})$  for which $\beta_{1}(z)=\beta_{2}(z)=0$ for all $z\in \mathbb{R}$, where the definitions of $\beta_{1}(z)$ and $\beta_{2}(z)$ are shown in Eqs.~\eqref{2.32}.
\end{theorem}

\begin{proof}
Assuming $\beta_{1}(z)=\beta_{2}(z)=0$ for $z\in \mathbb{R}$, we consider the relationship:
\begin{equation}\label{3.42}
\begin{split}
\psi_{-,3}(z)=\frac{\psi_{+,3}(z)}{h_{33}(z)}, \quad \Rightarrow \quad
\nu_{-,3}(z)=\frac{\nu_{+,3}(z)}{h_{33}(z)}.
\end{split}
\end{equation}
It is necessary to separately prove the orthogonal and non-orthogonal cases. Consider the asymptotic behavior as $z\rightarrow \infty$ in the non-orthogonal case, we have
\begin{equation}\label{3.43}
\everymath{\displaystyle}
\begin{split}
\begin{pmatrix}
    0   \\
    \frac{\mathbf{q}_{-}}{q_{0}}
  \end{pmatrix}=\begin{pmatrix}
    0   \\
    \frac{q_{0}\mathbf{q}_{+}}{\mathbf{q}_{-}^{\dagger}\mathbf{q}_{+}}
  \end{pmatrix},  \quad \Rightarrow \quad
\mathbf{q}_{-}=\frac{q_{0}^{2}\mathbf{q}_{+}}{\mathbf{q}_{-}^{\dagger}\mathbf{q}_{+}},
\end{split}
\end{equation}
this contradicts the definition in the non-parallel case, therefore it is impossible. The above method cannot prove the orthogonality, since the quantity $\mathbf{q}_{+}^{\dagger}\mathbf{q}_{-}=0$. Therefore, we assume again that $\beta_{1}(z)=\beta_{2}(z)=0$ for $z\in \mathbb{R}$. Proposition~\ref{pro:14} implies that $\lim_{z\rightarrow\infty}zh_{33}(z)\neq0$, then we can obtain
\begin{equation}\label{3.44}
\everymath{\displaystyle}
\begin{split}
\frac{\nu_{-,3}(z)}{z}=\frac{\nu_{+,3}(z)}{zh_{33}(z)}, \quad \Rightarrow \quad
\begin{pmatrix}
    0   \\
    \frac{\mathbf{q}_{-}}{zq_{0}}
  \end{pmatrix}=\begin{pmatrix}
    0   \\
    \frac{\mathrm{i}q_{0}\mathbf{q}_{+}}{\vartheta_{1}^{*}}
  \end{pmatrix}.
\end{split}
\end{equation}
Taking the limit $z\rightarrow\infty$ to obtain $\mathbf{q}_{+}=0$, this contradicts the assumption that $\left\|\mathbf{q}_{+}\right\|=q_{0}>0$.
\end{proof}

\section{Inverse problem}
\label{s:Inverse problem}

\subsection{Riemann-Hilbert problem}
\label{s:Riemann-Hilbert problem}

To formulate the matrix RH problem, it is essential to establish suitable transition conditions that define the behavior of eigenfunctions, which are characterized by their meromorphic nature within the specified domain. Given that certain Jost eigenfunctions lack analytic properties, it is necessary to define new modified meromorphic functions in the corresponding regions.
\begin{lemma}\label{lem:1}
Define the piecewise meromorphic matrices $\mathbf{R}^{\pm}(z)=\mathbf{R}^{\pm}(z;x,t)=(\mathbf{R}_{1}^{\pm},\mathbf{R}_{2}^{\pm},\mathbf{R}_{3}^{\pm})$,
\begin{subequations}\label{4.1}
\begin{align}
\mathbf{R}^{+}(z)&=\mathbf{\Psi}^{+}(z)\mathrm{e}^{-\mathrm{i}\mathbf{\Delta}(z)}
\operatorname{diag} \left( \frac{1}{h_{11}} ,\frac{1}{s_{33}}, \frac{1}{\mathrm{i}\widehat{\rho}} \right)=\left[ \frac{\nu_{+,1}(z)}{h_{11}}, -\frac{\widetilde{d}(z)}{s_{33}}, \frac{\nu_{-,3}(z)}{\mathrm{i}\widehat{\rho}} \right],  &z\in \mathbb{D}^{+}, \\
\mathbf{R}^{-}(z)&=\mathbf{\Psi}^{-}(z)\mathrm{e}^{-\mathrm{i}\mathbf{\Delta}(z)}
\operatorname{diag} \left( \frac{1}{\mathrm{i}\widehat{\rho}}\, , \frac{1}{s_{11}},\frac{1}{h_{33}} \right)= \left[ \frac{\nu_{-,1}(z)}{\mathrm{i}\widehat{\rho}}, \frac{d(z)}{s_{11}}, \frac{\nu_{+,3}(z)}{h_{33}} \right],  &z\in \mathbb{D}^{-},
\end{align}
\end{subequations}
where $\widehat{\rho}=\widehat{\rho}(z)$, $s_{11}=s_{11}(z)$, $h_{11}=h_{11}(z)$, $s_{33}=s_{33}(z)$ and $h_{33}=h_{33}(z)$. The corresponding jump condition is
\begin{equation}\label{4.2}
\begin{split}
\mathbf{R}^{+}(z)=\mathbf{R}^{-}(z)\mathrm{e}^{\mathrm{i}\mathbf{\Delta}(z)}
\mathbf{K}(z)\mathrm{e}^{-\mathrm{i}\mathbf{\Delta}(z)}, \quad z\in \mathbb{R} \backslash \{0,\pm q_{0}\},
\end{split}
\end{equation}
with
\begin{equation}\label{4.3}
\everymath{\displaystyle}
\begin{split}
\mathbf{K}(z)=\begin{pmatrix}
    \left[\mathrm{i}\widehat{\rho}\left( 1+\beta_{1}^{*}\widehat{\beta}_{1}^{*} \right)
    -\mathrm{i}\left|\beta_{2}\right|^{2}
    +\frac{z\beta_{1}^{*}\beta_{2}^{*}\widehat{\beta}_{2}}{q_{0}}\right] &
   \left[ \frac{z\widehat{\beta}_{1}^{*}\widehat{\beta}_{2}^{*}}{q_{0}\widehat{\rho}}
    -\mathrm{i}\beta_{2}^{*}
    -\frac{\mathrm{i}z^{2}\beta_{2}^{*}\left|\widehat{\beta}_{2}\right|^{2}}{q_{0}^{2}\widehat{\rho}} \right] & \left[ \frac{\mathrm{i}z\widehat{\beta}_{2}\beta_{2}^{*}}{q_{0}\widehat{\rho}}
    -\widehat{\beta}_{1}^{*}\right]  \\
    \beta_{2}+\frac{\mathrm{i}z}{q_{0}}\beta_{1}^{*}\widehat{\beta}_{2}  & 1+\frac{z^{2}}{q_{0}^{2}\widehat{\rho}} \left|\widehat{\beta}_{2}\right|^{2} & -\frac{z}{q_{0}\widehat{\rho}}\widehat{\beta}_{2}  \\
    -\beta_{1}^{*} & \frac{\mathrm{i}z}{q_{0}\widehat{\rho}}\widehat{\beta}_{2}^{*} & \frac{1}{\mathrm{i}\widehat{\rho}}
  \end{pmatrix},
\end{split}
\end{equation}
where $\beta_{\hbar}=\beta_{\hbar}(z)$ and $\widehat{\beta}_{\hbar}=\beta_{\hbar}(q_{0}^{2}/z)$ for $\hbar=1,2$.
\end{lemma}

By Proposition~\ref{pro:15}, it is known that $\mathbf{K}(z)=O(1)$ as $z\rightarrow \infty$. Since $\widehat{\rho}(z)=O(z^{2})$ as $z\rightarrow 0$, it can be inferred that $\mathbf{K}(z)$ has a double pole as $z\rightarrow 0$. Moreover, $\mathbf{R}^{-}(z)\mathrm{e}^{\mathrm{i}\mathbf{\Delta}(z)}
\mathbf{K}(z)\mathrm{e}^{-\mathrm{i}\mathbf{\Delta}(z)}$ has a simple pole as $z\rightarrow 0$ and $\det\mathbf{R}(z)=1$ when $\operatorname{Im}z\neq0$.

\begin{lemma}\label{lem:2}
If $\mathbf{q}(x,t)-\mathbf{q}_{\pm} \in L^{1,2}(\mathbb{R}_{x}^{ \pm})$, the matrices $\mathbf{R}^{\pm}(z)$ and $\mathbf{K}(z)$ exhibit the following characteristics:
\begin{subequations}\label{4.4}
\begin{align}
\mathbf{R}^{\pm}(z)&=\mathbf{R}_{\infty}^{\pm}+O\left(\frac{1}{z}\right), \quad \mathbf{K}(z)=\mathbf{K}_{\infty}+O\left(\frac{1}{z}\right), \qquad \, z\rightarrow \infty, \\
\mathbf{R}^{\pm}(z)&=\frac{1}{z}\mathbf{R}_{0}^{\pm}+O(1), \quad
\mathbf{K}(z)=-\frac{q_{0}^{2}}{z^{2}}\mathbf{K}_{0}^{(-2)}+O(1), \quad z\rightarrow 0,
\end{align}
\end{subequations}
where
\begin{subequations}\label{4.5}
\everymath{\displaystyle}
\begin{align}
\mathbf{R}_{\infty}^{+}&=\begin{pmatrix}
    \mathrm{i} & 0 & 0 \\
    \mathbf{0_{2\times1}} & \frac{q_{0}\mathbf{q}_{+}^{\perp}}{\mathbf{q}_{+}^{\dagger}\mathbf{q}_{-}} & \frac{\mathbf{q}_{-}}{\mathrm{i}q_{0}}
  \end{pmatrix}, \quad
\mathbf{R}_{\infty}^{-}=\begin{pmatrix}
    1 & 0 & 0 \\
    \mathbf{0_{2\times1}} & \frac{\mathbf{q}_{-}^{\perp}}{q_{0}} & \frac{q_{0}\mathbf{q}_{+}}{\mathbf{q}_{-}^{\dagger}\mathbf{q}_{+}}
  \end{pmatrix},  \quad
\mathbf{K}_{\infty}=\begin{pmatrix}
    \mathrm{i} & 0 & 0 \\
    0 & 1+ \left| \frac{(\mathbf{q}_{-}^{\perp})^{\dagger}\mathbf{q}_{+}}{\mathbf{q}_{-}^{\dagger}\mathbf{q}_{+}} \right|^{2}  &
    \frac{\mathrm{i}(\mathbf{q}_{-}^{\perp})^{\dagger}\mathbf{q}_{+}}{\mathbf{q}_{-}^{\dagger} \mathbf{q}_{+}} \\
    0 & \frac{\mathbf{q}_{-}^{\dagger}\mathbf{q}_{+}^{\perp}}{\mathbf{q}_{+}^{\dagger} \mathbf{q}_{-}} & -\mathrm{i}
  \end{pmatrix}, \\
\mathbf{R}_{0}^{+}&=\begin{pmatrix}
   \mathbf{0_{1\times2}} & -\mathrm{i}q_{0} \\
   \mathbf{0_{2\times2}} & \mathbf{0_{2\times1}}
  \end{pmatrix},  \quad
\mathbf{R}_{0}^{-}=\begin{pmatrix}
   0 & \mathbf{0_{1\times2}} \\
   \mathbf{q}_{-} & \mathbf{0_{2\times2}}
  \end{pmatrix}, \quad
\mathbf{K}_{0}^{(-2)}=\begin{pmatrix}
    \mathbf{0_{2\times2}} & \mathbf{0_{1\times2}} \\
    \mathbf{0_{2\times1}} & \mathrm{i}
  \end{pmatrix}.
\end{align}
\end{subequations}
\end{lemma}

The definition of $\mathbf{R}^{\pm}(z)$ is different from that in the parallel case~\cite{23}. Although the Jost eigenfunctions have poles at the branch point $z=\pm q_{0}$, $\mathbf{R}^{\pm}(z)$ are finite within these limits. It can be found that
\begin{equation}\label{4.6}
\begin{split}
\lim_{z \rightarrow 0} \left[ \mathbf{R}^{+}(z)-\mathbf{R}^{-}(z) \right]
=\frac{1}{z} \left( \mathbf{R}_{0}^{+}-\mathbf{R}_{0}^{-} \right), \quad
\lim_{z \rightarrow \infty} \left[ \mathbf{R}^{+}(z)-\mathbf{R}^{-}(z) \right]
=\mathbf{R}_{\infty}^{+}-\mathbf{R}_{\infty}^{-}.
\end{split}
\end{equation}
Therefore, the non-parallel NZBCs situation is different from the parallel NZBCs situation~\cite{23}, namely $\mathbf{R}_{\infty}^{+}\neq\mathbf{R}_{\infty}^{-}$ (Pay attention to the parallel situation~\cite{23} where $\mathbf{R}_{\infty}^{+}=\mathbf{R}_{\infty}^{-}$). Subsequently, two different methods will be implemented to address the challenge posed by the varying asymptotic behaviors of $\mathbf{R}(z)$ as $z$ approaches infinity.

\subsubsection{First method}
\label{s:First method}

\begin{lemma}\label{lem:3}
The meromorphic matrices defined satisfy the following residue conditions:
\begin{subequations}\label{4.7}
\begin{align}
\underset{z=z_{m}}{\operatorname{Res}}\, \mathbf{R}^{+}(z)
&=C_{m} \left[ \nu_{-,3}(z_{m}), \mathbf{0},\mathbf{0} \right]
=\mathbf{R}^{+}(z_{m})\begin{pmatrix}
    \mathbf{0_{2\times1}} & \mathbf{0_{2\times2}} \\
    \mathrm{i}C_{m}\widehat{\rho}(z_{m}) & \mathbf{0_{1\times2}}
  \end{pmatrix}, \\
\underset{z=z_{m}^{*}}{\operatorname{Res}}\, \mathbf{R}^{-}(z)
&=\bar{C}_{m} \left[ \mathbf{0},\mathbf{0},\nu_{-,1}(z_{m}^{*}) \right]
=\mathbf{R}^{-}(z_{m}^{*})\begin{pmatrix}
    \mathbf{0_{1\times2}} & \mathrm{i}\bar{C}_{m}\widehat{\rho}(z_{m}^{*}) \\
    \mathbf{0_{2\times2}} & \mathbf{0_{2\times1}}
  \end{pmatrix}, \\
\underset{z=\theta_{m}}{\operatorname{Res}}\, \mathbf{R}^{+}(z)
&=F_{m} \left[ \frac{\widetilde{d}(\theta_{m})}{s_{33}(\theta_{m})},\mathbf{0},\mathbf{0} \right]
=\mathbf{R}^{+}(\theta_{m})\begin{pmatrix}
    0 & \mathbf{0_{1\times2}} \\
    -F_{m} & \mathbf{0_{1\times2}}\\
    0 & \mathbf{0_{1\times2}}
  \end{pmatrix}, \\
\underset{z=\theta_{m}^{*}}{\operatorname{Res}}\, \mathbf{R}^{-}(z)
&=\bar{F}_{m} \left[ \mathbf{0}, \nu_{-,1}(\theta_{m}^{*}), \mathbf{0} \right]
=\mathbf{R}^{-}(\theta_{m}^{*})\begin{pmatrix}
    0 & \mathrm{i}\bar{F}_{m} \widehat{\rho}(\theta_{m}^{*}) & 0 \\
    \mathbf{0_{2\times1}} & \mathbf{0_{2\times1}} & \mathbf{0_{2\times1}}
  \end{pmatrix}, \\
\underset{z=q_{0}^{2}/\theta_{m}}{\operatorname{Res}} \mathbf{R}^{-}(z)
&=\check{F}_{m}\left[ \mathbf{0},\mathbf{0}, \frac{d(q_{0}^{2}/\theta_{m})}{s_{11}(q_{0}^{2}/\theta_{m})} \right]
=\mathbf{R}^{-}(q_{0}^{2}/\theta_{m})\begin{pmatrix}
    \mathbf{0_{1\times2}} & 0 \\
    \mathbf{0_{1\times2}} & \check{F}_{m} \\
    \mathbf{0_{1\times2}} & 0
  \end{pmatrix},  \\
\underset{z=q_{0}^{2}/\theta_{m}^{*}}{\operatorname{Res}} \mathbf{R}^{+}(z)
&=-\hat{F}_{m}\left[ \mathbf{0}, \nu_{-,3}(q_{0}^{2}/\theta_{m}^{*}), \mathbf{0} \right]=\mathbf{R}^{+}(q_{0}^{2}/\theta_{m}^{*}) \begin{pmatrix}
    \mathbf{0_{2\times1}} & \mathbf{0_{2\times1}} & \mathbf{0_{2\times1}} \\
    0 & -\mathrm{i}\hat{F}_{m}\widehat{\rho}(q_{0}^{2}/\theta_{m}^{*}) & 0
  \end{pmatrix},
\end{align}
\end{subequations}
with norming constants
\begin{subequations}\label{4.8}
\begin{align}
C_{m}&=\frac{c_{m}\mathrm{e}^{-2\mathrm{i}\delta_{1}(z_{m})}}{h_{11}^{\prime}(z_{m})}, \quad
\bar{C}_{m}=\frac{\bar{c}_{m}\mathrm{e}^{2\mathrm{i}\delta_{1}(z_{m}^{*})}}{h_{33}^{\prime}(z_{m}^{*})}, \quad
\check{F}_{m}=\frac{\check{f}_{m}s_{11}(q_{0}^{2}/\theta_{m})}{h_{33}^{\prime}(q_{0}^{2}/\theta_{m})}
\mathrm{e}^{\mathrm{i}[\delta_{2}(\theta_{m})-\delta_{1}(\theta_{m})]}, \\
F_{m}&=\frac{f_{m}\mathrm{e}^{\mathrm{i}[\delta_{2}(\theta_{m})-\delta_{1}(\theta_{m})]}}{h_{11}^{\prime}(\theta_{m})},
\quad  \bar{F}_{m}=\frac{\bar{f}_{m}\mathrm{e}^{\mathrm{i}[\delta_{1}(\theta_{m}^{*})
-\delta_{2}(\theta_{m}^{*})]}}{s_{11}^{\prime}(\theta_{m}^{*})},  \quad
\hat{F}_{m}=\frac{\hat{f}_{m}\mathrm{e}^{\mathrm{i}[\delta_{1}(\theta_{m}^{*})
-\delta_{2}(\theta_{m}^{*})]}}{s_{33}^{\prime}(q_{0}^{2}/\theta_{m}^{*})},
\end{align}
\end{subequations}
where $m=1,\ldots,M_{1}$ for $z_{m}$, and $m=1,\ldots,M_{2}$ for $\theta_{m}$.
\end{lemma}

\begin{lemma}\label{lem:4}
The norming constants $C_{m}$, $\bar{C}_{m}$, $F_{m}$, $\bar{F}_{m}$, $\hat{F}_{m}$ and $\check{F}_{m}$ adhere to the subsequent symmetry properties:
\begin{equation}\label{4.9}
\begin{split}
C_{m}^{*}=\bar{C}_{m}=\mathrm{e}^{-2\mathrm{i}\arg(z_{m})}C_{m}, \quad
\check{F}_{m}=\frac{\mathrm{i}q_{0}^{3}}{\theta_{m}^{3}}F_{m}, \quad
\bar{F}_{m}=-\frac{F_{m}^{*}}{\widehat{\rho}(\theta_{m}^{*})}, \quad
\hat{F}_{m}=\frac{\mathrm{i}q_{0}}{\theta_{m}^{*}\widehat{\rho}(\theta_{m}^{*})}F_{m}^{*}.
\end{split}
\end{equation}
\end{lemma}

Normalize the RH problem by eliminating the dominant asymptotic terms and all contributions from the discrete spectrum associated with the poles. Initially, under the assumption of an absence of a discrete spectrum, we subtract $\mathbf{R}_{\infty}^{+}$ and $\mathbf{R}_{0}^{+}/z$ from both sides of the jump condition given by Eq.~\eqref{4.2}. By expressing  $\mathbf{K}(z)=\mathbf{I}+\mathbf{L}(z)$, we proceed as follows:
\begin{equation}\label{4.10}
\begin{split}
\mathbf{R}^{+}(z)-\mathbf{R}_{\infty}^{+}-\frac{1}{z}\mathbf{R}_{0}^{+}
&=\mathbf{R}^{-}(z)-\mathbf{R}_{\infty}^{-}-\frac{1}{z}\mathbf{R}_{0}^{-}
+\mathbf{R}_{\infty}^{\nabla}+\frac{1}{z}\mathbf{R}_{0}^{\nabla}
+\mathbf{R}^{-}(z)\mathrm{e}^{\mathrm{i}\mathbf{\Delta}(z)}\mathbf{L}(z)\mathrm{e}^{-\mathrm{i}\mathbf{\Delta}(z)},
\end{split}
\end{equation}
where $\mathbf{R}_{\infty}^{\nabla}=\mathbf{R}_{\infty}^{-}-\mathbf{R}_{\infty}^{+}$ and $\mathbf{R}_{0}^{\nabla}=\mathbf{R}_{0}^{-}-\mathbf{R}_{0}^{+}$. Utilize the Plemelj's formula to establish the ensuing theorem:

\begin{theorem}\label{thm:4}
For all $z\notin \mathbb{R}$, the solution to the RH problem can be expressed through the subsequent formula:
\begin{equation}\label{4.11}
\begin{split}
\mathbf{R}(z)&=\mathbf{R}_{\infty}(z)+\frac{1}{z}\mathbf{R}_{0}(z)+\frac{1}{2\pi\mathrm{i}}\int_{\mathbb{R}} \left[ \mathbf{R}_{\infty}^{\nabla}+\frac{1}{\zeta} \mathbf{R}_{0}^{\nabla}
+\mathbf{R}^{-}(\zeta)\mathrm{e}^{\mathrm{i}\mathbf{\Delta}(\zeta)}\mathbf{L}(\zeta) \mathrm{e}^{-\mathrm{i}\mathbf{\Delta}(\zeta)} \right] \frac{\mathrm{d}\zeta}{\zeta-z} \\
&+\sum_{m=1}^{M_{1}}\left[ \frac{\operatorname{Res}_{z=z_{m}}\mathbf{R}^{+}(z)}{z-z_{m}} +\frac{\operatorname{Res}_{z=z_{m}^{*}}\mathbf{R}^{-}(z)}{z-z_{m}^{*}} \right]
+\sum_{m=1}^{M_{2}}\left[ \frac{\operatorname{Res}_{z=\theta_{m}}\mathbf{R}^{+}(z)}{z-\theta_{m}} +\frac{\operatorname{Res}_{z=\theta_{m}^{*}}\mathbf{R}^{-}(z)}{z-\theta_{m}^{*}} \right] \\
&+\sum_{m=1}^{M_{2}} \left[ \frac{\operatorname{Res}_{z=q_{0}^{2}/\theta_{m}^{*}}\mathbf{R}^{+}(z)}{z-(q_{0}^{2}/\theta_{m}^{*})} +\frac{\operatorname{Res}_{z=q_{0}^{2}/\theta_{m}}\mathbf{R}^{-}(z)}{z-(q_{0}^{2}/\theta_{m})} \right],
\end{split}
\end{equation}
where $\mathbf{L}(z)=\mathbf{K}(z)-\mathbf{I}$, $\mathbf{R}_{\infty}^{\nabla}=\mathbf{R}_{\infty}^{-}-\mathbf{R}_{\infty}^{+}$, $\mathbf{R}_{0}^{\nabla}=\mathbf{R}_{0}^{-}-\mathbf{R}_{0}^{+}$, $\mathbf{R}_{\infty}(z)=\mathbf{R}_{\infty}^{\pm}$ and  $\mathbf{R}_{0}(z)=\mathbf{R}_{0}^{\pm}$ for  $z\in \mathbb{D}^{\pm}$. Further, the pertinent columns associated with the residue conditions are specified as follows:
\begin{equation}\label{4.12}
\begin{split}
\mathbf{R}_{1}^{-}(z)&=\begin{pmatrix}
    1  \\
    \mathbf{q}_{-}/z
  \end{pmatrix}
+\frac{1}{2\pi\mathrm{i}}\int_{\mathbb{R}} \left[ \mathbf{R}_{\infty}^{\nabla}+\frac{1}{\zeta}\mathbf{R}_{0}^{\nabla}
+\mathbf{R}^{-}(\zeta)\mathrm{e}^{\mathrm{i}\mathbf{\Delta}(\zeta)}\mathbf{L}(\zeta) \mathrm{e}^{-\mathrm{i}\mathbf{\Delta}(\zeta)} \right]_{1} \frac{\mathrm{d}\zeta}{\zeta-z} \\
&+\sum_{m=1}^{M_{1}} \left[ \frac{\mathrm{i}C_{m}\widehat{\rho}(z_{m})\mathbf{R}_{3}^{+}(z_{m})}{z-z_{m}} \right]
-\sum_{m=1}^{M_{2}} \left[ \frac{F_{m}\mathbf{R}_{2}^{+}(\theta_{m})}{z-\theta_{m}} \right], \quad
z=\theta_{n}^{*}, z_{n}^{*}.
\end{split}
\end{equation}
\begin{equation}\label{4.13}
\begin{split}
\mathbf{R}_{3}^{+}(z)&=\begin{pmatrix}
    -\mathrm{i}q_{0}/z  \\
    -\mathrm{i}\mathbf{q}_{-}/q_{0}
  \end{pmatrix}
+\frac{1}{2\pi\mathrm{i}}\int_{\mathbb{R}} \left[ \mathbf{R}_{\infty}^{\nabla} +\frac{1}{\zeta}\mathbf{R}_{0}^{\nabla}
+\mathbf{R}^{-}(\zeta)\mathrm{e}^{\mathrm{i}\mathbf{\Delta}(\zeta)}\mathbf{L}(\zeta) \mathrm{e}^{-\mathrm{i}\mathbf{\Delta}(\zeta)} \right]_{3} \frac{\mathrm{d}\zeta}{\zeta-z} \\
&+\sum_{m=1}^{M_{1}} \left[ \frac{\mathrm{i}\bar{C}_{m}\widehat{\rho}(z_{m}^{*})\mathbf{R}_{1}^{-}(z_{m}^{*})}{z-z_{m}^{*}} \right]
+\sum_{m=1}^{M_{2}} \left[ \frac{\check{F}_{m}\mathbf{R}_{2}^{-}(q_{0}^{2}/\theta_{m})}{z-(q_{0}^{2}/\theta_{m})}
 \right], \quad z=z_{n},q_{0}^{2}/\theta_{n}^{*}.
\end{split}
\end{equation}
\begin{equation}\label{4.14}
\begin{split}
\mathbf{R}_{2}^{+}(\theta_{n})&=\begin{pmatrix}
    0  \\
    q_{0}\mathbf{q}_{+}^{\perp}/\mathbf{q}_{+}^{\dagger}\mathbf{q}_{-}
  \end{pmatrix}
+\frac{1}{2\pi\mathrm{i}}\int_{\mathbb{R}} \left[ \mathbf{R}_{\infty}^{\nabla} +\frac{1}{\zeta}\mathbf{R}_{0}^{\nabla}+\mathbf{R}^{-}(\zeta)\mathrm{e}^{\mathrm{i}\mathbf{\Delta}(\zeta)}\mathbf{L}(\zeta) \mathrm{e}^{-\mathrm{i}\mathbf{\Delta}(\zeta)} \right]_{2} \frac{\mathrm{d}\zeta}{\zeta-\theta_{n}} \\
&+\sum_{m=1}^{M_{2}} \left[ \frac{\mathrm{i}\bar{F}_{m}\widehat{\rho}(\theta_{m}^{*})
\mathbf{R}_{1}^{-}(\theta_{m}^{*})}{\theta_{n}-\theta_{m}^{*}} \right]
-\sum_{m=1}^{M_{2}} \left[ \frac{\mathrm{i}\hat{F}_{m}\widehat{\rho}(q_{0}^{2}/\theta_{m}^{*})
\mathbf{R}_{3}^{+}(q_{0}^{2}/\theta_{m}^{*})}{\theta_{n}-(q_{0}^{2}/\theta_{m}^{*})} \right].
\end{split}
\end{equation}
\begin{equation}\label{4.15}
\begin{split}
\mathbf{R}_{2}^{-}\left(\frac{q_{0}^{2}}{\theta_{n}}\right)&=\begin{pmatrix}
    0  \\
    \mathbf{q}_{-}^{\perp}/q_{0}
  \end{pmatrix}
+\frac{1}{2\pi\mathrm{i}}\int_{\mathbb{R}} \left[ \mathbf{R}_{\infty}^{\nabla} +\frac{1}{\zeta}\mathbf{R}_{0}^{\nabla}
+\mathbf{R}^{-}(\zeta)\mathrm{e}^{\mathrm{i}\mathbf{\Delta}(\zeta)}\mathbf{L}(\zeta) \mathrm{e}^{-\mathrm{i}\mathbf{\Delta}(\zeta)} \right]_{2} \frac{\mathrm{d}\zeta}{\zeta-(q_{0}^{2}/\theta_{n})} \\
&+\sum_{m=1}^{M_{2}} \left[ \frac{\mathrm{i}\bar{F}_{m}\widehat{\rho}(\theta_{m}^{*})
\mathbf{R}_{1}^{-}(\theta_{m}^{*})}{(q_{0}^{2}/\theta_{n})-\theta_{m}^{*}} \right]
-\sum_{m=1}^{M_{2}} \left[ \frac{\mathrm{i}\hat{F}_{m}\widehat{\rho}(q_{0}^{2}/\theta_{m}^{*})
\mathbf{R}_{3}^{+}(q_{0}^{2}/\theta_{m}^{*})}{(q_{0}^{2}/\theta_{n})-(q_{0}^{2}/\theta_{m}^{*})} \right].
\end{split}
\end{equation}
\end{theorem}

We consider the Laurent expansion where $z\rightarrow \infty$
\begin{equation}\label{4.16}
\begin{split}
\mathbf{R}(z)&=\frac{\mathrm{i}}{2\pi}\int_{\mathbb{R}} \left[ \mathbf{R}_{\infty}^{\nabla} +\frac{1}{\zeta} \mathbf{R}_{0}^{\nabla} +\mathbf{R}^{-}(\zeta)\mathrm{e}^{\mathrm{i}\mathbf{\Delta}(\zeta)}\mathbf{L}(\zeta) \mathrm{e}^{-\mathrm{i}\mathbf{\Delta}(\zeta)} \right] \frac{\mathrm{d}\zeta}{z} \\
&+\sum_{m=1}^{M_{2}} \frac{1}{z}\left[
\underset{z=\theta_{m}}{\operatorname{Res}}\,\mathbf{R}^{+}(z)
+\underset{z=\theta_{m}^{*}}{\operatorname{Res}}\,\mathbf{R}^{-}(z)
+\underset{z=q_{0}^{2}/\theta_{m}^{*}}{\operatorname{Res}} \mathbf{R}^{+}(z) +\underset{z=q_{0}^{2}/\theta_{m}}{\operatorname{Res}} \mathbf{R}^{-}(z) \right] \\
&+\mathbf{R}_{\infty}(z)+\frac{1}{z}\mathbf{R}_{0}(z)+\sum_{m=1}^{M_{1}} \frac{1}{z}\left[
\underset{z=z_{m}}{\operatorname{Res}}\,\mathbf{R}^{+}(z)
+\underset{z=z_{m}^{*}}{\operatorname{Res}}\,\mathbf{R}^{-}(z) \right]+O\left(\frac{1}{z^{2}}\right).
\end{split}
\end{equation}
By comparing $q_{\hbar}(x,t)=\lim_{z\rightarrow\infty}[z \mathbf{R}_{(\hbar+1)1}^{-}(z)]$ for $\hbar=1,2$ and choose $\mathbf{R}(z)=\mathbf{R}^{-}(z)$ in~\eqref{4.16}, the subsequent theorem can be derived.

\begin{theorem}\label{thm:5}
The corresponding solution $\mathbf{q}(x,t)$ of the defocusing-defocusing coupled Hirota equations with non-orthogonal boundary conditions is reconstructed in the following manner:
\begin{equation}\label{4.17}
\begin{split}
q_{\hbar}(x,t)&=q_{-,\hbar}+\frac{\mathrm{i}}{2\pi}\int_{\mathbb{R}} \left[ \frac{1}{\zeta}q_{-,\hbar}
+\left( \mathbf{R}^{-}(\zeta)\mathrm{e}^{\mathrm{i}\mathbf{\Delta}(\zeta)}\mathbf{L}(\zeta) \mathrm{e}^{-\mathrm{i}\mathbf{\Delta}(\zeta)} \right)_{(\hbar+1)1} \right] \mathrm{d}\zeta \\
&+\sum_{m=1}^{M_{1}}\mathrm{i}C_{m}\widehat{\rho}(z_{m})\mathbf{R}_{(\hbar+1)3}^{+}(z_{m})
-\sum_{m=1}^{M_{2}}F_{m}\mathbf{R}_{(\hbar+1)2}^{+}(\theta_{m}), \quad \hbar=1,2.
\end{split}
\end{equation}
\end{theorem}

\subsubsection{Second method}
\label{s:Second method}

In scenarios where boundary conditions are non-orthogonal, the RH problem can likewise be converted into a system of linear integral equations employing an alternative approach.
\begin{lemma}\label{lem:5}
The modified meromorphic matrices $\widetilde{\mathbf{R}}^{\pm}(z)$ are defined as follows:
\begin{equation}\label{4.18}
\begin{split}
\widetilde{\mathbf{R}}^{+}(z)=\mathbf{R}^{+}(z)(\mathbf{R}_{\infty}^{+})^{-1}, \quad z\in\mathbb{D}^{+}, \quad
\widetilde{\mathbf{R}}^{+}(z)=\mathbf{R}^{-}(z)(\mathbf{R}_{\infty}^{-})^{-1}, \quad z\in\mathbb{D}^{-},
\end{split}
\end{equation}
the corresponding jump condition is $\widetilde{\mathbf{R}}^{+}(z)=\widetilde{\mathbf{R}}^{-}(z)
\mathrm{e}^{\mathrm{i}\mathbf{H}_{2}(z)}\widetilde{\mathbf{K}}(z) \mathrm{e}^{-\mathrm{i}\mathbf{H}_{2}(z)}$ for $z\in\mathbb{R} \backslash \{0,\pm q_{0}\}$, where $\widetilde{\mathbf{K}}(z)=\mathbf{R}_{\infty}^{-}
\mathrm{e}^{-\mathrm{i}\mathbf{H}_{1}(z)}\mathbf{K}(z) \mathrm{e}^{\mathrm{i}\mathbf{H}_{1}(z)}(\mathbf{R}_{\infty}^{+})^{-1}$, with $\mathbf{H}_{1}(z)=\operatorname{diag} \left(0, -\delta_{1}, \delta_{2} \right)$ and $\mathbf{H}_{2}(z)=\operatorname{diag} \left(\delta_{1}, \delta_{2}-\delta_{1}, \delta_{2}-\delta_{1} \right)$.
\end{lemma}

\begin{lemma}\label{lem:6}
The matrices $\widetilde{\mathbf{R}}^{\pm}(z)$ and $\widetilde{\mathbf{K}}(z)$ exhibit the following asymptotic properties:
\begin{equation}\label{4.19}
\begin{split}
\widetilde{\mathbf{K}}(z)\rightarrow \mathbf{I}, \quad z \rightarrow \infty, \quad \Rightarrow \quad
\lim _{z \rightarrow \infty} \mathbf{R}_{\infty}^{-}\mathrm{e}^{-\mathrm{i}\mathbf{H}_{1}(z)}\mathbf{K}(z)
\mathrm{e}^{\mathrm{i}\mathbf{H}_{1}(z)}(\mathbf{R}_{\infty}^{+})^{-1}=\mathbf{I}.
\end{split}
\end{equation}
\begin{equation}\label{4.20}
\begin{split}
\widetilde{\mathbf{R}}^{\pm}(z)=\mathbf{I}+O\left(\frac{1}{z}\right), \quad z \rightarrow \infty,  \quad
\widetilde{\mathbf{R}}^{\pm}(z)=\frac{1}{z}\widetilde{\mathbf{R}}_{0}^{\pm}+O(1), \quad z \rightarrow 0,
\end{split}
\end{equation}
with
\begin{equation}\label{4.21}
\everymath{\displaystyle}
\begin{aligned}
\widetilde{\mathbf{R}}_{0}^{+}=\mathbf{R}_{0}^{+}(\mathbf{R}_{\infty}^{+})^{-1}
=\begin{pmatrix}
  0 &  \frac{q_{0}^{2}\mathbf{q}_{+}^{\dagger}}{\mathbf{q}_{+}^{\dagger}\mathbf{q}_{-}} \\
  \mathbf{0}_{2 \times 1} & \mathbf{0}_{2 \times 2}
  \end{pmatrix},  \quad
\widetilde{\mathbf{R}}_{0}^{-}=\mathbf{R}_{0}^{-}(\mathbf{R}_{\infty}^{-})^{-1}
=\begin{pmatrix}
0 & \mathbf{0}_{1 \times 2} \\
\mathbf{q}_{-} & \mathbf{0}_{2 \times 2}
\end{pmatrix}.
\end{aligned}
\end{equation}
\end{lemma}

Consequently, to formulate the RH problem, it is essential to establish the correct residue conditions.
\begin{lemma}\label{lem:7}
The modified meromorphic matrices $\widetilde{\mathbf{R}}^{\pm}(z)$ satisfy the following residue conditions:
\begin{subequations}\label{4.22}
\everymath{\displaystyle}
\begin{align}
\underset{z=z_{m}}{\operatorname{Res}}\,\widetilde{\mathbf{R}}^{+}(z)
&=C_{m} \left[ \nu_{-,3}(z_{m}),\mathbf{0},\mathbf{0} \right](\mathbf{R}_{\infty}^{+})^{-1}
=\widetilde{\mathbf{R}}^{+}(z_{m})\begin{pmatrix}
    0 & \mathbf{0_{1\times2}} \\
    C_{m}\widehat{\rho}(z_{m})\mathbf{q}_{-}/\mathrm{i}q_{0} & \mathbf{0_{2\times2}}
  \end{pmatrix},   \\
\underset{z=z_{m}^{*}}{\operatorname{Res}}\,\widetilde{\mathbf{R}}^{-}(z)
&=\bar{C}_{m} \left[ \mathbf{0},\mathbf{0},\nu_{-,1}(z_{m}^{*}) \right](\mathbf{R}_{\infty}^{-})^{-1}
=\widetilde{\mathbf{R}}^{-}(z_{m}^{*})\begin{pmatrix}
    0 & \mathrm{i}\bar{C}_{m}\widehat{\rho}(z_{m}^{*})\mathbf{q}_{-}^{\dagger}/q_{0} \\
    \mathbf{0_{2\times1}} & \mathbf{0_{2\times2}}
  \end{pmatrix}, \\
\underset{z=\theta_{m}}{\operatorname{Res}}\,\widetilde{\mathbf{R}}^{+}(z)
&=F_{m} \left[ \frac{\widetilde{d}(\theta_{m})}{s_{33}(\theta_{m})},\mathbf{0},\mathbf{0} \right](\mathbf{R}_{\infty}^{+})^{-1}
=\widetilde{\mathbf{R}}^{+}(\theta_{m})\begin{pmatrix}
    0 & \mathbf{0_{1\times2}} \\
    \mathrm{i}F_{m}q_{0}\mathbf{q}_{+}^{\perp}/\mathbf{q}_{+}^{\dagger}\mathbf{q}_{-} & \mathbf{0_{2\times2}}
  \end{pmatrix},   \\
\underset{z=\theta_{m}^{*}}{\operatorname{Res}}\,\widetilde{\mathbf{R}}^{-}(z)
&=\bar{F}_{m} \left[ \mathbf{0}, \nu_{-,1}(\theta_{m}^{*}),\mathbf{0} \right](\mathbf{R}_{\infty}^{-})^{-1}
=\widetilde{\mathbf{R}}^{-}(\theta_{m}^{*})\begin{pmatrix}
    0 & \mathrm{i}\bar{F}_{m}q_{0}\widehat{\rho}(\theta_{m}^{*})
    (\mathbf{q}_{+}^{\perp})^{\dagger}/\mathbf{q}_{-}^{\dagger}\mathbf{q}_{+} \\
    \mathbf{0_{2\times1}} & \mathbf{0_{2\times2}}
  \end{pmatrix}, \\
\underset{z=q_{0}^{2}/\theta_{m}}{\operatorname{Res}} \widetilde{\mathbf{R}}^{-}(z)
&=\check{F}_{m}\left[\mathbf{0},\mathbf{0},\frac{d(q_{0}^{2}/\theta_{m})}
{s_{11}(q_{0}^{2}/\theta_{m})}\right](\mathbf{R}_{\infty}^{-})^{-1}
=\widetilde{\mathbf{R}}^{-}(q_{0}^{2}/\theta_{m})\begin{pmatrix}
    0 & \mathbf{0_{1\times2}} \\
    \mathbf{0_{2\times1}} & \check{F}_{m}\mathbf{q}_{-}^{\perp}
    \mathbf{q}_{-}^{\dagger}/q_{0}^{2} \\
  \end{pmatrix}, \\
\underset{z=q_{0}^{2}/\theta_{m}^{*}}{\operatorname{Res}} \widetilde{\mathbf{R}}^{+}(z)
&=-\hat{F}_{m}\left[\mathbf{0}, \nu_{-,3}(q_{0}^{2}/\theta_{m}^{*}), \mathbf{0}  \right](\mathbf{R}_{\infty}^{+})^{-1}
=\widetilde{\mathbf{R}}^{+}(q_{0}^{2}/\theta_{m}^{*})\begin{pmatrix}
    0 & \mathbf{0_{1\times2}} \\
    \mathbf{0_{2\times1}} & -\hat{F}_{m}\widehat{\rho}(q_{0}^{2}/\theta_{m}^{*})
    \mathbf{q}_{-}(\mathbf{q}_{-}^{\perp})^{\dagger}/q_{0}^{2}
  \end{pmatrix},
\end{align}
\end{subequations}
where $m=1,\ldots,M_{1}$ for $z_{m}$, and $m=1,\ldots,M_{2}$ for $\theta_{m}$.
\end{lemma}

The residue conditions can alternatively be expressed as a set of related column vectors:
\begin{subequations}\label{4.23}
\everymath{\displaystyle}
\begin{align}
\underset{z=z_{m}}{\operatorname{Res}}\,\widetilde{\mathbf{R}}^{+}(z)
&=\left[ C_{m}\widehat{\rho}(z_{m})\widetilde{\mathbf{R}}^{+}(z_{m})\mathbf{R}_{\infty,3}^{+}, \mathbf{0},\mathbf{0} \right], \quad \underset{z=\theta_{m}}{\operatorname{Res}}\,\widetilde{\mathbf{R}}^{+}(z)
=\left[ \mathrm{i}F_{m}\widetilde{\mathbf{R}}^{+}(\theta_{m})\mathbf{R}_{\infty,2}^{+},
\mathbf{0},\mathbf{0} \right], \\
\underset{z=z_{m}^{*}}{\operatorname{Res}}\,\widetilde{\mathbf{R}}^{-}(z)
&=\left[ \mathbf{0},\frac{\mathrm{i}\bar{C}_{m}\widehat{\rho}(z_{m}^{*})q_{-,1}^{*}}{q_{0}}
\widetilde{\mathbf{R}}^{-}(z_{m}^{*})\mathbf{R}_{\infty,1}^{-},
\frac{\mathrm{i}\bar{C}_{m}\widehat{\rho}(z_{m}^{*})q_{-,2}^{*}}{q_{0}}
\widetilde{\mathbf{R}}^{-}(z_{m}^{*})\mathbf{R}_{\infty,1}^{-} \right], \\
\underset{z=\theta_{m}^{*}}{\operatorname{Res}}\,\widetilde{\mathbf{R}}^{-}(z)
&=\left[ \mathbf{0}, \frac{\mathrm{i}\bar{F}_{m}q_{0}\widehat{\rho}(\theta_{m}^{*})q_{+,2}}
{\mathbf{q}_{-}^{\dagger}\mathbf{q}_{+}}\widetilde{\mathbf{R}}^{-}(\theta_{m}^{*})
\mathbf{R}_{\infty,1}^{-}, -\frac{\mathrm{i}\bar{F}_{m}q_{0}\widehat{\rho}(\theta_{m}^{*})q_{+,1}}
{\mathbf{q}_{-}^{\dagger}\mathbf{q}_{+}}\widetilde{\mathbf{R}}^{-}(\theta_{m}^{*})\mathbf{R}_{\infty,1}^{-} \right], \\
\underset{z=q_{0}^{2}/\theta_{m}}{\operatorname{Res}} \widetilde{\mathbf{R}}^{-}(z)
&=\left[\mathbf{0}, \frac{\check{F}_{m}q_{-,1}^{*}}{q_{0}}\widetilde{\mathbf{R}}^{-}
\left(\frac{q_{0}^{2}}{\theta_{m}}\right)\mathbf{R}_{\infty,2}^{-},  \frac{\check{F}_{m}q_{-,2}^{*}}{q_{0}}\widetilde{\mathbf{R}}^{-}
\left(\frac{q_{0}^{2}}{\theta_{m}}\right)\mathbf{R}_{\infty,2}^{-} \right], \\
\underset{z=q_{0}^{2}/\theta_{m}^{*}}{\operatorname{Res}} \widetilde{\mathbf{R}}^{+}(z)
&=\left[\mathbf{0}, -\frac{\mathrm{i}\hat{F}_{m}\widehat{\rho}(q_{0}^{2}/\theta_{m}^{*})q_{-,2}}{q_{0}}
\widetilde{\mathbf{R}}^{+}\left(\frac{q_{0}^{2}}{\theta_{m}^{*}}\right)\mathbf{R}_{\infty,3}^{+}, \frac{\mathrm{i}\hat{F}_{m}\widehat{\rho}(q_{0}^{2}/\theta_{m}^{*})q_{-,1}}{q_{0}}
\widetilde{\mathbf{R}}^{+}\left(\frac{q_{0}^{2}}{\theta_{m}^{*}}\right)\mathbf{R}_{\infty,3}^{+} \right].
\end{align}
\end{subequations}

\begin{theorem}\label{thm:6}
The solution of the RH problem is delineated through the subsequent formula:
\begin{equation}\label{4.24}
\begin{split}
\widetilde{\mathbf{R}}(z)&=\mathbf{I}+\frac{1}{z}\widetilde{\mathbf{R}}_{0}(z)
+\frac{1}{2\pi\mathrm{i}}\int_{\mathbb{R}} \left[ \frac{1}{\zeta}\widetilde{\mathbf{R}}_{0}^{\nabla}
+\widetilde{\mathbf{R}}^{-}(\zeta)\mathrm{e}^{\mathrm{i}\mathbf{H}_{2}(\zeta)}
\widetilde{\mathbf{L}}(\zeta) \mathrm{e}^{-\mathrm{i}\mathbf{H}_{2}(\zeta)} \right] \frac{\mathrm{d}\zeta}{\zeta-z} \\
&+\sum_{m=1}^{M_{1}}\left[ \frac{\operatorname{Res}_{z=z_{m}}\widetilde{\mathbf{R}}^{+}(z)}{z-z_{m}} +\frac{\operatorname{Res}_{z=z_{m}^{*}}\widetilde{\mathbf{R}}^{-}(z)}{z-z_{m}^{*}} \right]
+\sum_{m=1}^{M_{2}}\left[ \frac{\operatorname{Res}_{z=\theta_{m}} \widetilde{\mathbf{R}}^{+}(z)}{z-\theta_{m}} +\frac{\operatorname{Res}_{z=\theta_{m}^{*}}
\widetilde{\mathbf{R}}^{-}(z)}{z-\theta_{m}^{*}} \right] \\
&+\sum_{m=1}^{M_{2}} \left[ \frac{\operatorname{Res}_{z=q_{0}^{2}/\theta_{m}^{*}}
\widetilde{\mathbf{R}}^{+}(z)}{z-(q_{0}^{2}/\theta_{m}^{*})} +\frac{\operatorname{Res}_{z=q_{0}^{2}/\theta_{m}}
\widetilde{\mathbf{R}}^{-}(z)}{z-(q_{0}^{2}/\theta_{m})} \right], \quad z\in \mathbb{C}/\mathbb{R},
\end{split}
\end{equation}
where $\widetilde{\mathbf{R}}_{0}^{\nabla}=\widetilde{\mathbf{R}}_{0}^{-}-\widetilde{\mathbf{R}}_{0}^{+}$, $\widetilde{\mathbf{L}}(z)=\widetilde{\mathbf{K}}(z)-\mathbf{I}$ and $\widetilde{\mathbf{R}}_{0}(z)=\widetilde{\mathbf{R}}_{0}^{\pm}$ for $z\in \mathbb{D}^{\pm}$. The eigenfunctions within the residue conditions are characterized as follows:
\begin{equation}\label{4.25}
\begin{split}
\widetilde{\mathbf{R}}^{-}(z)\mathbf{R}_{\infty,1}^{-}&=\begin{pmatrix}
    1  \\
    \mathbf{q}_{-}/z
  \end{pmatrix}
+\frac{1}{2\pi\mathrm{i}}\int_{\mathbb{R}} \left[ \frac{1}{\zeta}\widetilde{\mathbf{R}}_{0}^{\nabla}
+\widetilde{\mathbf{R}}^{-}(\zeta)\mathrm{e}^{\mathrm{i}\mathbf{H}_{2}(\zeta)}
\widetilde{\mathbf{L}}(\zeta) \mathrm{e}^{-\mathrm{i}\mathbf{H}_{2}(\zeta)} \right]_{1} \frac{\mathrm{d}\zeta}{\zeta-z} \\&+\sum_{m=1}^{M_{1}} \left[ \frac{C_{m}\widehat{\rho}(z_{m})\widetilde{\mathbf{R}}^{+}(z_{m})\mathbf{R}_{\infty,3}^{+}}{z-z_{m}} \right]
+\sum_{m=1}^{M_{2}} \left[ \frac{\mathrm{i}F_{m}\widetilde{\mathbf{R}}^{+}(\theta_{m})
\mathbf{R}_{\infty,2}^{+}}{z-\theta_{m}} \right], \quad z=\theta_{n}^{*},z_{n}^{*}.
\end{split}
\end{equation}
\begin{equation}\label{4.26}
\begin{split}
\widetilde{\mathbf{R}}^{+}(z)\mathbf{R}_{\infty,3}^{+}&=
\frac{1}{2\pi\mathrm{i}}\int_{\mathbb{R}} \left[ \frac{1}{\zeta}\widetilde{\mathbf{R}}_{0}^{\nabla}
+\widetilde{\mathbf{R}}^{-}(\zeta)\mathrm{e}^{\mathrm{i}\mathbf{H}_{2}(\zeta)}
\widetilde{\mathbf{L}}(\zeta) \mathrm{e}^{-\mathrm{i}\mathbf{H}_{2}(\zeta)} \right]\mathbf{R}_{\infty,3}^{+} \frac{\mathrm{d}\zeta}{\zeta-z} \\
&+\sum_{m=1}^{M_{1}} \left[ \frac{\bar{C}_{m}\widehat{\rho}(z_{m}^{*})}{z-z_{m}^{*}}
\widetilde{\mathbf{R}}^{-}(z_{m}^{*})\mathbf{R}_{\infty,1}^{-} \right]
-\sum_{m=1}^{M_{2}} \left[ \frac{\mathrm{i}\check{F}_{m}\widetilde{\mathbf{R}}^{-}
(q_{0}^{2}/\theta_{m})\mathbf{R}_{\infty,2}^{-}}{z-(q_{0}^{2}/\theta_{m})} \right] \\
&+\begin{pmatrix}
    -\mathrm{i}q_{0}/z  \\
    -\mathrm{i}\mathbf{q}_{-}/q_{0}
  \end{pmatrix}
+\sum_{m=1}^{M_{2}} \frac{(\mathbf{q}_{+}^{\perp})^{\dagger}\mathbf{q}_{-}}{\mathbf{q}_{-}^{\dagger}\mathbf{q}_{+}}
\left[ \frac{\bar{F}_{m}\widehat{\rho}(\theta_{m}^{*})}{z-\theta_{m}^{*}}
\widetilde{\mathbf{R}}^{-}(\theta_{m}^{*})\mathbf{R}_{\infty,1}^{-} \right], \quad z=z_{n},q_{0}^{2}/\theta_{n}^{*}.
\end{split}
\end{equation}
\begin{equation}\label{4.27}
\begin{split}
\widetilde{\mathbf{R}}^{+}(\theta_{n})\mathbf{R}_{\infty,2}^{+}&=\begin{pmatrix}
    0  \\
    q_{0}\mathbf{q}_{+}^{\perp}/\mathbf{q}_{+}^{\dagger}\mathbf{q}_{-}
  \end{pmatrix}
+\frac{1}{2\pi\mathrm{i}}\int_{\mathbb{R}} \left[ \frac{1}{\zeta}\widetilde{\mathbf{R}}_{0}^{\nabla}
+\widetilde{\mathbf{R}}^{-}(\zeta)\mathrm{e}^{\mathrm{i}\mathbf{H}_{2}(\zeta)}
\widetilde{\mathbf{L}}(\zeta) \mathrm{e}^{-\mathrm{i}\mathbf{H}_{2}(\zeta)} \right]\mathbf{R}_{\infty,2}^{+} \frac{\mathrm{d}\zeta}{\zeta-\theta_{n}} \\
&+\sum_{m=1}^{M_{1}}
\frac{\mathbf{q}_{-}^{\dagger}\mathbf{q}_{+}^{\perp}}{\mathbf{q}_{+}^{\dagger}\mathbf{q}_{-}}
\left[ \frac{\mathrm{i}\bar{C}_{m}\widehat{\rho}(z_{m}^{*})}{\theta_{n}-z_{m}^{*}}
\widetilde{\mathbf{R}}^{-}(z_{m}^{*})\mathbf{R}_{\infty,1}^{-} \right]
+\sum_{m=1}^{M_{2}} \left[
\frac{\mathrm{i}q_{0}^{4}\bar{F}_{m}\widehat{\rho}(\theta_{m}^{*})}{ | \mathbf{q}_{-}^{\dagger}\mathbf{q}_{+} |^{2} (\theta_{n}-\theta_{m}^{*}) }
\widetilde{\mathbf{R}}^{-}(\theta_{m}^{*})\mathbf{R}_{\infty,1}^{-} \right] \\
&+\sum_{m=1}^{M_{2}} \frac{\mathbf{q}_{-}^{\dagger}\mathbf{q}_{+}^{\perp}}{\mathbf{q}_{+}^{\dagger}\mathbf{q}_{-}}
\left[ \frac{\check{F}_{m}\widetilde{\mathbf{R}}^{-}(q_{0}^{2}/\theta_{m})
\mathbf{R}_{\infty,2}^{-}}{\theta_{n}-(q_{0}^{2}/\theta_{m})} \right]
-\sum_{m=1}^{M_{2}} \left[ \frac{\mathrm{i}\hat{F}_{m}\widehat{\rho}(q_{0}^{2}/\theta_{m}^{*})
\widetilde{\mathbf{R}}^{+}(q_{0}^{2}/\theta_{m}^{*})
\mathbf{R}_{\infty,3}^{+}}{\theta_{n}-(q_{0}^{2}/\theta_{m}^{*})} \right].
\end{split}
\end{equation}
\begin{equation}\label{4.28}
\begin{split}
\widetilde{\mathbf{R}}^{-}\left(\frac{q_{0}^{2}}{\theta_{n}}\right)\mathbf{R}_{\infty,2}^{-}&=\begin{pmatrix}
    0  \\
    \mathbf{q}_{-}^{\perp}/q_{0}
  \end{pmatrix}
+\frac{1}{2\pi\mathrm{i}}\int_{\mathbb{R}} \left[ \frac{1}{\zeta}\widetilde{\mathbf{R}}_{0}^{\nabla}
+\widetilde{\mathbf{R}}^{-}(\zeta)\mathrm{e}^{\mathrm{i}\mathbf{H}_{2}(\zeta)}
\widetilde{\mathbf{L}}(\zeta) \mathrm{e}^{-\mathrm{i}\mathbf{H}_{2}(\zeta)} \right]\mathbf{R}_{\infty,2}^{-} \frac{\mathrm{d}\zeta}{\zeta-(q_{0}^{2}/\theta_{n})} \\
&+\sum_{m=1}^{M_{2}}
\left[ \frac{\mathrm{i}\bar{F}_{m}\widehat{\rho}(\theta_{m}^{*})
\widetilde{\mathbf{R}}^{-}(\theta_{m}^{*})\mathbf{R}_{\infty,1}^{-}}{(q_{0}^{2}/\theta_{n})-\theta_{m}^{*}} \right]-\sum_{m=1}^{M_{2}} \left[ \frac{\mathrm{i}\hat{F}_{m}\widehat{\rho}(q_{0}^{2}/\theta_{m}^{*})
\widetilde{\mathbf{R}}^{+}(q_{0}^{2}/\theta_{m}^{*})
\mathbf{R}_{\infty,3}^{+}}{(q_{0}^{2}/\theta_{n})-(q_{0}^{2}/\theta_{m}^{*})} \right].
\end{split}
\end{equation}
\end{theorem}

\begin{lemma}\label{lem:8}
If the RH problem outlined in Lemmas~\ref{lem:5},~\ref{lem:6} and~\ref{lem:7} has a unique solution. Then the matrix $\widetilde{\mathbf{R}}(z)$ adheres to the modified Lax pair, which can be stated as:
\begin{equation}\label{4.29}
\begin{split}
\widetilde{\mathbf{R}}_{x}^{\pm}(z)&=\mathrm{i}k[\mathbf{J},\widetilde{\mathbf{R}}^{\pm}(z)]
+\mathrm{i}\mathbf{Q}\widetilde{\mathbf{R}}^{\pm}(z)
+\mathrm{i}\widetilde{\mathbf{R}}^{\pm}(z)\mathbf{R}_{\infty}^{\pm}
[k\mathbf{J}-\mathbf{\Lambda}_{1}(z)](\mathbf{R}_{\infty}^{\pm})^{-1}, \quad z\in \mathbb{D}^{\pm},
\end{split}
\end{equation}
and
\begin{equation}\label{4.30}
\begin{split}
\widetilde{\mathbf{R}}_{t}^{\pm}(z)&=4\mathrm{i}\sigma k^{3}[\mathbf{J},\widetilde{\mathbf{R}}^{\pm}(z)]
+[-\mathrm{i}q_{0}^{2}\mathbf{J}+(4\mathrm{i}\sigma\mathbf{Q}+2\mathrm{i}\mathbf{J})k^{2}
+(2\mathrm{i}\mathbf{Q}-2\sigma\mathbf{Q}_{x}\mathbf{J}-2\mathrm{i}\sigma\mathbf{J}\mathbf{Q}^{2})k
-2\mathrm{i}\sigma\mathbf{Q}^{3} \\
&-\mathrm{i}\mathbf{J}\mathbf{Q}^{2}-\mathrm{i}\sigma\mathbf{Q}_{xx}
-\mathbf{Q}_{x}\mathbf{J}+\sigma[\mathbf{Q},\mathbf{Q}_{x}]]\widetilde{\mathbf{R}}^{\pm}(z)
+\mathrm{i}\widetilde{\mathbf{R}}^{\pm}(z)\mathbf{R}_{\infty}^{\pm}
[4\sigma k^{3}\mathbf{J}-\mathbf{\Lambda}_{2}(z)](\mathbf{R}_{\infty}^{\pm})^{-1},
\quad z\in \mathbb{D}^{\pm},
\end{split}
\end{equation}
where
\begin{equation}\label{4.31}
\begin{split}
\mathbf{Q}(x,t)=-\frac{1}{2}\lim _{z\rightarrow \infty}z [\mathbf{J},\widetilde{\mathbf{R}}^{\pm}(z)].
\end{split}
\end{equation}
\end{lemma}

\begin{theorem}\label{thm:7}
Let $\widetilde{\mathbf{R}}(z)$ be the solution of the RH problem in Theorem~\ref{thm:6}. Then the corresponding solution $\mathbf{q}(x,t)$ of the defocusing-defocusing coupled Hirota equations with non-orthogonal boundary conditions is reconstructed as
\begin{equation}\label{4.32}
\begin{split}
q_{\hbar}(x,t)&=q_{-,\hbar}+\frac{\mathrm{i}}{2\pi}\int_{\mathbb{R}} \left[ \frac{q_{-,\hbar}}{\zeta}
+\left( \widetilde{\mathbf{R}}^{-}(\zeta)\mathrm{e}^{\mathrm{i}\mathbf{H}_{2}(\zeta)}
\widetilde{\mathbf{L}}(\zeta)\mathrm{e}^{-\mathrm{i}\mathbf{H}_{2}(\zeta)} \right)_{(\hbar+1)1} \right] \mathrm{d}\zeta \\
&+\sum_{m=1}^{M_{1}}C_{m}\widehat{\rho}(z_{m}) \left[ \widetilde{\mathbf{R}}^{+}(z_{m})\mathbf{R}_{\infty}^{+} \right]_{(\hbar+1)3}
+\sum_{m=1}^{M_{2}}\mathrm{i}F_{m} \left[ \widetilde{\mathbf{R}}^{+}(\theta_{m}) \mathbf{R}_{\infty}^{+} \right]_{(\hbar+1)2}, \quad \hbar=1,2.
\end{split}
\end{equation}
\end{theorem}

\begin{remark}\label{rem:6}
It is important to recognize that although the function $\mathbf{q}(x,t)$ as determined by the reconstruction formula, solves the defocusing-defocusing coupled Hirota equations, confirming its compliance with non-parallel boundary conditions remains a challenging endeavor. This holds true even when considering the defocusing Hirota equation under the conditions of fully asymmetric NZBCs~\cite{50}.
\end{remark}

\subsection{Trace formulae}
\label{s:Trace formulae}

The trace formula for the scattering coefficient based on the parallel case~\cite{23} is given
\begin{subequations}\label{4.33}
\begin{align}
h_{11}(z)&=\prod_{m=1}^{M_{1}}\frac{z-z_{m}}{z-z_{m}^{*}}
\prod_{m=1}^{M_{2}}\frac{z-\theta_{m}}{z-\theta_{m}^{*}}
\exp\left[ -\frac{1}{2\pi\mathrm{i}} \int_{\mathbb{R}} \ln \left[1-\left| \beta_{1}(\zeta) \right|^{2}-\frac{\left| \beta_{2}(\zeta) \right|^{2}}{\widehat{\rho}(\zeta)} \right] \frac{\mathrm{d}\zeta}{\zeta-z} \right], &z\in \mathbb{D}^{+}, \\
s_{11}(z)&=\prod_{m=1}^{M_{1}}\frac{z-z_{m}^{*}}{z-z_{m}}
\prod_{m=1}^{M_{2}}\frac{z-\theta_{m}^{*}}{z-\theta_{m}}\exp\left[ \frac{1}{2\pi\mathrm{i}} \int_{\mathbb{R}} \ln \left[1-\left| \beta_{1}(\zeta) \right|^{2}-\frac{\left| \beta_{2}(\zeta) \right|^{2}}{\widehat{\rho}(\zeta)} \right] \frac{\mathrm{d}\zeta}{\zeta-z} \right], &z\in \mathbb{D}^{-}.
\end{align}
\end{subequations}
In scenarios where boundary conditions are orthogonal, the asymptotic behavior of $h_{11}\rightarrow 0$ as expressed in Eq.~\eqref{3.13} when $z\rightarrow 0$ invalidates both the derivation and the validity of the trace formula itself. When dealing with non-orthogonal boundary conditions, the interplay between $\mathbf{q}_{+}$ and $\mathbf{q}_{-}$ becomes more intricate compared to scenarios with parallel boundary conditions. Leveraging Eq.~\eqref{4.33} along with the asymptotic behaviors of $h_{11}(z)$ and $\beta_{2}(z)$ as $z\rightarrow 0$, as detailed in Proposition \ref{pro:15}, we derive the subsequent relationship among $\mathbf{q}_{-}$, $\mathbf{q}_{+}$ and $\mathbf{q}_{+}^{\perp}$.
\begin{proposition}\label{pro:16}
The boundary values of the potential at $x\rightarrow \pm \infty$ are interconnected through the following formula:
\begin{equation}\label{4.34}
\begin{split}
\mathbf{q}_{-}=\eta_{1}\mathbf{q}_{+}+\eta_{2}\mathbf{q}_{+}^{\perp},
\end{split}
\end{equation}
where
\begin{subequations}\label{4.35}
\begin{align}
\eta_{1}&=\frac{\mathbf{q}_{+}^{\dagger}\mathbf{q}_{-}}{q_{0}^{2}}=\prod_{m=1}^{M_{1}}
\frac{z_{m}^{*}}{z_{m}} \prod_{m=1}^{M_{2}}\frac{\theta_{m}^{*}}{\theta_{m}}
\times \exp \left\{ \frac{1}{2\pi\mathrm{i}}\int_{\mathbb{R}} \ln \left[1-\left| \beta_{1}(\zeta) \right|^{2}-\frac{\left| \beta_{2}(\zeta) \right|^{2}}{\widehat{\rho}(\zeta)} \right] \frac{\mathrm{d}\zeta}{\zeta} \right\}, \\
\eta_{2}&=-\frac{(\mathbf{q}_{-}^{\perp})^{\dagger}\mathbf{q}_{+}}{q_{0}^{2}} =\frac{\mathbf{q}_{-}^{\dagger}\mathbf{q}_{+}}{\mathrm{i}q_{0}}\lim_{z\rightarrow0}\frac{\beta_{2}(z)}{z}.
\end{align}
\end{subequations}
\end{proposition}

Specifically, when $\eta_{2}=0$, Eq.~\eqref{4.34} reduces to the asymptotic phase difference in the particular scenario where $\mathbf{q}_{+}$ and $\mathbf{q}_{-}$ are parallel, as discussed in~\cite{23}.

\subsection{Riemann-Hilbert problem: orthogonal boundary conditions}
\label{s:Riemann-Hilbert problem: orthogonal boundary conditions}

When considering the RH problem under the orthogonal boundary conditions, additional difficulties arise due to the presence of orthogonal polarization vectors. Due to $\mathbf{q}_{+}^{\dagger}\mathbf{q}_{-}=0$, the asymptotic behavior of $h_{33}=O(1/z^{2})$ as $z\rightarrow \infty$ and $h_{11}=O(z^{2})$ as $z\rightarrow 0$ in Proposition~\ref{pro:14} make the expression of the inverse problem in IST very complex. This is because $\mathbf{R}^{\pm}(z)$ introduces a pole at $z=0$, and some terms in the leading order asymptotic behavior of $\mathbf{R}^{\pm}(z)$ diverge. To address this difficulty~\cite{12}, consider defining the new modified meromorphic matrices $\widehat{\mathbf{R}}^{\pm}(z)$:

\begin{lemma}\label{lem:9}
The modified meromorphic matrices $\widehat{\mathbf{R}}^{\pm}(z)$ are defined by the following expressions:
\begin{subequations}\label{4.36}
\begin{align}
\widehat{\mathbf{R}}^{+}(z)&=\mathbf{R}^{+}(z) \operatorname{diag} \left( 1, \frac{1}{z}, 1 \right)= \left[ \frac{\nu_{+,1}(z)}{h_{11}}, -\frac{\widetilde{d}(z)}{s_{33}z}, \frac{\nu_{-,3}(z)}{\mathrm{i}\widehat{\rho}} \right], &z\in\mathbb{D}^{+}, \\
\widehat{\mathbf{R}}^{-}(z)&=\mathbf{R}^{-}(z)
\operatorname{diag} \left( 1, 1, \frac{1}{z} \right)= \left[ \frac{\nu_{-,1}(z)}{\mathrm{i}\widehat{\rho}},
\frac{d(z)}{s_{11}}, \frac{\nu_{+,3}(z)}{h_{33}z} \right], &z\in\mathbb{D}^{-},
\end{align}
\end{subequations}
satisfy the new jump condition
\begin{equation}\label{4.37}
\begin{split}
\widehat{\mathbf{R}}^{+}(z)=\widehat{\mathbf{R}}^{-}(z)
\mathrm{e}^{\mathrm{i}\mathbf{\Delta}(z)}\widehat{\mathbf{K}}(z) \mathrm{e}^{-\mathrm{i}\mathbf{\Delta}(z)}, \quad z\in \mathbb{R} \backslash \{0,\pm q_{0}\},
\end{split}
\end{equation}
where $\widehat{\mathbf{K}}(z)=\operatorname{diag} \left( 1, 1, z \right) \mathbf{K}(z)
\operatorname{diag} \left( 1, 1/z, 1 \right)$, $\widehat{\rho}=\widehat{\rho}(z)$, $s_{11}=s_{11}(z)$, $h_{11}=h_{11}(z)$, $s_{33}=s_{33}(z)$ and $h_{33}=h_{33}(z)$.
\end{lemma}

\begin{lemma}\label{lem:10}
The matrices $\widehat{\mathbf{R}}^{\pm}(z)$ have the following asymptotic behavior:
\begin{equation}\label{4.38}
\begin{split}
\widehat{\mathbf{R}}^{\pm}(z)=\widehat{\mathbf{R}}_{\infty}^{\pm}+O\left(\frac{1}{z}\right), \quad z\rightarrow \infty, \quad \widehat{\mathbf{R}}^{\pm}(z)=\frac{1}{z}\widehat{\mathbf{R}}_{0}^{\pm}+O(1),  \quad z\rightarrow 0,
\end{split}
\end{equation}
where
\begin{subequations}\label{4.39}
\everymath{\displaystyle}
\begin{align}
\widehat{\mathbf{R}}_{0}^{+}&=\begin{pmatrix}
   0 & 0  & -\mathrm{i}q_{0} \\
   \mathbf{0_{2\times1}} & \frac{\mathbf{q}_{-}^{\perp}}{q_{0}} & \mathbf{0_{2\times1}}
  \end{pmatrix},  \quad
\widehat{\mathbf{R}}_{\infty}^{+}=\begin{pmatrix}
    \mathrm{i} & 0 & 0 \\
    \mathbf{0_{2\times1}} & \frac{q_{0}\mathbf{q}_{+}^{\perp}}{\mathrm{i}\vartheta_{1}} & \frac{\mathbf{q}_{-}}{\mathrm{i}q_{0}}
  \end{pmatrix},  \\
\widehat{\mathbf{R}}_{0}^{-}&=\begin{pmatrix}
   0 & 0  & 0 \\
   \mathbf{q}_{-} & \frac{q_{0}^{3}\mathbf{q}_{+}^{\perp}}{\mathrm{i}\vartheta_{1}} & \mathbf{0_{2\times1}}
  \end{pmatrix},  \quad
\widehat{\mathbf{R}}_{\infty}^{-}=\begin{pmatrix}
    1 & 0 & 0 \\
    \mathbf{0_{2\times1}} & \frac{\mathbf{q}_{-}^{\perp}}{q_{0}} & \frac{\mathrm{i}q_{0}\mathbf{q}_{+}}{\vartheta_{1}^{*}}
  \end{pmatrix}.
\end{align}
\end{subequations}
\end{lemma}

Consequently, to formulate the RH problem, one must establish the suitable residue conditions. Due to the different modified meromorphic matrices $\widehat{\mathbf{R}}^{\pm}(z)$, there are significant differences in the results compared to the non-orthogonal NZBCs.

\begin{lemma}\label{lem:11}
The modified meromorphic matrices $\widehat{\mathbf{R}}^{\pm}(z)$ satisfy the following residue conditions:
\begin{subequations}\label{4.40}
\everymath{\displaystyle}
\begin{align}
\underset{z=z_{m}}{\operatorname{Res}}\,\widehat{\mathbf{R}}^{+}(z)
&=C_{m} \left[ \nu_{-,3}(z_{m}),\mathbf{0},\mathbf{0} \right]
=\widehat{\mathbf{R}}^{+}(z_{m})\begin{pmatrix}
    \mathbf{0_{2\times1}} & \mathbf{0_{2\times2}} \\
    \mathrm{i}C_{m}\widehat{\rho}(z_{m}) & \mathbf{0_{1\times2}}
  \end{pmatrix}, \\
\underset{z=z_{m}^{*}}{\operatorname{Res}}\,\widehat{\mathbf{R}}^{-}(z)
&= \frac{\bar{C}_{m}}{z_{m}^{*}} \left[ \mathbf{0},\mathbf{0}, \nu_{-,1}(z_{m}^{*}) \right]
=\widehat{\mathbf{R}}^{-}(z_{m}^{*})\begin{pmatrix}
    \mathbf{0_{1\times2}} & \frac{\mathrm{i}\bar{C}_{m}}{2\lambda(z_{m}^{*})} \\
    \mathbf{0_{2\times2}} & \mathbf{0_{2\times1}}
  \end{pmatrix}, \\
\underset{z=\theta_{m}}{\operatorname{Res}}\,\widehat{\mathbf{R}}^{+}(z)
&=F_{m} \left[ \frac{\widetilde{d}(\theta_{m})}{s_{33}(\theta_{m})} ,\mathbf{0},\mathbf{0} \right]
=\widehat{\mathbf{R}}^{+}(\theta_{m})\begin{pmatrix}
    0 & \mathbf{0_{1\times2}} \\
    -F_{m}\theta_{m} & \mathbf{0_{1\times2}}\\
    0 & \mathbf{0_{1\times2}}
  \end{pmatrix}, \\
\underset{z=\theta_{m}^{*}}{\operatorname{Res}}\,\widehat{\mathbf{R}}^{-}(z)
&=\bar{F}_{m} \left[\mathbf{0},\nu_{-,1}(\theta_{m}^{*}),\mathbf{0}\right]
=\widehat{\mathbf{R}}^{-}(\theta_{m}^{*})\begin{pmatrix}
    0 & \mathrm{i}\bar{F}_{m}\widehat{\rho}(\theta_{m}^{*}) & 0 \\
    \mathbf{0_{2\times1}} & \mathbf{0_{2\times1}} & \mathbf{0_{2\times1}}
  \end{pmatrix}, \\
\underset{z=q_{0}^{2}/\theta_{m}}{\operatorname{Res}} \widehat{\mathbf{R}}^{-}(z)
&=\frac{\check{F}_{m}\theta_{m}}{q_{0}^{2}} \left[\mathbf{0},\mathbf{0},
\frac{d(q_{0}^{2}/\theta_{m})}{s_{11}(q_{0}^{2}/\theta_{m})} \right]
=\widehat{\mathbf{R}}^{-}\left(\frac{q_{0}^{2}}{\theta_{m}}\right)\begin{pmatrix}
    \mathbf{0_{1\times2}} & 0 \\
    \mathbf{0_{1\times2}} & \frac{\check{F}_{m}\theta_{m}}{q_{0}^{2}} \\
    \mathbf{0_{1\times2}} & 0
  \end{pmatrix}, \\
\underset{z=q_{0}^{2}/\theta_{m}^{*}}{\operatorname{Res}} \widehat{\mathbf{R}}^{+}(z)
&=-\frac{\hat{F}_{m}\theta_{m}^{*}}{q_{0}^{2}} \left[\mathbf{0},
\nu_{-,3}\left(\frac{q_{0}^{2}}{\theta_{m}^{*}}\right), \mathbf{0} \right]
=\widehat{\mathbf{R}}^{+}\left(\frac{q_{0}^{2}}{\theta_{m}^{*}}\right)\begin{pmatrix}
    \mathbf{0_{2\times1}} & \mathbf{0_{2\times1}} & \mathbf{0_{2\times1}} \\
    0 & \frac{\mathrm{i}\hat{F}_{m}}{2\lambda(\theta_{m}^{*})} & 0
  \end{pmatrix},
\end{align}
\end{subequations}
where $m=1,\ldots,M_{1}$ for $z_{m}$, and $m=1,\ldots,M_{2}$ for $\theta_{m}$.
\end{lemma}

The residue conditions can be alternatively articulated as a series of interconnected column vectors:
\begin{subequations}\label{4.41}
\everymath{\displaystyle}
\begin{align}
\underset{z=z_{m}}{\operatorname{Res}}\,\widehat{\mathbf{R}}^{+}(z)
&=\left[\mathrm{i}C_{m}\widehat{\rho}(z_{m})\mathbf{R}_{3}^{+}(z_{m}),\mathbf{0},\mathbf{0}\right], \quad
\underset{z=\theta_{m}^{*}}{\operatorname{Res}}\,\widehat{\mathbf{R}}^{-}(z)
=\left[\mathbf{0},\mathrm{i}\bar{F}_{m}\widehat{\rho}(\theta_{m}^{*})
\mathbf{R}_{1}^{-}(\theta_{m}^{*}),\mathbf{0}\right],  \\
\underset{z=\theta_{m}}{\operatorname{Res}}\,\widehat{\mathbf{R}}^{+}(z)
&=\left[ -F_{m}\mathbf{R}_{2}^{+}(\theta_{m}),\mathbf{0},\mathbf{0} \right], \quad
\underset{z=q_{0}^{2}/\theta_{m}^{*}}{\operatorname{Res}} \widehat{\mathbf{R}}^{+}(z)
=\left[\mathbf{0}, \frac{\hat{F}_{m}\theta_{m}^{*}}{\mathrm{i}q_{0}^{2}}
\widehat{\rho}\left(\frac{q_{0}^{2}}{\theta_{m}^{*}}\right) \mathbf{R}_{3}^{+}
\left(\frac{q_{0}^{2}}{\theta_{m}^{*}} \right), \mathbf{0} \right], \\
\underset{z=z_{m}^{*}}{\operatorname{Res}}\,\widehat{\mathbf{R}}^{-}(z)
&=\left[\mathbf{0}, \frac{\mathrm{i}\bar{C}_{m}\widehat{\rho}(z_{m}^{*})
\mathbf{R}_{1}^{-}(z_{m}^{*})}{z_{m}^{*}}, \mathbf{0}\right], \quad
\underset{z=q_{0}^{2}/\theta_{m}}{\operatorname{Res}} \widehat{\mathbf{R}}^{-}(z)
=\left[ \mathbf{0},\mathbf{0}, \frac{\check{F}_{m}\theta_{m}}{q_{0}^{2}}
\mathbf{R}_{2}^{-}\left(\frac{q_{0}^{2}}{\theta_{m}}\right) \right].
\end{align}
\end{subequations}

\begin{theorem}\label{thm:8}
The solution of the RH problem is delineated by the ensuing formula:
\begin{equation}\label{4.42}
\begin{split}
\widehat{\mathbf{R}}(z)&=\widehat{\mathbf{R}}_{\infty}(z)+\frac{1}{z}\widehat{\mathbf{R}}_{0}(z)
+\frac{1}{2\pi\mathrm{i}}\int_{\mathbb{R}} \left[ \widehat{\mathbf{R}}_{\infty}^{\nabla} +\frac{1}{\zeta}\widehat{\mathbf{R}}_{0}^{\nabla}
+\widehat{\mathbf{R}}^{-}(\zeta)\mathrm{e}^{\mathrm{i}\mathbf{\Delta}(\zeta)}\widehat{\mathbf{L}}(\zeta) \mathrm{e}^{-\mathrm{i}\mathbf{\Delta}(\zeta)} \right] \frac{\mathrm{d}\zeta}{\zeta-z} \\
&+\sum_{m=1}^{M_{1}}\left[ \frac{\operatorname{Res}_{z=z_{m}}\widehat{\mathbf{R}}^{+}(z)}{z-z_{m}} +\frac{\operatorname{Res}_{z=z_{m}^{*}}\widehat{\mathbf{R}}^{-}(z)}{z-z_{m}^{*}} \right]
+\sum_{m=1}^{M_{2}}\left[ \frac{\operatorname{Res}_{z=\theta_{m}}\widehat{\mathbf{R}}^{+}(z)}{z-\theta_{m}} +\frac{\operatorname{Res}_{z=\theta_{m}^{*}}\widehat{\mathbf{R}}^{-}(z)}{z-\theta_{m}^{*}} \right] \\
&+\sum_{m=1}^{M_{2}} \left[ \frac{\operatorname{Res}_{z=q_{0}^{2}/\theta_{m}^{*}}\widehat{\mathbf{R}}^{+}(z)}{z-(q_{0}^{2}/\theta_{m}^{*})} +\frac{\operatorname{Res}_{z=q_{0}^{2}/\theta_{m}}\widehat{\mathbf{R}}^{-}(z)}{z-(q_{0}^{2}/\theta_{m})} \right],
\quad z\in \mathbb{C}/\mathbb{R},
\end{split}
\end{equation}
where $\widehat{\mathbf{R}}_{\infty}^{\nabla}=\widehat{\mathbf{R}}_{\infty}^{-}-\widehat{\mathbf{R}}_{\infty}^{+}$, $\widehat{\mathbf{R}}_{0}^{\nabla}=\widehat{\mathbf{R}}_{0}^{-}-\widehat{\mathbf{R}}_{0}^{+}$, $\widehat{\mathbf{L}}(z)=\widehat{\mathbf{K}}(z)-\mathbf{I}$, $\widehat{\mathbf{R}}_{\infty}(z)=\widehat{\mathbf{R}}_{\infty}^{\pm}$ and $\widehat{\mathbf{R}}_{0}(z)=\widehat{\mathbf{R}}_{0}^{\pm}$ for $z\in \mathbb{D}^{\pm}$. The eigenfunctions encompassed within the residue conditions are as follows:
\begin{equation}\label{4.43}
\begin{split}
\mathbf{R}_{1}^{-}(z)&=\begin{pmatrix}
    1  \\
    \mathbf{q}_{-}/z
  \end{pmatrix}
+\frac{1}{2\pi\mathrm{i}}\int_{\mathbb{R}} \left[ \widehat{\mathbf{R}}_{\infty}^{\nabla} +\frac{1}{\zeta}\widehat{\mathbf{R}}_{0}^{\nabla}
+\widehat{\mathbf{R}}^{-}(\zeta)\mathrm{e}^{\mathrm{i}\mathbf{\Delta}(\zeta)}\widehat{\mathbf{L}}(\zeta) \mathrm{e}^{-\mathrm{i}\mathbf{\Delta}(\zeta)} \right]_{1} \frac{\mathrm{d}\zeta}{\zeta-z} \\
&+\sum_{m=1}^{M_{1}} \frac{\mathrm{i}C_{m}\widehat{\rho}(z_{m})\mathbf{R}_{3}^{+}(z_{m})}{z-z_{m}}
-\sum_{m=1}^{M_{2}} \frac{F_{m}\mathbf{R}_{2}^{+}(\theta_{m})}{z-\theta_{m}}, \quad z=\theta_{n}^{*},z_{n}^{*}.
\end{split}
\end{equation}
\begin{equation}\label{4.44}
\begin{split}
\mathbf{R}_{3}^{+}(z)&=\begin{pmatrix}
    -\mathrm{i}q_{0}/z  \\
    -\mathrm{i}\mathbf{q}_{-}/q_{0}
  \end{pmatrix}
+\frac{1}{2\pi\mathrm{i}}\int_{\mathbb{R}} \left[ \widehat{\mathbf{R}}_{\infty}^{\nabla} +\frac{1}{\zeta}\widehat{\mathbf{R}}_{0}^{\nabla}
+\widehat{\mathbf{R}}^{-}(\zeta)\mathrm{e}^{\mathrm{i}\mathbf{\Delta}(\zeta)}\widehat{\mathbf{L}}(\zeta) \mathrm{e}^{-\mathrm{i}\mathbf{\Delta}(\zeta)} \right]_{3} \frac{\mathrm{d}\zeta}{\zeta-z} \\
&+\sum_{m=1}^{M_{1}} \frac{\mathrm{i}\bar{C}_{m}\mathbf{R}_{1}^{-}(z_{m}^{*})}{2\lambda(z_{m}^{*})(z-z_{m}^{*})}
+\sum_{m=1}^{M_{2}}
\frac{\check{F}_{m}\theta_{m}\mathbf{R}_{2}^{-}(q_{0}^{2}/\theta_{m})}{q_{0}^{2}[z-(q_{0}^{2}/\theta_{m})]},
\quad  z=z_{n},q_{0}^{2}/\theta_{n}^{*}.
\end{split}
\end{equation}
\begin{equation}\label{4.45}
\everymath{\displaystyle}
\begin{split}
\mathbf{R}_{2}^{+}(\theta_{n})&=\begin{pmatrix}
    0  \\
    \frac{q_{0}\theta_{n}\mathbf{q}_{+}^{\perp}}{\mathrm{i}\vartheta_{1}}
    +\frac{\mathbf{q}_{-}^{\perp}}{q_{0}}
  \end{pmatrix}
+\frac{\theta_{n}}{2\pi\mathrm{i}}\int_{\mathbb{R}} \left[ \widehat{\mathbf{R}}_{\infty}^{\nabla}+\frac{1}{\zeta}\widehat{\mathbf{R}}_{0}^{\nabla}
+\widehat{\mathbf{R}}^{-}(\zeta)\mathrm{e}^{\mathrm{i}\mathbf{\Delta}(\zeta)}\widehat{\mathbf{L}}(\zeta) \mathrm{e}^{-\mathrm{i}\mathbf{\Delta}(\zeta)} \right]_{2} \frac{\mathrm{d}\zeta}{\zeta-\theta_{n}} \\
&+\sum_{m=1}^{M_{2}} \frac{\mathrm{i}\bar{F}_{m}\theta_{n}\widehat{\rho}(\theta_{m}^{*})
\mathbf{R}_{1}^{-}(\theta_{m}^{*})}{\theta_{n}-\theta_{m}^{*}}
+\sum_{m=1}^{M_{2}} \frac{\mathrm{i}\hat{F}_{m}\mathbf{R}_{3}^{+}(q_{0}^{2}/\theta_{m}^{*})}
{2\lambda(\theta_{m}^{*})[\theta_{n}-(q_{0}^{2}/\theta_{m}^{*})]}.
\end{split}
\end{equation}
\begin{equation}\label{4.46}
\everymath{\displaystyle}
\begin{split}
\mathbf{R}_{2}^{-}\left(\frac{q_{0}^{2}}{\theta_{n}}\right)&=\begin{pmatrix}
    0  \\
    \frac{q_{0}\theta_{n}\mathbf{q}_{+}^{\perp}}{\mathrm{i}\vartheta_{2}}
    +\frac{\mathbf{q}_{-}^{\perp}}{q_{0}}
  \end{pmatrix}
+\frac{1}{2\pi\mathrm{i}}\int_{\mathbb{R}} \left[ \widehat{\mathbf{R}}_{\infty}^{\nabla}
+\frac{1}{\zeta} \widehat{\mathbf{R}}_{0}^{\nabla}
+\widehat{\mathbf{R}}^{-}(\zeta)\mathrm{e}^{\mathrm{i}\mathbf{\Delta}(\zeta)}\widehat{\mathbf{L}}(\zeta) \mathrm{e}^{-\mathrm{i}\mathbf{\Delta}(\zeta)} \right]_{2} \frac{\mathrm{d}\zeta}{\zeta-(q_{0}^{2}/\theta_{n})} \\
&+\sum_{m=1}^{M_{2}} \frac{\mathrm{i}\bar{F}_{m}\widehat{\rho}(\theta_{m}^{*})
\mathbf{R}_{1}^{-}(\theta_{m}^{*})}{(q_{0}^{2}/\theta_{n})-\theta_{m}^{*}}
+\sum_{m=1}^{M_{2}} \frac{\mathrm{i}\hat{F}_{m}\mathbf{R}_{3}^{+}(q_{0}^{2}/\theta_{m}^{*})}
{2\lambda(\theta_{m}^{*})[(q_{0}^{2}/\theta_{n})-(q_{0}^{2}/\theta_{m}^{*})]}.
\end{split}
\end{equation}
\end{theorem}

\begin{theorem}\label{thm:9}
Let $\widehat{\mathbf{R}}(z)$ be the solution of the RH problem in Theorem~\ref{thm:8}. Then the corresponding solution $\mathbf{q}(x,t)$ of the defocusing-defocusing coupled Hirota equations with orthogonal boundary conditions is reconstructed in the following manner:
\begin{equation}\label{4.47}
\begin{split}
q_{\hbar}(x,t)&=q_{-,\hbar}+\frac{\mathrm{i}}{2\pi}\int_{\mathbb{R}} \left[
\frac{q_{-,\hbar}}{\zeta} +\left[ \widehat{\mathbf{R}}^{-}(\zeta)\mathrm{e}^{\mathrm{i}\mathbf{\Delta}(\zeta)}\widehat{\mathbf{L}}(\zeta) \mathrm{e}^{-\mathrm{i}\mathbf{\Delta}(\zeta)} \right]_{(\hbar+1)1} \right] \mathrm{d}\zeta \\
&+\sum_{m=1}^{M_{1}}\mathrm{i}C_{m}\widehat{\rho}(z_{m})\mathbf{R}_{(\hbar+1)3}^{+}(z_{m})
-\sum_{m=1}^{M_{2}}F_{m}\mathbf{R}_{(\hbar+1)2}^{+}(\theta_{m}), \quad \hbar=1,2.
\end{split}
\end{equation}
\end{theorem}

\begin{remark}\label{rem:7}
By ingeniously constructing the new modified meromorphic matrices $\widehat{\mathbf{R}}^{\pm}(z)$, we have successfully achieved solutions to the defocusing-defocusing coupled Hirota equations under orthogonal boundary conditions. This method not only enriches the existing forms of solutions, but also provides a new perspective for exploring solutions under different boundary conditions. This innovative mathematical tool enables us to gain a deeper understanding and analysis of physical phenomena and mathematical models under various boundary conditions.
\end{remark}

\section{Discussion and final remarks}
\label{s:Discussion and final remarks}

We apply the IST tool to the defocusing-defocusing coupled Hirota equations with non-parallel boundary conditions and obtain some interesting results by constructing the matrix RH problem. Within the framework of the IST, we observe that there are significant differences between non-parallel and parallel cases when dealing with the defocusing-defocusing coupled Hirota equations. Particularly in the non-parallel case, the reflection coefficient cannot be completely zero which is determined by the nature of the direct problem and does not depend on the method used to address the inverse problem. Due to the continuous spectrum's inevitable presence, the solution always contains a dispersion component which is due to the contribution of the reflection coefficient to the IVP. This means that the problem itself does not support the existence of pure soliton solutions, regardless of whether we explore it through direct methods or by algebraically simplifying the RH problem.

From a practical standpoint, the findings of this paper are anticipated to provide valuable insights for interpreting recent studies in the fields of nonlinear optics~\cite{53} and Bose-Einstein condensations~\cite{25,26}. From a mathematical perspective, new orthogonal and non-orthogonal boundary solutions have been constructed using various mathematical tools, offering a new perspective on understanding and solving problems. Compared to solutions with parallel boundaries, these solutions exhibit greater diversity in form.

The forthcoming research objective is to investigate the existence and uniqueness of solution to the RH problem, as detailed in Theorems~\ref{thm:4},~\ref{thm:6} and~\ref{thm:8}. We can refer to the technical methods in references~\cite{10,54,A3,A4,A5,A6,A7,A8} for analyzing and exploring this issue. The focusing and defocusing NLS equations~\cite{55,56} with fully asymmetric NZBCs have been studied and correspondingly extended to the focusing and defocusing mKdV equations~\cite{57} in few-cycle pulses, the defocusing Hirota equation~\cite{50} and the defocusing Lakshmanan-Porsezian-Daniel equation~\cite{58}. A further issue to address is the case of the defocusing-defocusing coupled Hirota equations with fully asymmetric NZBCs. Here we note that although the defocusing-defocusing coupled Hirota equations have established appropriate IST for cases with symmetric NZBCs (also known as parallel NZBCs~\cite{23}), studying fully asymmetric NZBCs case is much more complex due to the analytical defects of some Jost eigenfunctions. Consequently, developing the IST for the defocusing-defocusing coupled Hirota equations with fully asymmetric NZBCs is expected to demand considerable further investigation. This involves combining the method proposed in this study with fully asymmetric NZBCs applied to the scalar defocusing NLS equation~\cite{56} to achieve deeper theoretical development and application expansion. Consider extending these research findings to a wider range of nonlinear evolution equations, hoping to discover new physical phenomena and mathematical structures.

Exploring the existence of spectral singularities is a fascinating theoretical problem, which refers to the analysis of the zero points of scattering coefficients within the continuous spectrum. In the scalar defocusing case, such zero points do not exist. In the focusing case, there are already multiple known examples of potentials capable of generating such a zero point. However, whether such spectral singularities exist in the defocusing-defocusing coupled Hirota equations remains an unresolved issue. This challenge persists even under the scenario of parallel NZBCs. Furthermore, in the case of non-parallel NZBCs, deriving the trace formula for orthogonal polarization vectors remains an unresolved issue. Finally, the lack of reflectionless potentials in the non-parallel NZBCs case suggests the possibility of generating dispersive shocks. Another interesting area of study is the behavior of multi-component Hirota equations in the case of non-parallel NZBCs. We anticipate that this discussion will inspire further research on several of the issues mentioned earlier.

\section*{Acknowledgement}
Zhang's and Ma's works were partially supported by the National Natural Science Foundation of China (Grant Nos. 11371326 and 12271488). Ma's work was also partially supported by the Ministry of Science and Technology of China (G2021016032L and G2023016011L).

\section*{Conflict of interests}
The authors declare that there is no conflict of interests regarding the research effort and the publication of this paper.

\section*{Data availability statements}
All data generated or analyzed during this study are included in this published article.

\section*{ORCID}
{\setlength{\parindent}{0cm}
Peng-Fei Han https://orcid.org/0000-0003-1164-5819 \\
Wen-Xiu Ma https://orcid.org/0000-0001-5309-1493 \\
Yi Zhang https://orcid.org/0000-0002-8483-4349}

\end{document}